\documentclass[11pt]{article}

\newcommand{\floor}[1]{\lfloor #1 \rfloor}
\newcommand{\ceil}[1]{\lceil #1 \rceil}
\usepackage{xcolor}
\usepackage{amsfonts}
\usepackage{amsthm}
\usepackage{amsmath}
\newtheorem{theorem}{Theorem}
\newtheorem{lemma}{Lemma}
\newtheorem{define}{Definition}
\usepackage{graphicx}
\usepackage[sort]{natbib}
\usepackage{algpseudocode}
\usepackage{algorithm}
\usepackage{multirow}
\usepackage{hyperref}
\usepackage{accents}
\usepackage{chngcntr}

\newcommand{\tSigma}{\underaccent{\tilde}{\Sigma}}
\newcommand{\tmathcalT}{\underaccent{\tilde}{\mathcal{T}}}
\newcommand{\tphi}{\underaccent{\tilde}{\phi}}

\newcommand{\ts}{\underaccent{\tilde}{s}}
\allowdisplaybreaks

\DeclareMathOperator{\argmax}{arg\,max}

\begin{document}

\title{Efficient Topic Model Estimation under Heavy-Tailed Document Lengths}
\author{Daniel Cirkovic\thanks{Department of Mathematical and Statistical Sciences, Marquette University, Milwaukee, WI 53233, U.S.A.} \and Tiandong Wang\thanks{Shanghai Center for Mathematical Sciences, Fudan University, Shanghai 200438, China.}}
\date{\today}
\maketitle

\begin{abstract}
Early inquiries into the statistical properties of natural language found that words tend to occur with power-law frequencies. This observation, closely associated with Zipf's law, has spurred many investigations into why this power-law pattern emerges with such regularity. Rarely, however, has this property of text been leveraged in statistical inference. In this paper, we demonstrate that the Latent Dirichlet Allocation (LDA) model can accommodate power-law word frequencies. In particular, when the document length distribution is regularly varying, the word frequency distribution admits a hierarchy of power laws across documents and topics. We further leverage this finding to develop an efficient tensor decomposition algorithm for estimating the topic matrix via the moments of normalized extreme word frequencies. Applying our algorithm to the twenty newsgroups corpus reveals that the extreme-value methodology exhibits robustness to certain choices made in the pre-processing of the data. This work furthers the recent interest in adapting machine learning methods to the study of multivariate extremes.
\end{abstract}

\medskip

\section{Introduction}

Folklore in the natural language processing literature states that words tend to occur with power-law frequencies. That is, the probability that word $w$ will occur with frequency greater than $j$ is approximately proportional to $j^{-\alpha}$ for some $\alpha > 0$ and $j$ sufficiently large. This observation is related to the classical rank--frequency formulation of Zipf's law, which states that the frequency of the $r$-th most common word in a corpus decays like a power of the rank $r$ \cite{zipf1935psycho}. Power-law patterns in text have generated many hypotheses for why they emerge with such regularity. Potential explanations include but are not limited to the preferential reuse of words and communicative efficiency \cite{yule2014statistical, simon1955class, zipf2016human}. Despite the plethora of hypotheses suggested to describe this phenomenon, the natural language processing and text mining literature lacks proposals leveraging the power-law property in the statistical inference of text data. In this paper, we make inroads into this domain as it pertains to the modeling of topics.

We demonstrate that the Latent Dirichlet Allocation (LDA) model, perhaps the most well-known topic model, can accommodate power-law word frequencies. In particular, we show that if the document lengths are power-law distributed, then so are the word frequencies. Moreover, the LDA model exhibits a hierarchy of power-laws across corpus-, document- and topic-level word frequencies with behavior that is ultimately governed by the parameters of the model. In contrast to previous work, we leverage this information to propose a topic estimation procedure that utilizes the higher-order moments of normalized extreme word frequencies. This strategy augments a previously proposed tensor decomposition algorithm for LDA \cite{anandkumar2012spectral, anandkumar2014tensor, anandkumar2015spectral}. Simulations demonstrate that by leveraging only the extreme word frequencies, the procedure greatly improves computational efficiency while also retaining estimation accuracy. Additionally, we show that a classic topic modeling benchmark, the twenty newsgroups corpus, exhibits these power-law properties and that the proposed methodology is robust to the often overlooked choices made in the pre-processing of the dataset. 

Given the prevalence of power-law word frequencies in text corpora, it is naturally desirable to develop generative models of text that reflect such behavior. If this pattern is not captured by the model, inferences of the parameters may be distorted \cite{goldwater2005interpolating, goldwater2011producing}. Previous work has imbued the LDA model with this capacity through the use of Pitman-Yor priors on parameters of the model \cite{sato2010topic, wang2025fine}. This class of priors improves the out-of-sample perplexity of the LDA model when applied to various corpora, but the power-law property is not explicitly leveraged to improve the efficiency of the fitting process. An exception is the spectral estimator of \cite{ke2024using}, whose SCORE-type normalization is designed to accommodate the severe word-frequency heterogeneity associated with Zipf's law. Similar methods have also been employed in the context of $n$-gram language models \cite{teh2006hierarchical}.

Our methodology connects to a recent line of work that adapts techniques from machine learning and high-dimensional statistics to multivariate extremes \cite[see][for a review]{engelke2021sparse}. One such adaptation is the extension of clustering methods such as spherical $k$-means and spectral clustering to extreme observations \cite{janssen2020k, medina2024spectral}. The spectral tensor decomposition of extreme word frequencies presented herein more closely follows the extensions of PCA and dependence measures to multivariate extreme-value settings \cite{butsch2025estimation, medina2025insights, reinbott2026principal}. Such extensions are often non-trivial since the notions of dependence used in extreme-value theory are distinct from measures traditionally used in statistics such as correlation \cite{das2017hidden}. In particular, we decompose the moments of the angular measure, quantities that are often used to describe the dependence structure between components of a random vector with a heavy tail \cite{cooley2019decompositions}.

The paper is organized as follows. We present the probabilistic topic model and LDA in Section \ref{sec:tm}. Section \ref{sec:mrv} reviews the technical details behind multivariate regular variation, the mathematical framework used to describe the power-law behavior of word frequencies. Subsequently, the multivariate regular variation of various word frequencies is derived. Section \ref{sec:est} introduces the estimation algorithm for the topic matrix based on extreme normalized word frequencies. We also provide the necessary background on tensor decompositions and identifiability assumptions. A simulation study evaluating the accuracy and computational efficiency of the proposed algorithm is provided in Section \ref{sec:sim}. The twenty newsgroups corpus is then analyzed in Section \ref{sec:20ng} and concluding remarks are provided in Section \ref{sec:conc}. Section \ref{sec:tech} contains proofs of the main theorems and supporting technical results.

\section{Topic Models}\label{sec:tm}

Consider a corpus of $D$ documents. Topic models represent each document $d = 1, \dots, D$ as a sequence of words $(X_{d1}, \dots, X_{dN_{d\cdot\cdot}})$ where $N_{d\cdot\cdot}$ is the length of document $d$. It is typically assumed that each word is drawn from a common vocabulary of $W$ words shared across all documents. We generically label the words using the integers $1, 2, \dots, W$. The probabilistic topic model, introduced by \cite{hofmann1999probabilistic} as probabilistic latent semantic indexing, assumes that the choice of $X_{di} \in \{1, \dots, W \}$ is driven by a latent topic $Z_{di} \in \{1, \dots, K\}$, where $K$ is the number of topics in the corpus. Each document is then characterized by a distribution over the latent topics. In other words, $(X_{di}, Z_{di})$ are independently generated according to
\begin{align*}
&Z_{di} \mid \theta_d \sim \text{Categorical}(\theta_{d}), \\
&X_{di} \mid Z_{di}, \phi_1, \dots, \phi_K \sim \text{Categorical}(\phi_{Z_{di}}).
\end{align*}
Here, $\theta_d \in \Delta^{K - 1}$ is a document specific probability distribution describing how often each of the $K$ topics appear in document $d$. Each topic vector $\phi_1, \dots, \phi_K \in \Delta^{W - 1}$  is a probability distribution over words describing how often each word appears in a topic. We assume each word appears in at least one topic with positive probability. It will be convenient to concatenate the topic vectors into a topic matrix 
\begin{align*}
\Phi = \begin{bmatrix} \phi_1 & \cdots & \phi_K\end{bmatrix}.
\end{align*}
Note that, given the model parameters, the independence of the $(X_{di}, Z_{di})$  is assumed within and across documents. Notably, the model disregards any sequential dependence of words within each document. 

Since topic models treat each document as a bag-of-words \cite[see][for an elaboration]{wallach2006topic}, it is useful to consider the probabilistic properties of various word frequencies.  From the model construction, it is easily seen that for any $w = 1, \dots, W$ and $k = 1, \dots, K$
\begin{align}
\label{eq:wordtopicjoint}
\mathbb{P}\left(X_{di} = w, Z_{di} = k \mid \theta_d, \Phi\right) = \theta_{dk}\phi_{kw}.
\end{align}
Let the latent count $N_{dkw} = \sum_{i = 1}^{N_{d\cdot\cdot}} 1_{\{X_{di} = w, Z_{di} = k\}}$ denote the frequency with which word $w$ appears as an instance of topic $k$ in document $d$.  Independence of the $(X_{di}, Z_{di})$ over $i$ gives that
\begin{align}
\label{eq:dkwcount}
\left(N_{dkw}\right)_{k, w} \mid N_{d\cdot\cdot}, \theta_d, \Phi \sim \text{Multinomial}\left(N_{d\cdot\cdot}, \left(\theta_{dk}\phi_{kw} \right)_{k, w}\right),
\end{align}
independently over $d = 1, \dots, D$. From this fact, multinomial aggregation gives the distribution of $N_{d\cdot w} = \sum_{k = 1}^K N_{dkw}$, the observable word counts within each document, regardless of topic. That is, independently for $d = 1,\dots, D$
\begin{align}
\label{eq:ddotwcount}
\left(N_{d\cdot 1}, \dots, N_{d\cdot W}\right) \mid N_{d\cdot\cdot}, \theta_d, \Phi \sim \text{Multinomial}\left(N_{d\cdot\cdot}, p_d \right),
\end{align}
where $p_d = \Phi \theta_d$. The goal of topic modeling is then to estimate the topic matrix $\Phi$ from the independent realizations $\left(N_{d\cdot 1}, \dots, N_{d\cdot W}\right)$, $d = 1,\dots, D$. We note that in \eqref{eq:ddotwcount} the topic matrix is not identifiable without further assumptions, a discussion we defer to Section \ref{sec:id}. From \eqref{eq:dkwcount}, the distributions of other word frequencies are easily ascertained. Defining $N_{d k \cdot} = \sum_{w = 1}^W N_{dkw}$ and $N_{\cdot k w} = \sum_{d = 1}^D N_{dkw}$, we have that independently over $d = 1,\dots, D$
\begin{align}
\label{eq:dkdotcount}
\left(N_{d1\cdot}, \dots, N_{dK\cdot}\right) \mid N_{d\cdot\cdot}, \theta_d \sim \text{Multinomial}\left(N_{d\cdot\cdot}, \theta_d\right),
\end{align}
and for $k = 1, \dots, K$
\begin{align}
\label{eq:dotkwcount}
\left(N_{\cdot k 1}, \dots, N_{\cdot k W}\right) \mid N_{\cdot k \cdot}, \phi_k \sim \text{Multinomial}\left(N_{\cdot k\cdot}, \phi_{k}\right).
\end{align}

The latent Dirichlet allocation (LDA) model of \cite{blei2003latent} further assumes that the parameters $\theta_1, \dots, \theta_D$ are independently generated according to
\begin{align}
\label{eq:dirichlet}
\theta_d \sim \text{Dirichlet}(\alpha_1, \dots, \alpha_K), \qquad d = 1, \dots, D.
\end{align}
This is a common assumption made in topic modeling \cite{anandkumar2012spectral, arora2012learning}. Hence, we will work with the LDA model unless otherwise specified.

\section{Regular variation of word frequencies}\label{sec:mrv}

\subsection{Background on regular variation}\label{sec:mrvback}

In Section \ref{sec:mrvcounts}, we will characterize the power-law nature of various word frequency statistics under the topic model. To formalize our findings, we provide some definitions related to multivariate regular variation (MRV). Whereas the original definition of MRV was constructed by employing the vague convergence of measures, $\mathbb{M}$-convergence permits greater flexibility when working with extremes in more than one dimension \cite{das2013living, das2017hidden, lindskog2014regularly}. 

Let $\mathbb{C}_0 \subset \mathbb{C} \subset \mathbb{R}^p_+$ be two closed cones. In our setting, we will usually have $\mathbb{C}_0 = \{ 0\}$ and $\mathbb{C} = \mathbb{R}^p_+$ where $0$ represents the zero vector in $\mathbb{R}^p_+$. Let $\mathbb{M}(\mathbb{C}\setminus \mathbb{C}_0)$ be the set of all Borel measures that are finite on sets bounded away from $\mathbb{C}_0$. Denote the set of all continuous, bounded, positive functions on $\mathbb{C}\setminus \mathbb{C}_0$ whose supports are bounded away from $\mathbb{C}_0$ via $\mathcal{C}(\mathbb{C}\setminus \mathbb{C}_0)$. Without loss of generality, we may assume that the functions are uniformly continuous \cite{lindskog2014regularly}. The definition of $\mathbb{M}$-convergence is provided below.

\begin{define}
\label{def:Mconv}
For $\mu_n, \mu \in \mathbb{M}(\mathbb{C}\setminus \mathbb{C}_0)$ we say that $\mu_n \rightarrow \mu$ in $\mathbb{M}(\mathbb{C}\setminus \mathbb{C}_0)$ if $\int f d\mu_n \rightarrow \int f d\mu$ for all $f \in \mathcal{C}(\mathbb{C}\setminus \mathbb{C}_0)$. 
\end{define}

With the notion of $\mathbb{M}$-convergence, we may define MRV for the distribution of a random vector in $\mathbb{C} = \mathbb{R}^p_+$. We specialize to the case where $\mathbb{C}_0 = \{ 0\}$. Below, we let $RV_c$ denote the class of all functions $U: \mathbb{R}_+ \rightarrow \mathbb{R}_+$ satisfying $\lim_{t \rightarrow \infty} U(tx)/U(t) = x^{c}$ for $c > 0, x > 0$. 

\begin{define}
\label{def:mrv}
The distribution of a random vector $Z \in \mathbb{R}^p_+$, $\mathbb{P}(Z \in \cdot)$, is multivariate regularly varying on $\mathbb{R}_+^p \setminus \{0\}$ with index $c > 0$ if there exists a scaling function $b(t) \in RV_{1/c}$ and a limit measure $\nu \in \mathbb{M}(\mathbb{R}^p_+\setminus \{0\}), \nu \not\equiv 0$ such that as $t \rightarrow \infty$
\begin{align*}
t\mathbb{P}\left(Z/b(t) \in \cdot \right) \rightarrow \nu(\cdot), \qquad \text{in } \mathbb{M}(\mathbb{R}^p_+ \setminus \{0\}).
\end{align*}
\end{define}   

If $\mathbb{P}(Z \in \cdot)$ satisfies Definition \ref{def:mrv}, we write $\mathbb{P}(Z \in \cdot) \in \text{MRV}(c, b(t), \nu, \mathbb{R}^p_+ \setminus \{0\})$. From a statistical  perspective, it is useful to restate multivariate regular variation in terms of polar coordinates since it reveals a semi-parametric structure for the transformed limit measure. That is, let $\|\cdot \|$ denote a norm on $\mathbb{R}_+^p$ and let $\aleph_0  = \{x \in \mathbb{R}^p_+: \|x \| = 1 \}$. Then Definition \ref{def:mrv} is equivalent to stating that for some appropriately chosen $b(t)$ 
\begin{align}
\label{eq:polar}
t\mathbb{P}\left[\left(\frac{Z}{\|Z\|}, \frac{\|Z\|}{b(t)} \right) \in \cdot \right] \rightarrow (S \times \nu_c)(\cdot), \quad \text{in } \mathbb{M}(\aleph_0  \times (0, \infty)),
\end{align}
where $\nu_c((x, \infty)) = x^{-c}$ for $x > 0$ and $S$ is a probability measure on $\aleph_0$ referred to as the angular measure.

In order to prove MRV of various word frequency statistics, we rely on a generalized Breiman's theorem \cite{wang2022asymptotic, janssen2025multivariate}. The result is stated in Section \ref{sec:gbt} for convenience.

\subsection{Multivariate regular variation of word frequencies}\label{sec:mrvcounts}

In this section, we present the multivariate regular variation of various word frequencies in the LDA model under the assumption that the document lengths are drawn from a power-law distribution. That is, suppose that $N_{1\cdot\cdot}, \dots, N_{D\cdot\cdot}$ are independently and identically distributed. We assume that for some $\iota > 0$ there exists a scaling function $b(t) \in RV_{1/\iota}$ such that as $t \rightarrow \infty$
\begin{align}
\label{eq:lengthrv}
t\mathbb{P}\left(N_{1\cdot\cdot} > b(t)x \right) \rightarrow x^{-\iota} = \nu_\iota((x, \infty)), \quad x > 0.
\end{align}
Under this assumption, we consider the MRV of the observed document-word counts,
\begin{align}
\label{eq:wfreqd}
\left(N_{d\cdot 1}, \dots, N_{d\cdot W}\right),
\end{align}
as well as the unobserved document-topic counts
\begin{align}
\label{eq:kfreqd}
\left(N_{d 1\cdot}, \dots, N_{d K\cdot}\right).
\end{align}
For a fixed topic $k$, we also consider MRV of the aggregate topic-word counts across all documents
\begin{align}
\label{eq:wfreq}
\left(N_{\cdot k 1}, \dots, N_{\cdot k W}\right).
\end{align}
The multivariate regular variation of these frequencies is presented in Theorem \ref{thm:freqmrv}. Since the finite sum of independent multivariate regularly varying random vectors is also multivariate regularly varying, a straightforward consequence of Theorem \ref{thm:freqmrv}(a) is that under the same conditions
\begin{align*}
\left(N_{\cdot\cdot 1}, \dots, N_{\cdot\cdot W}\right) = \sum_{d = 1}^D \left(N_{d\cdot 1}, \dots, N_{d\cdot W}\right) \in \text{MRV}(\iota, D^{1/\iota}b(t), \mu, \mathbb{R}^W_+ \setminus \{0\}).
\end{align*}
Hence, we obtain a hierarchy of power-laws. Whether word frequencies are observed within topics, within documents, or across an entire corpus, they all have the same power-law tail index $\iota$. Although the power-law behavior of word frequencies at these different levels of aggregation is not well-studied, the homogeneity of the tail index seems to be supported by the literature \cite{montemurro2001beyond, sato2010topic, williams2015text}. What differs across these hierarchies, however, is where large word frequencies concentrate. Theorem \ref{thm:freqmrv} indicates that high frequency words within topic $k$ are governed by $\phi_k$, while high frequency words within document $d$ modulate in accordance with $\Phi\theta_d$. While Theorem \ref{thm:freqmrv} is stated for the LDA model, we note that parts (a) and (b) hold for $\theta_d$ drawn from general distributions on the simplex. A variant of part (c) also holds assuming $\mathbb{P}(\theta_d \in \cdot)$ places non-negligible probability mass on all components of $\theta_d$. 

\begin{theorem}
\label{thm:freqmrv}
Assume $N_{1 \cdot \cdot}, \dots, N_{D \cdot \cdot}$ are drawn independently from a common regularly varying distribution satisfying \eqref{eq:lengthrv}. In addition, suppose $\theta_1,\dots,\theta_D$ are drawn independently according to \eqref{eq:dirichlet}.
\begin{itemize}
\item[(a)] The document-word counts are multivariate regularly varying with tail index $\iota$, i.e.
\begin{align*}
\left(N_{d\cdot 1}, \dots, N_{d\cdot W}\right) \in \text{MRV}(\iota, b(t), \mu, \mathbb{R}^W_+ \setminus \{0\}),
\end{align*}
where $\mu \in\mathbb{M}(\mathbb{R}_+^W \setminus \{0\})$ satisfies for any $f \in \mathcal{C}(\mathbb{R}_+^W \setminus \{0\})$
\begin{align*}
\mu(f) = \int_0^\infty \mathbb{E}\left[f(y\Phi \theta_d )\right] \nu_\iota(dy).
\end{align*}
Hence in $\mathbb{M}\left(\Delta^{W-1} \times (\mathbb{R}_+ \setminus \{0 \}) \right)$ 
\begin{align*}
t\mathbb{P}\left[ \left(\left(\frac{N_{d\cdot 1}}{N_{d\cdot\cdot}}, \dots, \frac{N_{d\cdot W}}{N_{d\cdot\cdot}}\right), \frac{N_{d\cdot\cdot}}{b(t)} \right) \in \cdot \right] \rightarrow \mathbb{P}\left(\Phi \theta_d \in \cdot \right) \times \nu_\iota(\cdot),
\end{align*}
as $t \rightarrow \infty$. 
\item[(b)] The document-topic counts are multivariate regularly varying with tail index $\iota$, i.e.
\begin{align*}
\left(N_{d1\cdot }, \dots, N_{dK\cdot}\right) \in \text{MRV}(\iota, b(t), \eta, \mathbb{R}^K_+ \setminus \{0\}),
\end{align*}
where $\eta \in\mathbb{M}(\mathbb{R}_+^K \setminus \{0\})$ satisfies for any $f \in \mathcal{C}(\mathbb{R}_+^K \setminus \{0\})$
\begin{align*}
\eta(f) = \int_0^\infty \mathbb{E}\left[f(y \theta_d )\right] \nu_\iota(dy).
\end{align*}
Hence in $\mathbb{M}\left(\Delta^{K-1} \times (\mathbb{R}_+ \setminus \{0 \}) \right)$ 
\begin{align*}
t\mathbb{P}\left[ \left(\left(\frac{N_{d1\cdot}}{N_{d\cdot\cdot}}, \dots, \frac{N_{dK \cdot}}{N_{d\cdot\cdot}}\right), \frac{N_{d\cdot\cdot}}{b(t)} \right) \in \cdot \right] \rightarrow \mathbb{P}\left(\theta_d \in \cdot \right) \times \nu_\iota(\cdot),
\end{align*}
as $t \rightarrow \infty$. 

\item[(c)] The topic-word counts are multivariate regularly varying with tail index $\iota$, i.e.
\begin{align*}
\left(N_{\cdot k 1}, \dots, N_{\cdot k W}\right) \in \text{MRV}(\iota, \tilde{b}_k(t), \upsilon, \mathbb{R}^W_+ \setminus \{0\}),
\end{align*}
where $\upsilon \in\mathbb{M}(\mathbb{R}_+^W \setminus \{0\})$ satisfies for any $f \in \mathcal{C}(\mathbb{R}_+^W \setminus \{0\})$
\begin{align*}
\upsilon(f) = \int_0^\infty f(y\phi_k) \nu_\iota(dy).
\end{align*}
Here, 
\begin{align*}
\tilde{b}_k(t) = \left(D \frac{B(\alpha_k + \iota, \alpha_0 - \alpha_k)}{B(\alpha_k, \alpha_0 - \alpha_k)} \right)^{1/\iota} b(t).
\end{align*}
Hence in $\mathbb{M}\left(\Delta^{W-1} \times (\mathbb{R}_+ \setminus \{0 \}) \right)$ 
\begin{align*}
t\mathbb{P}\left[ \left(\left(\frac{N_{\cdot k 1}}{N_{\cdot k \cdot}}, \dots, \frac{N_{\cdot k W}}{N_{\cdot k \cdot}}\right), \frac{N_{\cdot k \cdot}}{\tilde{b}_k(t)} \right) \in \cdot \right] \rightarrow \epsilon_{\phi_k}(\cdot) \times \nu_\iota(\cdot),
\end{align*}
as $t \rightarrow \infty$, where $\epsilon_{\phi_k}$ denotes the Dirac measure at $\phi_k$.
\end{itemize}
\end{theorem}

The proof of Theorem \ref{thm:freqmrv} is provided in Section \ref{sec:mrvproof}. Note that Theorem~\ref{thm:freqmrv} identifies document length as the source of the multivariate regular variation of the document-word counts. The topic proportions and word occurrences are not assumed to be independent. Indeed, conditional on $\theta_d$, words are generated according to the LDA model with probability vector $p_d=\Phi\theta_d$. Moreover, for finite document lengths, the normalized counts
\[
\Omega_d =
\left(
\frac{N_{d\cdot1}}{N_{d\cdot\cdot}},\ldots,
\frac{N_{d\cdot W}}{N_{d\cdot\cdot}}
\right)
\]
depend on $N_{d\cdot\cdot}$. The product limit in the theorem shows that this dependence vanishes asymptotically so that among increasingly long documents, $\Omega_d$ converges to $\Phi\theta_d$, while $N_{d\cdot\cdot}$ governs the radial magnitude.

Word co-occurrences, or the number of times two words occur together in the same document, play a seminal role in topic model estimation. In fact, \cite{arora2012learning} employs nonnegative matrix factorization on the word co-occurrence matrix in their estimation of the topic matrix in the LDA model. We extend the regular variation properties of the word frequencies to the word co-occurrences. Note that the number of times two distinct words $u$ and $v$ co-occur in document $d$ is given by $N_{d \cdot u}N_{d\cdot v}$. Hence, one way regular variation of the co-occurrences for distinct words can be achieved is through the closure properties for products of the components of a regularly varying random vector \cite[see Proposition 7.6 of][for example]{resnick2007heavy}. We note that Lemma \ref{lem:cooccur} extends to word triplets, quadruplets and so on. For example, the number of times three distinct words $u, v$ and $w$ co-occur is regularly varying with index $\iota/3$. The proof of Lemma \ref{lem:cooccur} is provided in Section \ref{sec:coproof}.

\begin{lemma}
\label{lem:cooccur}
Consider the setting of Theorem \ref{thm:freqmrv}. Fix any $u = 1, \dots, W$. Then in $\mathbb{M}\left(\mathbb{R}_+^{W}  \setminus \{0 \} \right)$ 
\begin{align*}
\left(N_{d\cdot u}N_{d\cdot 1}, \dots, N_{d\cdot u}N_{d\cdot W}\right) \in \text{MRV}(\iota/2, b^2(t), \zeta, \mathbb{R}^W_+ \setminus \{0\}),
\end{align*}
where $\zeta \in\mathbb{M}(\mathbb{R}_+^W \setminus \{0\})$ satisfies for any $f \in \mathcal{C}(\mathbb{R}_+^W \setminus \{0\})$
\begin{align*}
\zeta(f) = \int_0^\infty \mathbb{E}\left[f(y p_{du} p_d)\right] \nu_{\iota/2}(dy).
\end{align*}
\end{lemma}

\section{Extreme-value estimation of the topic matrix}\label{sec:est}

In this section, we develop an estimation procedure for the topic matrix based on the probabilistic results presented in Theorem \ref{thm:freqmrv}. In particular, we establish a method of moments procedure that utilizes the moments of the angular measure $\mathbb{P}\left(\Phi \theta_d \in \cdot \right)$ in part (a) of Theorem \ref{thm:freqmrv}. The estimation scheme is motivated by a tensor decomposition method developed in \cite{anandkumar2012spectral, anandkumar2015spectral}, but leverages the power-law nature of word frequencies to derive simpler estimators and mitigate computational burdens.

\subsection{Background on orthogonal tensor decomposition}\label{sec:tensor}

Consider a generic cubical tensor $T \in \mathbb{R}^{n \times n \times n}$. In order to adequately describe the tensor decomposition algorithms, we must define a notion of tensor-matrix multiplication via multiplication along modes \cite{kolda2009tensor, anandkumar2014tensor}. That is, for matrices $A_i \in \mathbb{R}^{n \times p_i}$, let $T(A_1, A_2, A_3) \in \mathbb{R}^{p_1 \times p_2 \times p_3}$ be the tensor with $(a_1, a_2, a_3)$-th entry
\begin{align*}
[T(A_1, A_2, A_3)]_{a_1, a_2, a_3} = \sum_{i = 1}^n \sum_{j = 1}^n \sum_{l = 1}^n T_{ijl}[A_1]_{i, a_1}[A_2]_{j, a_2}[A_3]_{l, a_3}.
\end{align*}

Our algorithms will only be concerned with symmetric cubical tensors, or those that satisfy $T_{ijl} = T_{\tau(i, j, l)}$ for any permutation $\tau$ of the arguments. Following \cite{comon2008symmetric}, we define the notion of symmetric rank for $T$. Here, $\otimes$ denotes the usual tensor product.
\begin{define}
\label{def:srank}
The symmetric rank of a symmetric cubical tensor $T$ is the smallest non-negative integer $R$ such that $T = \sum_{r = 1}^R u_r \otimes u_r \otimes u_r$ for some $u_1, \dots, u_R \in \mathbb{R}^n$.
\end{define} 
More generally, the decomposition of a generic tensor into the sum of rank-one terms is called a canonical polyadic (CP) decomposition \cite[see][]{kolda2009tensor}. Following \cite{anandkumar2014tensor}, we now define what it means for a symmetric cubical tensor to admit an orthogonal decomposition.
\begin{define}
A symmetric cubical tensor $T \in \mathbb{R}^{n \times n \times n}$ has an orthogonal decomposition if there exists a set of orthonormal vectors $\{v_1, \dots, v_R\}$ and strictly positive scalars $\lambda_1, \dots, \lambda_R > 0$ such that $T = \sum_{r = 1}^R \lambda_r (v_r \otimes v_r \otimes v_r)$.
\end{define}
A symmetric cubical tensor with orthogonal decomposition $T = \sum_{r = 1}^R \lambda_r (v_r \otimes v_r \otimes v_r)$ has symmetric rank $R$ \cite{kolda2015symmetric}. An orthogonal decomposition is not guaranteed to exist for every symmetric tensor. However, if an orthogonal decomposition does exist, it is unique up to permutation \cite{anandkumar2014tensor}. 

In order to compute the decomposition, \cite{anandkumar2014tensor} note that the pairs $(v_r, \lambda_r)$, $r = 1, \dots, R$, form a set of tensor eigenvector/eigenvalue pairs since they satisfy a third-order generalization of the matrix notion:
\begin{align}
\label{eq:Teigen}
T(I, u, u) = \lambda u,
\end{align}
for some unit vector $u$ and $\lambda \in \mathbb{R}$. However, $\{v_1, \dots, v_R\}$ are not the only unit vectors satisfying \eqref{eq:Teigen}. Fortunately, these additional vectors can be discarded by only considering the set of vectors that are recoverable through tensor power iterations.
\begin{define}
A unit vector $u \in \mathbb{R}^n$ is a robust eigenvector of $T \in \mathbb{R}^{n \times n \times n}$ if there exists an $\varepsilon > 0$ such that for all $\theta \in \{ s \in \mathbb{R}^n : \|u - s \|_2 \leq \varepsilon \}$, repeated iteration of the map
\begin{align*}
\bar{\theta} \mapsto \frac{T(I, \bar{\theta}, \bar{\theta})}{\|T(I, \bar{\theta}, \bar{\theta}) \|_2},
\end{align*}
starting from $\theta$ converges to $u$.
\end{define}

In particular, Theorem 4.1 of \cite{anandkumar2014tensor} asserts that the set of robust eigenvectors of $T$ is exactly equal to $\{v_1, \dots, v_R\}$. In practice, rarely does one observe an orthogonally decomposable tensor $T$. Rather, one observes a noisy estimate $\hat{T} = T + E$ where $E$ is an error tensor with small operator norm (see \eqref{eq:op} for a definition). In order to account for the perturbation from this error tensor, \cite{anandkumar2014tensor} provide the robust tensor power method which we record in Algorithm \ref{alg:tensorpower}. Given a starting vector $u_0$, tensor power iterations of an orthogonally decomposable $T$ converge to the component $v_i$ which maximizes $|\lambda_i v_i'u_0|$ \citep[see Lemma 5.1 of][]{anandkumar2014tensor}. Given a perturbed tensor $\hat{T}$, however, Algorithm \ref{alg:tensorpower} uses $L$ trials to ensure that a starting vector $u_0$ that sufficiently separates the $|\lambda_i v_i'u_0|$ values relative to the error tensor is eventually found with high probability. Naturally, Algorithm \ref{alg:tensorpower} must be repeated multiple times with deflation in order to compute all robust eigenvectors.

\begin{algorithm}
\caption{Robust tensor power method \cite{anandkumar2014tensor}.}
\label{alg:tensorpower}
\begin{algorithmic}
\Require Symmetric tensor $T \in \mathbb{R}^{K \times K \times K}$, \# random restarts $L$, \# power iterations $N$.
\Ensure Estimated eigenvalue $\lambda$ with corresponding eigenvector $v$, deflated tensor $T - \lambda (v \otimes v \otimes v)$.
\For{$l = 1$ to $L$}
\Statex Sample $v^{(l)}_0$ from the uniform distribution on the unit sphere in $\mathbb{R}^K$.
\For{$t = 1$ to $N$}
\begin{align*}
v^{(l)}_t = \frac{T(I, v^{(l)}_{t-1}, v^{(l)}_{t-1})}{\| T(I, v^{(l)}_{t-1}, v^{(l)}_{t-1}) \|_2}
\end{align*}
\EndFor
\EndFor
\Statex Compute $\ell = \argmax_{l = 1, \dots, L} T(v^{(l)}_{N}, v^{(l)}_{N}, v^{(l)}_{N})$.
\For{$t =  N + 1$ to $2N$}
\begin{align*}
v^{(\ell)}_t = \frac{T(I, v^{(\ell)}_{t-1}, v^{(\ell)}_{t-1})}{\| T(I, v^{(\ell)}_{t-1}, v^{(\ell)}_{t-1}) \|_2}
\end{align*}
\EndFor
\Statex Set $v = v^{(\ell)}_{2N}$ and $\lambda = T(v, v, v)$.
\end{algorithmic}
\end{algorithm}

In order to analyze the accuracy of Algorithm \ref{alg:tensorpower} in our setting, we introduce some matrix and tensor norms. For a matrix $A \in \mathbb{R}^{m \times n}$, we let $\|A\|_F$ and $\|A\|$ denote the Frobenius and operator norms, respectively. For a symmetric tensor $T \in \mathbb{R}^{n \times n \times n}$, these notions are defined analogously. The Frobenius norm of $T$ is given by
\begin{align}
\label{eq:frob}
\|T\|_F = \sqrt{\sum_{i = 1}^n \sum_{j = 1}^n \sum_{l = 1}^n T^2_{ijl}},
\end{align}
while the operator norm is given by
\begin{align}
\label{eq:op}
\|T\| = \sup_{\|x\|_2 = 1} |T(x, x, x)|.
\end{align}
As in the matrix setting, $\|T\| \leq \| T \|_F$. See \cite{anandkumar2014tensor, kolda2009tensor, zhang2012best} for more details on tensor norms.

\subsection{Identifiability and estimation of topics}\label{sec:id}

A variety of methods have been proposed to estimate the topic model of \eqref{eq:ddotwcount}. In this section, we briefly discuss the assumptions that these procedures make in order to guarantee identifiability of the topic model parameters. Recall the definition $p_d = \Phi \theta_d$ for $d = 1, \dots, D$ and consider the matrices
\begin{align*}
P \equiv \begin{bmatrix} p_1 & \cdots & p_D \end{bmatrix} = \Phi \begin{bmatrix} \theta_1 & \cdots & \theta_D \end{bmatrix} \equiv \Phi \Theta.
\end{align*}
One class of methods used to estimate the topic matrix constructs an estimate of $P$ using the empirical word frequencies and then use non-negative matrix factorization algorithms to recover $\Phi$ \cite{lee1999learning}. Similar approaches employ the word co-occurrence matrix instead \cite{arora2012learning}. In order to ensure that $\Phi$ and $\Theta$ are identifiable, this class of methods assume that each topic has at least one word, called an anchor word, that occurs with strictly positive probability in that topic and probability zero in the other topics  \cite{arora2012learning, bing2020fast}. Under the anchor word assumption, \cite{ke2024using} develop a fast SVD-based procedure, Topic-SCORE, whose normalization is tailored to the severe word-frequency heterogeneity implied by Zipf's law, and establish rates of convergence for the estimated topic matrix. 

The class of methods that we consider estimates the topic matrix in LDA through tensor decomposition of the third-order moments \cite{anandkumar2012spectral, anandkumar2014tensor, anandkumar2015spectral, decarolis2020end}. Tensor decomposition methods require that the columns of $\Phi$ are linearly independent, which is naturally less stringent than assuming the existence of anchor words. Given that these methods were developed for the LDA model, they are known to be sensitive to the Dirichlet specification of the topic proportions in \eqref{eq:dirichlet}. However, more recent developments have extended the application of tensor methods to topic proportion distributions with more realistic dependence structures between topics \cite{arabshahi2017spectral}. Henceforth, we assume linear independence of the topics $\phi_1, \dots, \phi_K$.

\subsection{Tensor decomposition of the angular measure moments}

We now motivate our extreme-value based estimation procedure under the assumption that $\Phi$ has full column rank. Consider the convergence of the angular component of the word frequencies in part (a) of Theorem \ref{thm:freqmrv}. This statement implies that for any $d = 1, \dots, D$
\begin{align}
\label{eq:angledist}
\left(\frac{N_{d\cdot 1}}{N_{d\cdot\cdot}}, \dots, \frac{N_{d\cdot W}}{N_{d\cdot\cdot}}\right) \big| \left\lbrace N_{d\cdot\cdot} > t \right\rbrace \Rightarrow \Phi \theta_d,
\end{align} 
as $t \rightarrow \infty$ where $\Rightarrow$ denotes weak convergence. Hence, for sufficiently large document lengths, we may model the word-document relative frequencies as independent draws having identical distributions to that of $\Phi \theta_d$. In order to estimate the topic matrix, we employ a method-of-moments scheme inspired by \cite{anandkumar2012spectral, anandkumar2015spectral}. 

In order to motivate the procedure, assume that we observe draws $p_d = \Phi \theta_d$ where $\theta_d$ is independently drawn from a  $\text{Dirichlet}(\alpha_1, \dots, \alpha_K)$ distribution for $d = 1, \dots, D$. The method of moments aims to match the model parameters to the moments of $p_d$, which are readily estimated by their sample counterparts. After appropriate adjustments, the raw lower-order moments of $p_d$ exhibit a useful decomposition that can be exploited to recover the topics \cite{anandkumar2012spectral}. Below we define the adjusted moments, all of which are functions of the moments of $p_d$:
\begin{align*}
M =& \mathbb{E}\left[p_d\right] = \frac{1}{\alpha_0} \Phi \alpha, \\
\Sigma =& \mathbb{E}\left[ p_d \otimes p_d \right] - \frac{\alpha_0}{\alpha_0 + 1} M \otimes M, \\
\mathcal{T} =& \mathbb{E}\left[  p_d \otimes p_d \otimes p_d \right] - \frac{\alpha_0}{\alpha_0 + 2} \left( \mathbb{E}\left[  p_d \otimes p_d \otimes M \right] + \mathbb{E}\left[  p_d \otimes M \otimes p_d \right] + \mathbb{E}\left[ M \otimes p_d \otimes p_d \right] \right) \\
&+ \frac{2\alpha_0^2}{(\alpha_0 + 2)(\alpha_0 + 1)} M \otimes M \otimes M.
\end{align*}
The following lemma from \cite{anandkumar2012spectral} reveals that $\Sigma$ and $\mathcal{T}$ admit CP decompositions where the factors of the decompositions are functions of the topic vectors $\phi_1, \dots, \phi_K$.  Note that while the decomposition of $\Sigma$ in Lemma \ref{lem:moments} is suggestive, knowledge of $\Sigma$ alone does not uniquely identify the parameters $\alpha$ and $\Phi$. There are instances where different choices of $(\alpha, \Phi)$ give rise to the same $\Sigma$. Hence, one must appeal to the third-order moments. In Section \ref{sec:ntopic}, however, we do employ the decomposition of $\Sigma$ to identify the number of topics. For completeness, we provide the proof of Lemma \ref{lem:moments} in Section \ref{sec:mproof}.

\begin{lemma}
\label{lem:moments}
Under the LDA model, the adjusted moments may be expressed as
\begin{align*}
\Sigma =& \frac{1}{\alpha_0(\alpha_0 + 1)} \sum_{k = 1}^K \alpha_k (\phi_k \otimes \phi_k), \\
\mathcal{T} =& \frac{2}{\alpha_0(\alpha_0 + 1)(\alpha_0 + 2)} \sum_{k = 1}^K \alpha_k (\phi_k \otimes \phi_k \otimes \phi_k).
\end{align*}
\end{lemma}

We now discuss how the third-order moments can be used to identify the topic matrix $\Phi$. Lemma \ref{lem:moments} indicates that $\mathcal{T}$ has a symmetric rank $K$ decomposition \citep[see Lemma 5.1 of][]{comon2008symmetric}. Given that $\Phi$ has full column rank, Kruskal's result implies that this CP decomposition is unique up to permutation and scaling, the latter of which can be resolved by  the simplex constraint on the topics \cite[see Section 3.2 of][for more on uniqueness of tensor decompositions]{kolda2009tensor}. Hence, given $\alpha_0$, $(\alpha, \Phi)$ are identifiable from $\mathcal{T}$ up to permutation. Given an estimate of $\mathcal{T}$, one may apply tensor decomposition algorithms such as alternating least squares (ALS) to estimate the model parameters \cite{harshman1970foundations, kolda2009tensor}. However, the non-convexity of the ALS objective and high dimension of $\mathcal{T}$ present challenges to such strategies. 

In order to facilitate computation of the tensor decomposition, we appeal to the strategy of \cite{anandkumar2012spectral, anandkumar2015spectral} and transform $\mathcal{T}$ to an orthogonally decomposable tensor. This allows us to apply the methodology discussed in Section \ref{sec:tensor} to recover the topics. In order to conveniently describe the procedure, we follow \cite{huang2017spectral} and instead work with the scaled quantities
\begin{align*}
\tSigma =& \alpha_0(\alpha_0 + 1) \Sigma = \sum_{k = 1}^K \alpha_k (\phi_k \otimes \phi_k) = \sum_{k = 1}^K \tphi_k \otimes \tphi_k, \\
\tmathcalT =& \frac{\alpha_0(\alpha_0 + 1)(\alpha_0 + 2)}{2} \mathcal{T} = \sum_{k = 1}^K \alpha^{-1/2}_k (\tphi_k \otimes \tphi_k \otimes \tphi_k),
\end{align*}
where $\tphi_k = \sqrt{\alpha_k}\phi_k$ for $k = 1, \dots, K$. From $\tSigma$, \cite{anandkumar2012spectral, anandkumar2015spectral} use the eigendecomposition $\tSigma = U\Lambda U'$ where $U \in \mathbb{R}^{W \times K}$, $\Lambda \in \mathbb{R}^{K\times K}$ to construct a whitening matrix $\mathcal{W} = U \Lambda^{-1/2}$. Note that whitened vectors $\varphi_k = \mathcal{W}'\tphi_k \in \mathbb{R}^K$ are orthonormal since
\begin{align*}
I = \mathcal{W}'\tSigma\mathcal{W} = \sum_{k = 1}^K \varphi_k \otimes \varphi_k.
\end{align*}  
We note that applying the whitening matrix to the weighted topics orthonormalizes them within the topic subspace. This reduces the dimension of the topic vectors from $W$ to $K$ without losing information since they can be recovered using $U\Lambda^{1/2}$. Applying the same whitening matrix to each mode of the tensor $\tmathcalT$ gives that
\begin{align*}
\tmathcalT(\mathcal{W}, \mathcal{W}, \mathcal{W}) = \sum_{k = 1}^K \alpha^{-1/2}_k(\varphi_k \otimes \varphi_k \otimes \varphi_k). 
\end{align*}
In other words, the whitened $\tmathcalT$ admits a symmetric orthogonal decomposition and the pairs $(\varphi_k, \alpha^{-1/2}_k)$ can be recovered via the tensor power iterations discussed in Section \ref{sec:tensor}. From there, $\phi_1, \dots, \phi_K$ can be obtained via transformation by $U\Lambda^{1/2}$ and rescaling. 

Given estimates of the moments $M$, $\Sigma$ and $\mathcal{T}$, Algorithm \ref{alg:specLDA} presents the estimation procedure for the topic matrix. We discuss how the power-law nature of the word frequencies can be used to derive computationally efficient estimators of these quantities in Section \ref{sec:moment_est}. We allow the number of tensor power iterations, $N_D \equiv N$, to depend on the number of documents in Algorithm \ref{alg:specLDA} since theoretical analysis requires $N_D$ to diverge as $D \rightarrow \infty$. Note that since the algorithm inputs estimates of $M$, $\Sigma$ and $\mathcal{T}$ rather than the true values, the topic vector estimates output by the algorithm described previously may not be probability vectors. Hence, following \cite{huang2017spectral}, the last step of Algorithm \ref{alg:specLDA} projects the output vectors to $\Delta^{W - 1}$ using the Algorithm in Figure 1 of \cite{duchi2008efficient}. We further note that while $\hat{\mathcal{T}}_D$ is listed as an input in Algorithm \ref{alg:specLDA}, the $W \times W \times W$ tensor need not be explicitly formed. Rather, only an estimate of the $K \times K \times K$ tensor $\mathcal{T}(\mathcal{W}, \mathcal{W}, \mathcal{W})$ needs to be formed in order to execute the algorithm.

\begin{algorithm}
\caption{Estimation of topic matrix from tensor decomposition of adjusted moments \cite{anandkumar2012spectral, anandkumar2015spectral, huang2017spectral}}
\label{alg:specLDA}
\begin{algorithmic}
\Require Estimates $\hat{\Sigma}_D$ and $\hat{\mathcal{T}}_D$, \# of topics $K$, Dirichlet concentration $\alpha_0$, \# random restarts $L$, \# power iterations $N_D$.
\Ensure Parameter estimates $(\hat{\alpha}_D, \hat{\Phi}_D)$. 
\Statex 1. For convenience, transform $\hat{\tSigma}_D = \alpha_0(\alpha_0 + 1) \hat{\Sigma}_D$, $\hat{\tmathcalT}_D = \frac{\alpha_0(\alpha_0 + 1)(\alpha_0 + 2)}{2} \hat{\mathcal{T}}_D$.
\Statex 2. Eigendecompose $\hat{\tSigma}_D \approx \hat{U}_D \hat{\Lambda}_D \hat{U}_D'$  where $\hat{U}_D \in \mathbb{R}^{W \times K}$, $\hat{\Lambda}_D \in \mathbb{R}^{K\times K}$ and set $\hat{\mathcal{W}}_D = \hat{U}_D \hat{\Lambda}_D^{-1/2}$. 
\Statex 3. Construct $\hat{\tmathcalT}_D(\hat{\mathcal{W}}_D, \hat{\mathcal{W}}_D, \hat{\mathcal{W}}_D)$ and perform the tensor power method (Algorithm \ref{alg:tensorpower}) $K$ times with parameters $L$ and $N_D$. Recover tensor eigenvectors $\hat{\varphi}_{D, 1}, \dots, \hat{\varphi}_{D, K}$ with tensor eigenvalues $\hat{\lambda}_{D, 1}, \dots, \hat{\lambda}_{D, K}$.
\Statex 4. Compute $\hat{\beta}_{D,k} = \hat{\lambda}_{D,k} \hat{U}_D\hat{\Lambda}_D^{1/2} \hat{\varphi}_{D,k}$ and $\hat{\alpha}_{D,k} = \hat{\lambda}_{D,k}^{-2}$ for $k = 1, \dots, K$.
\Statex 5. Map $\hat{\beta}_{D,k}$ to $\hat{\phi}_{D,k}$ by projecting it to $\Delta^{W - 1}$ using \cite{duchi2008efficient}.
\end{algorithmic}
\end{algorithm}

\subsection{Estimation of angular measure moments}\label{sec:moment_est}

Given estimates of the adjusted moments $M, \Sigma$ and $\mathcal{T}$, Algorithm \ref{alg:specLDA} can be applied to estimate the topic matrix $\Phi$.  In \cite{anandkumar2012spectral, anandkumar2015spectral, huang2015online, decarolis2020end}, the raw sample moments of the word frequencies are used to estimate the population moments. Such moment estimators must account for the dependence between the components of multinomial vectors, which complicates their computation and implementation. Using the convergence in part (a) of Theorem \ref{thm:freqmrv}, our extreme-value theory motivated estimators are able to side-step the multinomial noise and provide comparatively simpler moment estimates. 

In order to introduce the proposed moment estimators, define the sample  radii $R_d = N_{d\cdot\cdot}$ and sample angles
\begin{align*}
\Omega_d = \left(\omega_{d1}, \dots, \omega_{dW}\right) = \left(\frac{N_{d\cdot 1}}{N_{d \cdot \cdot}}, \dots, \frac{N_{d\cdot W}}{N_{d \cdot \cdot}} \right), 
\end{align*}
for $d = 1, \dots, D$. Let $S(\cdot) = \mathbb{P}(\Phi \theta_d \in \cdot)$ denote the angular measure in (a) of Theorem \ref{thm:freqmrv}. The population moments that form $M$, $\Sigma$ and $\mathcal{T}$ may be expressed as
\begin{align}
\label{eq:p1}
\mathbb{E}[p_d] =& \int_{\Delta^{W-1}} p_d S(dp_d), \\
\label{eq:p2}
\mathbb{E}[p_d \otimes p_d] =& \int_{\Delta^{W-1}} p_d \otimes p_d S(dp_d), \\
\label{eq:p3}
\mathbb{E}[p_d \otimes p_d \otimes p_d] =& \int_{\Delta^{W-1}} p_d \otimes p_d \otimes p_d S(dp_d). 
\end{align}
Let $R_{(1)} \geq \dots \geq R_{(D)}$ denote the decreasing radii order statistics. Following Chapter 9.2 of \cite{resnick2007heavy}, a consistent estimator of $S$ is given by
\begin{align*}
\hat{S}_D(\cdot) = \frac{1}{k_D} \sum_{d = 1}^D 1_{\left\lbrace R_d \geq R_{(k_D)}, \Omega_d \in \cdot \right\rbrace}.
\end{align*} 
In other words, considering $\hat{S}_D$ as a random measure in $\mathbb{M}(\Delta^{W-1})$, we have that $\hat{S}_D \Rightarrow S$ in $\mathbb{M}(\Delta^{W-1})$ as $D \rightarrow \infty$, $k_D \rightarrow \infty$ and $k_D/D \rightarrow 0$ (see Chapter 5.3 of \cite{resnick2024art}, Chapter 9.2 of \cite{resnick2007heavy} or Proposition 6.2 of \cite{das2013living}). Via the estimator $\hat{S}_D$, the following plug-in estimates of the moments may be derived:
\begin{align}
\label{eq:s1}
\hat{\mathbb{E}}_D[p_d] =& \frac{1}{k_D} \sum_{d = 1}^D \Omega_d 1_{\left\lbrace R_d \geq R_{(k_D)}\right\rbrace}, \\
\label{eq:s2}
\hat{\mathbb{E}}_D[p_d \otimes p_d] =& \frac{1}{k_D} \sum_{d = 1}^D (\Omega_d \otimes \Omega_d) 1_{\left\lbrace R_d \geq R_{(k_D)}\right\rbrace}, \\
\label{eq:s3}
\hat{\mathbb{E}}_D[p_d \otimes p_d \otimes p_d] =& \frac{1}{k_D} \sum_{d = 1}^D (\Omega_d \otimes \Omega_d \otimes \Omega_d) 1_{\left\lbrace R_d \geq R_{(k_D)}\right\rbrace}. 
\end{align}
Since the integrands in \eqref{eq:p1}, \eqref{eq:p2} and \eqref{eq:p3} are bounded and continuous, the plug-in estimates are element-wise consistent as $D \rightarrow \infty$, $k_D \rightarrow \infty$ and $k_D/D \rightarrow 0$. Likewise, consistent estimates of $M, \Sigma$ and $\mathcal{T}$ (or similarly, $\tSigma$ and $\tmathcalT$), may be easily obtained by replacing \eqref{eq:p1}, \eqref{eq:p2} and \eqref{eq:p3} by \eqref{eq:s1}, \eqref{eq:s2} and \eqref{eq:s3}, respectively. Let $\hat{M}_D$, $\hat{\Sigma}_D$ and $\hat{\mathcal{T}}_D$ denote these plug-in estimators.

In comparison to the estimators provided by \cite{anandkumar2012spectral, anandkumar2015spectral, huang2015online, decarolis2020end}, the estimators $\hat{M}_D$, $\hat{\Sigma}_D$ and $\hat{\mathcal{T}}_D$ have a few main advantages. First, extreme-value theory allows us to circumvent the dependence between the multinomial components in \eqref{eq:ddotwcount} by only using the documents with sufficiently long length. This results in simpler estimators, especially in regards to the third moment tensor $\mathcal{T}$ \citep[see equation (3) of][for example]{huang2015online}. In addition, the extreme-value theory estimators employ $k_D$ documents rather than $D$. Hence, in large-scale applications, the estimators are faster to compute. Further, the extreme-value theory based estimators yield algorithmic benefits as well. For example, consider the estimator
\begin{align*}
\hat{\Sigma}_D = \hat{\mathbb{E}}_D\left[ p_d \otimes p_d \right] - \frac{\alpha_0}{\alpha_0 + 1} \hat{\mathbb{E}}_D\left[ p_d \right] \otimes \hat{\mathbb{E}}_D\left[ p_d \right].
\end{align*}
Step 2 of Algorithm \ref{alg:specLDA} computes the eigendecomposition of $\alpha_0(\alpha_0 + 1)\hat{\Sigma}_D$. We note that by using Lanczos style algorithms, $\hat{\Sigma}_D$ need not be formed to compute the eigendecomposition. Rather, only products of the form 
\begin{align*}
\hat{\Sigma}_D x = \frac{1}{k_D} \sum_{d = 1}^D \Omega_d (\Omega_d'x) 1_{\left\lbrace R_d \geq R_{(k_D)}\right\rbrace} - \frac{\alpha_0}{\alpha_0 + 1} \hat{\mathbb{E}}_D\left[ p_d \right] \left( \frac{1}{k_D} \sum_{d = 1}^D \Omega_d'x 1_{\left\lbrace R_d \geq R_{(k_D)}\right\rbrace} \right),
\end{align*}
for $x \in \mathbb{R}^W$ are needed \cite{lehoucq1998arpack}. These iterations are faster to compute since the sums are over $k_D = o(D)$ documents. 

We now analyze the consistency of Algorithm \ref{alg:specLDA} with the extreme-value motivated estimators of the adjusted moments. In order to do so, it is useful to derive asymptotic normality of the empirical moments in \eqref{eq:s1}, \eqref{eq:s2}, \eqref{eq:s3}. This in turn provides a rate of convergence for the adjusted moment estimators. In addition to their use in the adjusted moment estimators, the off-diagonal elements of $\hat{\mathbb{E}}_D[p_d \otimes p_d]$ also provide an alternative, extreme-value based estimator of the probability that two randomly selected words from a typical document will be of a given pair. These word co-occurrence statistics play a large role in topic model estimation. For example, factorizations of modified co-occurrence statistics serve as precursors to modern word embeddings \cite{deerwester1990indexing, landauer1997solution}.

Throughout, let $F(x) = \mathbb{P}(R_d \leq x)$ for $x \geq 0$ and define $b(t) = \left( \frac{1}{1 - F} \right)^\leftarrow(t)$ for $t \geq 1$. We present asymptotic normality of the estimators in the following theorem. The proof of Theorem \ref{thm:EDM} is provided in Section \ref{sec:edmproof}.

\begin{theorem}
\label{thm:EDM}
Consider the setting of Theorem \ref{thm:freqmrv}. Suppose that as $D \rightarrow \infty$, $k_D \rightarrow \infty$ and $k_D/D \rightarrow 0$. If, in addition
\begin{align*}
\sqrt{k_D}/b(D/k_D) \rightarrow 0, \qquad \text{as } D \rightarrow \infty,
\end{align*}
then for any fixed $u, v, w \in \{1, 2, \dots, W \}$, we have the marginal convergences as $D \rightarrow \infty$
\begin{gather}
\sqrt{k_D}\left(\hat{\mathbb{E}}_D[p_{du}] - \mathbb{E}[p_{du}] \right) \Rightarrow N(0, \xi^2_{u}), \\
\sqrt{k_D}\left(\hat{\mathbb{E}}_D[p_{du}p_{dv}] - \mathbb{E}[p_{du}p_{dv}] \right) \Rightarrow N(0, \sigma^2_{u, v}), \\
\sqrt{k_D}\left(\hat{\mathbb{E}}_D[p_{du}p_{dv}p_{dw}] - \mathbb{E}[p_{du}p_{dv}p_{dw}] \right) \Rightarrow N(0, \gamma^2_{u, v, w}), \label{eq:norm3}
\end{gather}
where $\xi^2_{u} = \text{Var}(p_{du})$, $\sigma^2_{u, v} = \text{Var}(p_{du}p_{dv})$ and $\gamma^2_{u, v, w} = \text{Var}(p_{du}p_{dv}p_{dw})$, with all variances computed under the angular measure $S$, that is, for $p_d = \Phi\theta_d$ with $\theta_d$ distributed according to \eqref{eq:dirichlet}.
\end{theorem}

With Theorem \ref{thm:EDM} in hand, we may now state rates of convergence for the extreme-value based adjusted moment estimators. These rates are employed to derive the consistency of Algorithm \ref{alg:specLDA} using $\hat{\Sigma}_D$ and $\hat{\mathcal{T}}_D$ as inputs. The proof of Lemma \ref{lem:Op} is provided in Section \ref{sec:Opproof}.

\begin{lemma}
\label{lem:Op}
Consider the setting of Theorem \ref{thm:freqmrv}. Suppose that as $D \rightarrow \infty$, $k_D \rightarrow \infty$ and $k_D/D \rightarrow 0$. If, in addition
\begin{align*}
\sqrt{k_D}/b(D/k_D) \rightarrow 0, \qquad \text{as } D \rightarrow \infty,
\end{align*}
we have that
\begin{align*}
\|\hat{\tSigma}_D - \tSigma \| = O_p\left(\frac{1}{\sqrt{k_D}}\right), \quad \|\hat{\tmathcalT}_D - \tmathcalT \| = O_p\left(\frac{1}{\sqrt{k_D}}\right),
\end{align*}
where $\hat{\tmathcalT}_D$ and $\hat{\tSigma}_D$ denote the plug-in estimators for $\tmathcalT$ and $\tSigma$ using the extreme-value moment estimators defined in \eqref{eq:s1}, \eqref{eq:s2} and \eqref{eq:s3}.
\end{lemma}

Theorem \ref{thm:roc} establishes the $\sqrt{k_D}$ rate of convergence for the topic estimators in Step 4 of Algorithm \ref{alg:specLDA} under appropriate conditions on $k_D$ and $N_D$. Recall that the former condition is a bias condition needed in order to establish asymptotic normality of the extreme-value based moment estimators while the latter ensures that sufficiently many tensor power iterations are performed in accordance with Theorem 5.1 of \cite{anandkumar2014tensor}. Note that for fixed $L$, the robust tensor power method can fail with positive probability: as discussed in Section \ref{sec:tensor}, the $L$ random restarts are required to obtain, with high probability, an initialization that sufficiently separates the tensor eigenvalues relative to the error tensor. Hence, the failure probability is controlled by the parameter $\delta$ in the theorem statement. The proof of Theorem \ref{thm:roc}, along with supporting lemmas, is provided in Section \ref{sec:rocproof}.

\begin{theorem}
\label{thm:roc}
Consider the setting of Theorem \ref{thm:freqmrv}. Fix $\varsigma > 0$. Suppose that as $D \rightarrow \infty$, $k_D \rightarrow \infty$ and $k_D/D \rightarrow 0$. Further suppose that
\begin{align*}
&\sqrt{k_D}/b(D/k_D) \rightarrow 0, \\ 
& N_D/\log\log k_D \rightarrow \infty,
\end{align*}
as $D \rightarrow \infty$ and for a given $\delta > 0$
\begin{align*}
L \geq& \text{poly}(K)\log(1/\delta),
\end{align*}
for some fixed polynomial defined in Theorem 5.1 of \cite{anandkumar2014tensor}. Then there exists a constant $\kappa > 0$ such that for $D$ sufficiently large, after appropriate reordering
\begin{align*}
\mathbb{P}\left(\max_{k = 1, \dots, K} \left\Vert \hat{\beta}_{D,k} - \phi_k \right\Vert_2 > \frac{\kappa}{\sqrt{k_D}} \right) < \delta + \varsigma.
\end{align*}
\end{theorem}

\subsection{Selecting the number of topics}\label{sec:ntopic}

Algorithm \ref{alg:specLDA} estimates the topic model assuming that the number of topics, $K$, is known. Various strategies for choosing $K$ in topic modeling have been proposed. One class of methods maximizes the marginal likelihood or BIC-like metrics \cite{griffiths2004finding, bystrov2024choosing}. Another class of methods employs measures such as perplexity or topic coherence metrics on held-out documents \cite{blei2003latent, newman2011improving}. In accordance with our extreme-value motivated methodology, we propose a scree-plot procedure based on the spectral decomposition of $\hat{\Sigma}_D$. A similar procedure for the LDA model was analyzed by \cite{cheng2015model}.

Lemma \ref{lem:moments} reveals that $\Sigma$ is the sum of $K$ rank-1 symmetric matrices with linearly independent factors and thus $\Sigma$ is rank $K$. Denote the eigenvalues of $\Sigma$ by $s^{(1)} \geq s^{(2)} \geq \dots \geq s^{(W)}$. Since $\Sigma$ is positive semi-definite and rank $K$, the first $K$ eigenvalues are strictly positive while the rest are zero. Hence, if $\Sigma$ were known, the number of topics could be chosen by identifying the number of non-zero eigenvalues. Instead, we have access to the consistent estimator $\hat{\Sigma}_D$. Denote the eigenvalues of $\hat{\Sigma}_D$ by $\hat{s}_D^{(1)} \geq \hat{s}_D^{(2)} \geq \dots \geq \hat{s}_D^{(W)}$. Lemma \ref{lem:eigen} shows that the eigenvalues of $\hat{\Sigma}_D$ consistently estimate the eigenvalues of $\Sigma$, and thus one may choose the number of topics by identifying the number of non-zero eigenvalues of $\hat{\Sigma}_D$.

\begin{lemma}
\label{lem:eigen}
Consider the setting of Theorem \ref{thm:freqmrv}. Suppose that as $D \rightarrow \infty$, $k_D \rightarrow \infty$ and $k_D/D \rightarrow 0$. If, in addition
\begin{align*}
\sqrt{k_D}/b(D/k_D) \rightarrow 0, \qquad \text{as } D \rightarrow \infty,
\end{align*}
we have that
\begin{align*}
(\hat{s}_D^{(1)}, \hat{s}_D^{(2)}, \dots, \hat{s}_D^{(W)}) = (s^{(1)}, s^{(2)}, \dots, s^{(W)}) + O_p\left( \frac{1}{\sqrt{k_D}}\right).
\end{align*}
\end{lemma}

\begin{proof}
Using Weyl's inequality \citep[see Lemma 4.1.14 of][]{vershynin2026high},
\begin{align*}
\max_{w = 1, \dots, W} |\hat{s}_D^{(w)} - s^{(w)} | \leq \| \hat{\Sigma}_D - \Sigma \|.
\end{align*}
Hence, applying Lemma \ref{lem:Op} gives the result.
\end{proof}

\section{Simulation studies}\label{sec:sim}

In this section, the empirical performance of the proposed extreme-value estimation method for the topic matrix is investigated through simulations. We compare the estimation error for $\Phi$ and wall-clock time of Algorithm \ref{alg:specLDA} to that of the full spectral LDA algorithm proposed by \cite{anandkumar2012spectral, anandkumar2015spectral} and the variational expectation maximization algorithm of \cite{blei2003latent}. The first three simulations investigate how changes in the distribution of document lengths impact the performance of the three algorithms. The fourth simulation evaluates how misspecification of the $\alpha_0$ parameter impacts error rates. The simulations are run on an Intel Xeon Gold 6140 processor with 256 GB of memory.

In order to implement the VEM algorithm, we use the \texttt{topicmodels} R package \cite{grun2011topicmodels}. In VEM, the latent topic weights are governed by a Dirichlet prior as in \eqref{eq:dirichlet}. In the first three simulations of Section \ref{sec:dlsim}, we set the prior parameters of the Dirichlet distribution to the true values under the data generating process. This is done to draw a fair comparison to the spectral LDA methods, which only require specification of $\alpha_0$, the concentration parameter. For the other settings, we retain the defaults provided by the \texttt{topicmodels} package.

The spectral LDA algorithm is implemented using Python code that is available from \url{https://github.com/Mega-DatA-Lab/SpectralLDA} \cite{huang2017spectral}. This implementation of spectral LDA differs from Algorithm \ref{alg:specLDA} in two ways. First, one iteration of a randomized SVD algorithm is employed to accelerate the eigendecomposition of $\hat{\tSigma}_D$. This technique is not leveraged in our implementation of the extreme-value estimation method and hence advantages the full spectral LDA algorithm in our runtime comparisons. Secondly, an alternating least squares algorithm, an alternative to the tensor power method, is used to decompose the whitened tensor $\hat{\tmathcalT}_D(\hat{\mathcal{W}}_D, \hat{\mathcal{W}}_D, \hat{\mathcal{W}}_D)$. 

Our implementation of the extreme-value estimation method uses a Lanczos style algorithm from the \texttt{RSpectra} package to compute the eigendecomposition of $\hat{\tSigma}_D$ \cite{rSpectra}. The tensor power method requires specification of the number of random restarts $L$ and the number of power iterations $N_D$, both of which we set to 20. 

\subsection{Effect of document lengths}\label{sec:dlsim}

We consider documents generated from the LDA model presented in Section \ref{sec:tm}. We generate documents with $K = 5$ topics. We allow the number of documents to vary across $D \in \{1000, 2000, 5000 \}$ and the number of words in the vocabulary to vary between $W \in \{2000, 5000, 10000 \}$. For each vocabulary size $W$, we sample the true topic vectors $\phi_1, \dots, \phi_K$ from a symmetric Dirichlet distribution with concentration parameter $W \times 0.1$. This results in sparse topic vectors, as commonly found in practice. 

For each pair $(D, W)$ we simulate 1{,}000 corpus realizations. The document-topic probabilities in \eqref{eq:dirichlet} are drawn from a symmetric Dirichlet distribution with $\alpha_0 = 1$. In each of the three simulation scenarios, the document lengths are drawn from a different power-law distribution satisfying \eqref{eq:lengthrv}. For each estimation method, we report the average Frobenius error after alignment (via the Hungarian algorithm) and average run time across all realizations. The extreme-value method requires specification of $k_D$, the number of documents used in the estimation scheme. In order to assess the impact of this choice, we allow $k_D $ to vary according to $k_D = \gamma D$, $\gamma \in \{ 0.01, 0.025, 0.05, 0.10, 0.20\}$.
We also report the average number of iterations needed to perform the eigendecomposition of $\hat{\tSigma}_D$ in the extreme-value application of Algorithm \ref{alg:specLDA}. This provides context for the non-monotonicity of the run times.

In the first simulation, we assume that $N_{d\cdot\cdot} = \ceil{Y_d}$ where $Y_1, \dots, Y_D$ are drawn iid from a Lomax distribution:
\begin{align*}
\mathbb{P}(Y_1 > y) = \left(1 + \frac{y}{\lambda}\right)^{-\iota}, \qquad y \geq 0,
\end{align*}
where $\lambda = 50$ and $\iota = 2$. This distribution produces power-law document lengths with non-negligible average length since $\mathbb{E}[Y_1] = 50$. The average Frobenius error for estimating $\Phi$ over the 1{,}000 realizations is reported in Table \ref{tab:sim1error}. Across sample sizes and vocabularies, the extreme-value method outperforms the full spectral LDA method for nearly all choices of $k_D$. The full spectral LDA method uses all documents with three or more words in its estimation of the moments, and as such, the inclusion of short documents seems to produce noisy estimators. The VEM algorithm produces the uniformly smallest estimation error, but at the expense of exceptionally long runtimes as seen in Table \ref{tab:sim1time}. The shortest average runtimes belong to the extreme-value method, lending credence to the computational benefits of using much fewer documents in the estimation procedure. While the average runtime tends to increase with $k_D$, the effect is non-uniform. The non-uniformity is correlated with the average number of iterations required to decompose $\hat{\tSigma}_D$ in Algorithm \ref{alg:specLDA}, as reported in Table \ref{tab:sim1iter}. As $k_D$ increases, more documents are used to construct $\hat{\tSigma}_D$, resulting in less noisy estimates. However, $\hat{\Sigma}_D$ is a function of $\hat{\mathbb{E}}_D[p_d \otimes p_d]$ which has rank at most $\min\left\lbrace k_D, W \right\rbrace$, ignoring ties in document lengths. Hence, increasing $k_D$ has some non-monotone effects on the computational burden of the eigendecomposition. 

\begin{table}
\centering
\begin{tabular}{ c c c c c c c c c}
\hline \hline
 & & \multicolumn{5}{c}{EVT, $k_D = \gamma D$} & & \\
\cline{3-7}
$D$ & $W$ & $\gamma = 0.01$ & $\gamma = 0.025$ & $\gamma = 0.05$ & $\gamma = 0.10$ & $\gamma = 0.20$ & SLDA & VEM  \\
\hline
 1000  &2000 &0.1085 &0.0725 &0.0584 &0.0502 &0.0457 &0.2342  &0.0340 \\
 1000  &5000 &0.0993 &0.0733 &0.0589 &0.0507 &0.0461 &0.3978 &0.0319 \\
 1000 &10000 &0.0999  &0.0820  &0.0640 &0.0538  &0.0490 &0.3918 &0.0319 \\
 2000  &2000 &0.0664 &0.0481 &0.0401  &0.0350 &0.0323  &0.0820 &0.0255 \\
 2000  &5000  &0.0630 &0.0469 &0.0396 &0.0346 &0.0318 &0.2203 &0.0234 \\
 2000 &10000 &0.0663 &0.0482 &0.0403 &0.0352 &0.0323 &0.2692 &0.0222 \\
 5000  &2000 &0.0384 &0.0296 &0.0252 &0.0222 &0.0211 &0.0381  &0.0180 \\
 5000  &5000 &0.0364 &0.0287 &0.0246 &0.0217   &0.0200  &0.0550 &0.0167 \\
 5000 &10000 &0.0361 &0.0285 &0.0245 &0.0216 &0.0199  &0.1120 &0.0155 \\
\hline \hline
\end{tabular}
\caption{Average Frobenius error for estimating $\Phi$ when document lengths are drawn from a rounded Lomax distribution.}\label{tab:sim1error}
\end{table}

\begin{table}
\centering
\begin{tabular}{ c c c c c c c c c}
\hline \hline
 & & \multicolumn{5}{c}{EVT, $k_D = \gamma D$} & & \\
\cline{3-7}
$D$ & $W$ & $\gamma = 0.01$ & $\gamma = 0.025$ & $\gamma = 0.05$ & $\gamma = 0.10$ & $\gamma = 0.20$ & SLDA & VEM  \\
\hline
 1000  & 2000 & 0.0123 & 0.0175 & 0.0183 & 0.0195 & 0.0211 & 0.0615  & 4.1119 \\
 1000  & 5000 & 0.0192 & 0.0339 & 0.0328 & 0.0316  & 0.0330 & 0.1058  & 5.8667 \\
 1000 & 10000 & 0.0287  & 0.0680 & 0.0712 & 0.0654 & 0.0658 & 0.1514    & 8.7600 \\
 2000  & 2000 & 0.0159 & 0.0177 & 0.0184 & 0.0211 & 0.0257 & 0.0631  & 8.1018 \\
 2000  & 5000 & 0.0295  & 0.0310 & 0.0322 & 0.0341  & 0.0370  & 0.1050   & 9.8690 \\
 2000 & 10000 & 0.0518  & 0.0550 & 0.0503 & 0.0508 & 0.0534 & 0.1557 & 12.2372 \\
 5000  & 2000 & 0.0165 & 0.0186 & 0.0215 & 0.0255 & 0.0342 & 0.0978 & 19.0227 \\
 5000  & 5000 & 0.0324 & 0.0354 & 0.0397 & 0.0458 & 0.0529 & 0.1352 & 22.7959 \\
 5000 & 10000 & 0.0509 & 0.0539 & 0.0579 & 0.0633 & 0.0697 & 0.1947 & 25.6062 \\
\hline \hline
\end{tabular}
\caption{Average runtime (in seconds) when document lengths are drawn from a rounded Lomax distribution.}\label{tab:sim1time}
\end{table}

\begin{table}
\centering
\begin{tabular}{ c c c c c c c }
\hline \hline
$D$ & $W$ & $\gamma = 0.01$ & $\gamma = 0.025$ & $\gamma = 0.05$ & $\gamma = 0.10$ & $\gamma = 0.20$  \\
\hline
 1000  &2000     &1.000 &2.035     &2.000     &2.000     &2.000 \\
 1000  &5000     &1.000 &2.578 &2.306 &2.056  &2.020 \\
 1000 &10000     &1.000 &3.605 &3.604 &3.144 &3.034 \\
 2000  &2000 &1.567 &1.729 &1.614  &1.760     &2.000 \\
 2000  &5000 &2.103 &2.024     &2.000     &2.000     &2.000 \\
 2000 &10000 &2.531 &2.573 &2.091 &2.004 &2.001 \\
 5000  &2000 &1.026 &1.002     &1.000 &1.002 &1.461 \\
 5000  &5000 &2.002 &1.995 &1.991     &2.000     &2.000 \\
 5000 &10000 &2.025     &2.000     &2.000     &2.000     &2.000 \\
\hline \hline
\end{tabular}
\caption{Average number of eigendecomposition iterations for the extreme-value method when document lengths are drawn from a rounded Lomax distribution.}\label{tab:sim1iter}
\end{table}

Table \ref{tab:sim2error} reports the average Frobenius error $\Phi$ over 1{,}000 corpus realizations of documents with lengths satisfying $N_{d\cdot\cdot} = \ceil{Y_d}$ where $Y_1, \dots, Y_D$ are drawn iid from a Pareto distribution:
\begin{align*}
\mathbb{P}(Y_1 > y) = (y/C)^{-\iota}, \qquad y \geq C,
\end{align*}
where $C = 49$ and $\iota = 2$. This distribution produces power-law document lengths with minimum length 50. Naturally, the full spectral LDA method produces uniformly lower error than  the extreme-value procedure since all documents are sufficiently long to bear reliable moment estimates. Given the discrepancy in sample size, however, the extreme-value procedure remains competitive. For larger corpora, it also outperforms the VEM method while generally matching its performance otherwise. As presented in Table \ref{tab:sim2time}, the extreme-value method retains its computational advantages in this setting. The average run time can be three to four times as fast as the full spectral LDA method. As $k_D$ increases, the average run time tends to increase, with the same non-monotonicity as observed in the previous setting. The average number of iterations required to decompose $\hat{\tSigma}_D$ in Algorithm \ref{alg:specLDA} is reported in Table \ref{tab:sim2iter}.

\begin{table}
\centering
\begin{tabular}{ c c c c c c c c c}
\hline \hline
 & & \multicolumn{5}{c}{EVT, $k_D = \gamma D$} & & \\
\cline{3-7}
$D$ & $W$ & $\gamma = 0.01$ & $\gamma = 0.025$ & $\gamma = 0.05$ & $\gamma = 0.10$ & $\gamma = 0.20$ & SLDA & VEM  \\
\hline
 1000  &2000 &0.1054 &0.0681 &0.0535 &0.0438 &0.0365  &0.0240  &0.0250 \\
 1000  &5000 &0.0949 &0.0666 &0.0525 &0.0431 &0.0359 &0.0238 &0.0239 \\
 1000 &10000 &0.0948 &0.0717 &0.0543 &0.0437 &0.0361  &0.0240 &0.0221 \\
 2000  &2000 &0.0643 &0.0458 &0.0371 &0.0307 &0.0259 &0.0169 &0.0187 \\
 2000  &5000 &0.0605 &0.0442 &0.0361 &0.0301 &0.0253 &0.0167 &0.0171 \\
 2000 &10000 &0.0624 &0.0445 &0.0362   &0.0300 &0.0252 &0.0167 &0.0169 \\
 5000  &2000 &0.0376 &0.0282 &0.0233 &0.0196 &0.0168 &0.0107 &0.0126 \\
 5000  &5000 &0.0353 &0.0271 &0.0226  &0.0190 &0.0161 &0.0106 &0.0117 \\
 5000 &10000 &0.0349 &0.0269 &0.0224 &0.0188 &0.0159 &0.0106  &0.0110 \\
\hline \hline
\end{tabular}
\caption{Average Frobenius error for estimating $\Phi$ when document lengths are drawn from a rounded Pareto distribution.}\label{tab:sim2error}
\end{table}

\begin{table}
\centering
\begin{tabular}{ c c c c c c c c c}
\hline \hline
 & & \multicolumn{5}{c}{EVT, $k_D = \gamma D$} & & \\
\cline{3-7}
$D$ & $W$ & $\gamma = 0.01$ & $\gamma = 0.025$ & $\gamma = 0.05$ & $\gamma = 0.10$ & $\gamma = 0.20$ & SLDA & VEM  \\
\hline
 1000  & 2000 & 0.0126  & 0.0170 & 0.0179  & 0.0190 & 0.0209 & 0.0603  & 7.8639 \\
 1000  & 5000 & 0.0194 & 0.0317 & 0.0307 & 0.0314 & 0.0333  & 0.0860  & 9.4576 \\
 1000 & 10000 & 0.0286 & 0.0622 & 0.0598 & 0.0514 & 0.0504 & 0.1231 & 11.7312 \\
 2000  & 2000 & 0.0162 & 0.0168 & 0.0168 & 0.0185 & 0.0218 & 0.0804 & 15.3605 \\
 2000  & 5000 & 0.0298 & 0.0311 & 0.0327 & 0.0356   & 0.0400 & 0.1104 & 17.7308 \\
 2000 & 10000 & 0.0517 & 0.0519 & 0.0507 & 0.0535 & 0.0577 & 0.1513 & 19.9068 \\
 5000  & 2000 & 0.0195 & 0.0218 & 0.0249   & 0.0300 & 0.0378 & 0.1508 & 37.8513 \\
 5000  & 5000 & 0.0353  & 0.0370 & 0.0374 & 0.0424 & 0.0545 & 0.1942   & 44.0300 \\
 5000 & 10000 & 0.0535 & 0.0568 & 0.0609 & 0.0677 & 0.0784  & 0.2480 & 48.0551 \\
\hline \hline
\end{tabular}
\caption{Average runtime (in seconds) when document lengths are drawn from a rounded Pareto distribution.}\label{tab:sim2time}
\end{table}

\begin{table}
\centering
\begin{tabular}{ c c c c c c c }
\hline \hline
$D$ & $W$ & $\gamma = 0.01$ & $\gamma = 0.025$ & $\gamma = 0.05$ & $\gamma = 0.10$ & $\gamma = 0.20$  \\
\hline
 1000  &2000     &1.000 &1.949 &1.972 &1.963 &1.943 \\
 1000  &5000     &1.000 &2.351 &2.065 &2.001     &2.000 \\
 1000 &10000     &1.000  &3.290 &2.964 &2.249 &2.001 \\
 2000  &2000 &1.434 &1.389 &1.108 &1.016 &1.041 \\
 2000  &5000 &2.064 &2.011     &2.000     &2.000     &2.000 \\
 2000 &10000 &2.415 &2.232 &2.004     &2.000     &2.000 \\
 5000  &2000  &1.020     &1.000     &1.000     &1.000 &1.001 \\
 5000  &5000 &1.991 &1.834 &1.392 &1.275 &1.526 \\
 5000 &10000 &2.013     &2.000     &2.000     &2.000     &2.000 \\
\hline \hline
\end{tabular}
\caption{Average number of eigendecomposition iterations for the extreme-value method when document lengths are drawn from a rounded Pareto distribution.}\label{tab:sim2iter}
\end{table}

As an intermediate setting, we now consider document lengths generated according to $N_{d\cdot\cdot} = \ceil{Y_d}$ where $Y_1, \dots, Y_D$ are drawn iid from a mixture of Pareto distributions:
\begin{align}
\label{eq:mix}
\mathbb{P}(Y_1 > y) = 0.5\min\left\lbrace 1, (y/C_1)^{-\iota_1} \right\rbrace + 0.5 \min\left\lbrace 1, (y/C_2)^{-\iota_2}\right\rbrace,
\end{align}
where $\iota_1 = 2$ and $C_1 = 49$ while $\iota_2 = 3$ and $C_2 = 9$. Hence, the minimum document length is 10, but over half of the documents have a length of at least 50. Table \ref{tab:sim3error} again reports the average Frobenius error $\Phi$ over 1{,}000 corpus realizations. For sufficiently large choices of $k_D$, the extreme-value method outperforms the full spectral LDA method. A small choice of $k_D$ produces noisy sample moments, while tuning $k_D$ to use many documents over the length 50 threshold produces quality estimates. On the other hand, the inclusion of short documents tends to degrade the performance of the spectral LDA method. The VEM algorithm again outperforms the other methods in terms of estimation error with the drawback of onerous computational demands as seen in Table \ref{tab:sim3time}. The extreme-value method is again the most computationally efficient among the three, no matter the choice of $k_D$. Table \ref{tab:sim3iter} reports the average number of iterations required to decompose $\hat{\tSigma}_D$ in Algorithm \ref{alg:specLDA}.

\begin{table}
\centering
\begin{tabular}{ c c c c c c c c c}
\hline \hline
 & & \multicolumn{5}{c}{EVT, $k_D = \gamma D$} & & \\
\cline{3-7}
$D$ & $W$ & $\gamma = 0.01$ & $\gamma = 0.025$ & $\gamma = 0.05$ & $\gamma = 0.10$ & $\gamma = 0.20$ & SLDA & VEM  \\
\hline
 1000  &2000 &0.1187 &0.0801 &0.0635 &0.0523 &0.0439 &0.0547  &0.0310 \\
 1000  &5000 &0.1118 &0.0821 &0.0637  &0.0520 &0.0433 &0.0684 &0.0309 \\
 1000 &10000 &0.1139  &0.0940 &0.0687 &0.0538  &0.0440 &0.0763 &0.0296 \\
 2000  &2000 &0.0749 &0.0538 &0.0439 &0.0367 &0.0312 &0.0365 &0.0231 \\
 2000  &5000 &0.0734 &0.0529 &0.0433 &0.0361 &0.0304 &0.0418 &0.0224 \\
 2000 &10000  &0.0790 &0.0548 &0.0439 &0.0362 &0.0304 &0.0513 &0.0215 \\
 5000  &2000 &0.0434 &0.0331 &0.0275 &0.0233 &0.0203 &0.0223 &0.0164 \\
 5000  &5000 &0.0417 &0.0322 &0.0268 &0.0226 &0.0193 &0.0228 &0.0159 \\
 5000 &10000 &0.0419 &0.0321 &0.0268 &0.0225 &0.0191 &0.0265 &0.0148 \\
\hline \hline
\end{tabular}
\caption{Average Frobenius error for estimating $\Phi$ when document lengths are drawn from a mixture of rounded Pareto distributions.}\label{tab:sim3error}
\end{table}

\begin{table}
\centering
\begin{tabular}{ c c c c c c c c c}
\hline \hline
 & & \multicolumn{5}{c}{EVT, $k_D = \gamma D$} & & \\
\cline{3-7}
$D$ & $W$ & $\gamma = 0.01$ & $\gamma = 0.025$ & $\gamma = 0.05$ & $\gamma = 0.10$ & $\gamma = 0.20$ & SLDA & VEM  \\
\hline
 1000  & 2000 & 0.0119 & 0.0173 & 0.0178 & 0.0184   & 0.0200 & 0.0535  & 4.7256 \\
 1000  & 5000 & 0.0183 & 0.0353 & 0.0338 & 0.0308 & 0.0318 & 0.0801  & 6.4017 \\
 1000 & 10000 & 0.0272 & 0.0687 & 0.0718 & 0.0629 & 0.0591  & 0.1340   & 9.8760 \\
 2000  & 2000 & 0.0164 & 0.0183 & 0.0192 & 0.0212 & 0.0246 & 0.0688  & 9.2662 \\
 2000  & 5000 & 0.0294 & 0.0306 & 0.0311 & 0.0328 & 0.0358 & 0.0888 & 11.1045 \\
 2000 & 10000 & 0.0533 & 0.0603 & 0.0515 & 0.0502  & 0.0530 & 0.1307 & 13.6178 \\
 5000  & 2000 & 0.0173 & 0.0183 & 0.0208  & 0.0250 & 0.0318 & 0.1094 & 22.2926 \\
 5000  & 5000 & 0.0327  & 0.0350 & 0.0385 & 0.0444 & 0.0523 & 0.1444 & 25.6023 \\
 5000 & 10000 & 0.0518 & 0.0527  & 0.0560  & 0.0610  & 0.0680 & 0.1855 & 28.7198 \\
\hline \hline
\end{tabular}
\caption{Average runtime (in seconds) when document lengths are drawn from a mixture of rounded Pareto distributions.}\label{tab:sim3time}
\end{table}

\begin{table}
\centering
\begin{tabular}{ c c c c c c c }
\hline \hline
$D$ & $W$ & $\gamma = 0.01$ & $\gamma = 0.025$ & $\gamma = 0.05$ & $\gamma = 0.10$ & $\gamma = 0.20$  \\
\hline
 1000  &2000     &1.000 &2.083 &2.012     &2.000     &2.000 \\
 1000  &5000     &1.000  &2.870 &2.546 &2.088 &2.001 \\
 1000 &10000     &1.000 &3.891 &3.847 &3.146  &2.830 \\
 2000  &2000 &1.816 &1.991 &1.956 &1.956 &1.997 \\
 2000  &5000 &2.192 &2.076     &2.000     &2.000     &2.000 \\
 2000 &10000 &2.772 &3.043 &2.281 &2.001     &2.000 \\
 5000  &2000 &1.198 &1.015 &1.003 &1.012 &1.176 \\
 5000  &5000 &2.003     &2.000     &2.000     &2.000     &2.000 \\
 5000 &10000 &2.124     &2.000     &2.000     &2.000     &2.000 \\
\hline \hline
\end{tabular}
\caption{Average number of eigendecomposition iterations for the extreme-value method when document lengths are drawn from a mixture of rounded Pareto distributions.}\label{tab:sim3iter}
\end{table}

\subsection{Effect of misspecification}\label{sec:alphasim}

Both the extreme-value and full spectral LDA methods require specification of the concentration parameter $\alpha_0$, while the VEM algorithm requires full specification of the prior parameters $(\alpha_1, \dots, \alpha_K)$. Note that an empirical Bayes method can be applied to estimate the prior parameters of the VEM algorithm \cite{blei2003latent}. In this section, we evaluate how misspecification of these parameters affects the performance of the algorithms. 

Again consider 1{,}000 realizations of a corpus with $K = 5$ topics. We simulate corpora of $D = 2000$ documents and $W = 5000$ words. The true topic vectors $\phi_1, \dots, \phi_K$ are sampled independently from a symmetric Dirichlet distribution with concentration parameter $W \times 0.1$. The topic-document probabilities $\theta_1, \dots, \theta_D$ are drawn from a symmetric Dirichlet distribution with concentration parameter $\alpha_0 = 1$. We generate document lengths from the mixture in \eqref{eq:mix}. 

In order to assess how misspecification of $\alpha_0$ degrades the performance of the extreme-value and spectral LDA algorithms, we apply the algorithms with $\alpha_0 \in \{0.1, 0.5, 1, 2, 10 \}$. For the VEM algorithm, we set the prior parameters to $\alpha_k = \alpha_0/K$ for $\alpha_0 \in \{0.1, 0.5, 1, 2, 10 \}$. Hence, the case where $\alpha_0 = 1$ is correctly specified.  Table \ref{tab:sim4error} reports the average Frobenius error for estimating $\Phi$ over the $1{,}000$ iterations. In this setting, the extreme-value method produces lower estimation error for large choices of $k_D$ except in the extreme case where the concentration is misspecified to $10$. For the VEM algorithm, misspecification of $\alpha_0$ only affects the prior. Hence, since the number of words is comparatively large, the likelihood dominates and the VEM algorithm is very robust to such misspecification. As seen in Tables \ref{tab:sim4time} and \ref{tab:sim4iter}, the same computational advantages of the extreme-value method are retained under misspecification.

\begin{table}
\centering
\begin{tabular}{ c c c c c c c c}
\hline \hline
 & \multicolumn{5}{c}{EVT, $k_D = \gamma D$} & & \\
\cline{2-6}
$\alpha_0$ & $\gamma = 0.01$ & $\gamma = 0.025$ & $\gamma = 0.05$ & $\gamma = 0.10$ & $\gamma = 0.20$ & SLDA & VEM  \\
\hline
0.1 &0.0656 &0.0498 &0.0429 &0.0379 &0.0342 &0.0448 &0.0229 \\
0.5 &0.0683 &0.0499 &0.0415 &0.0352 &0.0303 &0.0422 &0.0222 \\
1 &0.0728 &0.0525 &0.0431 &0.0359 &0.0303 &0.0421 &0.0223 \\
2 &0.0835 &0.0621 &0.0523 &0.0452 &0.0401 &0.0487 &0.0217 \\
10 &0.1795 &0.1533  &0.1360 &0.1291 &0.1306 &0.1211 &0.0209 \\
\hline \hline
\end{tabular}
\caption{Average Frobenius error for estimating $\Phi$ when document lengths are drawn from a mixture of rounded Pareto distributions and $\alpha_0$ is misspecified.}\label{tab:sim4error}
\end{table}

\begin{table}
\centering
\begin{tabular}{ c c c c c c c c}
\hline \hline
 & \multicolumn{5}{c}{EVT, $k_D = \gamma D$} & & \\
\cline{2-6}
$\alpha_0$ & $\gamma = 0.01$ & $\gamma = 0.025$ & $\gamma = 0.05$ & $\gamma = 0.10$ & $\gamma = 0.20$ & SLDA & VEM  \\
\hline
0.1 & 0.0291   & 0.0300 & 0.0307 & 0.0325 & 0.0356 & 0.0888  & 9.1122 \\
0.5 & 0.0293 & 0.0301 & 0.0308 & 0.0326 & 0.0357 & 0.0892 & 10.3751 \\
1 & 0.0297 & 0.0306 & 0.0313 & 0.0329 & 0.0359  & 0.0890 & 10.9873 \\
2   & 0.0300 & 0.0307 & 0.0314 & 0.0333 & 0.0363 & 0.0879  & 11.8260 \\
10 & 0.0346 & 0.0527 & 0.0602 & 0.0628 & 0.0559 & 0.0997  & 6.8245 \\
\hline \hline
\end{tabular}
\caption{Average runtime (in seconds) when document lengths are drawn from a mixture of rounded Pareto distributions and $\alpha_0$ is misspecified.}\label{tab:sim4time}
\end{table}

\begin{table}
\centering
\begin{tabular}{ c c c c c c }
\hline \hline
$\alpha_0$ & $\gamma = 0.01$ & $\gamma = 0.025$ & $\gamma = 0.05$ & $\gamma = 0.10$ & $\gamma = 0.20$  \\
\hline
0.1 &2.175 &2.078     &2.000     &2.000     &2.000 \\
0.5 &2.185 &2.075     &2.000     &2.000     &2.000 \\
1 &2.207 &2.076     &2.000     &2.000     &2.000 \\
2 &2.233 &2.092 &2.001     &2.000     &2.000 \\
10 &2.864 &5.063 &5.818 &5.717 &4.357 \\
\hline \hline
\end{tabular}
\caption{Average number of eigendecomposition iterations for the extreme-value method when document lengths are drawn from a mixture of rounded Pareto distributions and $\alpha_0$ is misspecified. Here, $k_D = \gamma D$.}\label{tab:sim4iter}
\end{table}

\section{Analysis of twenty newsgroups}\label{sec:20ng}

In this section, we apply the extreme-value and full spectral LDA procedures to the estimation of topics in the twenty newsgroups corpus \cite{lang1995newsweeder}. This dataset, publicly available at \url{http://qwone.com/~jason/20Newsgroups/}, is a well-known benchmark appearing in popular machine learning libraries \cite{pedregosa2011scikit}. In particular, we consider the training split of the dataset. Our analysis highlights the robustness of the extreme-value methodology to the often understated choices made in the pre-processing of text data. In order to prepare the dataset for topic modeling, we follow the data cleaning steps outlined in Section 9.1 of \cite{silge2017text}. The data cleaning steps taken in this case study of the twenty newsgroups dataset include the removal of headers, email signatures and nested quotes. These steps are further inspected at the end of this section. After pre-processing, the cleaned corpus contains $D = 9{,}750$ documents and $W = 68{,}031$ unique words. While $W$ is relatively large in comparison to $D$, we follow other analyses of the dataset and refrain from further word removal \cite{soleimani2014parsimonious, dieng2020topic}.

The pre-processed corpus prominently exhibits the power-law behavior of word frequencies. A plot displaying the empirical tail of the document length distribution on a log-log scale is provided in the left-hand panel of Figure \ref{fig:tail}. Notice that, on a log-log scale, there is a strong linear relationship between the tail probability and the document length, indicating that the document lengths are indeed power-law distributed. In order to determine the threshold beyond which a power-law tail fits best, we employ the minimum distance procedure \cite{clauset2009power, drees2020minimum}. This procedure chooses the threshold that minimizes the Kolmogorov-Smirnov distance between the empirical and fitted tail beyond the threshold. It returns an optimal threshold of $109$, which is indicated by the dashed line in Figure \ref{fig:tail}. The number of documents with length exceeding this threshold is $k_D = 1441$. We employ exactly these documents in our extreme-value estimation procedure. While Theorems \ref{thm:EDM} and \ref{thm:roc} require the number of extreme documents to satisfy the bias condition $\sqrt{k_D}/b(D/k_D) \rightarrow 0$, such rate conditions offer limited practical guidance for a fixed corpus. The minimum distance procedure provides a data-driven choice of $k_D$, and the simulations in Section \ref{sec:sim} indicate that the performance of the extreme-value method is stable over a wide range of $k_D$. The right-hand panel of Figure \ref{fig:tail} plots the Hill estimator, a popular estimator of the tail index, as a function of $k_D$. Note that the estimator stabilizes around $k_D = 1441$ and returns a document length tail index estimate of $\hat{\iota} = 1.55$.  

\begin{figure}[h]
\centering
\includegraphics[width=\textwidth]{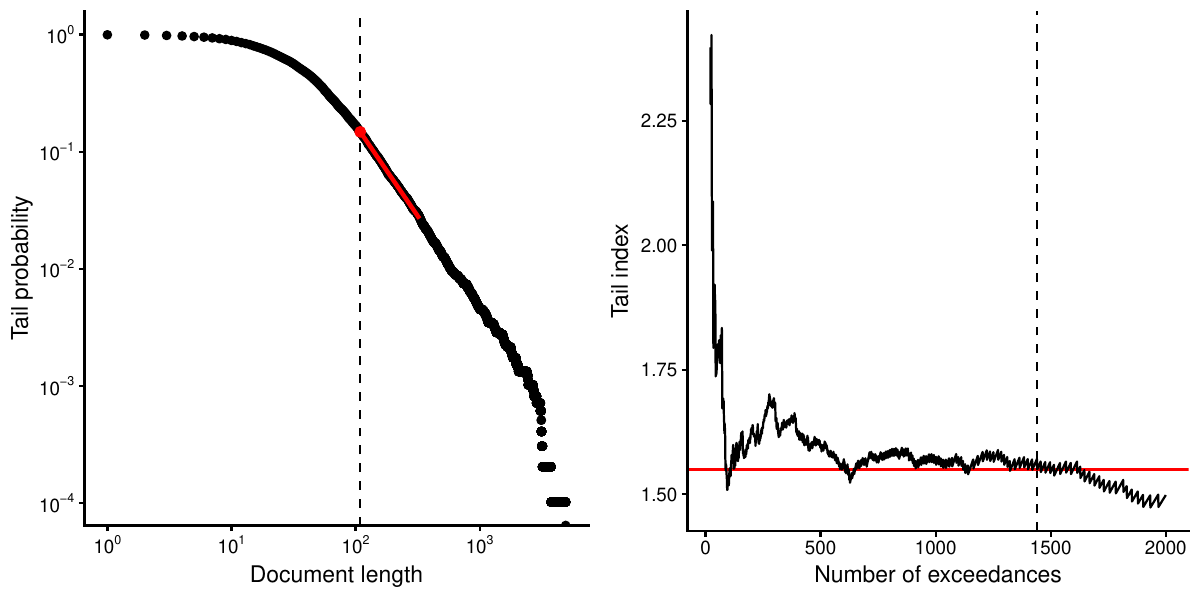}
\caption{Left: Log-log plot of the empirical tail of document lengths. Right: Hill plot for estimating the tail index of document lengths. The dashed lines indicate the point beyond which a power-law fits best as determined by the minimum distance procedure.}\label{fig:tail}
\end{figure}

In order to estimate the topics in the twenty newsgroups corpus with either the extreme-value or full spectral LDA estimation procedures, the parameters $K$ and $\alpha_0$ must be specified. Here, we choose these parameters by jointly minimizing the reconstruction error for $\alpha_0$. That is, given $\alpha_0$, the spectral methods estimate the Dirichlet parameters $(\alpha_1, \dots, \alpha_K)$ through the robust tensor eigenvalues. These estimates are not constrained to sum to $\alpha_0$. Hence, in order to retain model consistency, it is sensible to minimize the discrepancy between $\hat{\alpha}_0$ and $\alpha_0$. 

The Dirichlet parameters of the document-topic proportions in the twenty newsgroups corpus are estimated by applying the extreme-value based procedure over the grid $K \in \{2, 3, \dots, 50\}$ and $\alpha_0 \in \{0.001, 0.01, 0.1, 1 \}$. Note that the scalability of the extreme-value method ensures the computational feasibility of this procedure. Since the grid of $\alpha_0$ values vary in scale, we choose the $K$ and $\alpha_0$ pair that minimize the absolute relative error between $\hat{\alpha}_0$ and $\alpha_0$. The right-hand panel of Figure \ref{fig:Kalpha0} displays the relative reconstruction error for $\alpha_0$ over the grid. Note that for $\alpha_0 \in \{0.001, 0.01, 0.1\}$, the absolute relative error is minimized at $K = 42$, while for $\alpha_0 = 1$ it is minimized at $K = 46$. Overall, the error is jointly minimized at $(K, \alpha_0) = (42, 0.01)$. For each choice of $\alpha_0$, the top 100 eigenvalues of $\hat{\Sigma}_D$ are plotted in the left-hand panel of Figure \ref{fig:Kalpha0}. The choice of $K = 42$ is consistent with the flat spectrum of $\hat{\Sigma}_D$ beyond that point. Other analyses of this corpus have arrived at a similar choice in the number of topics \cite{soleimani2014parsimonious, srivastava2017autoencoding}.

\begin{figure}[h]
\centering
\includegraphics[width=\textwidth]{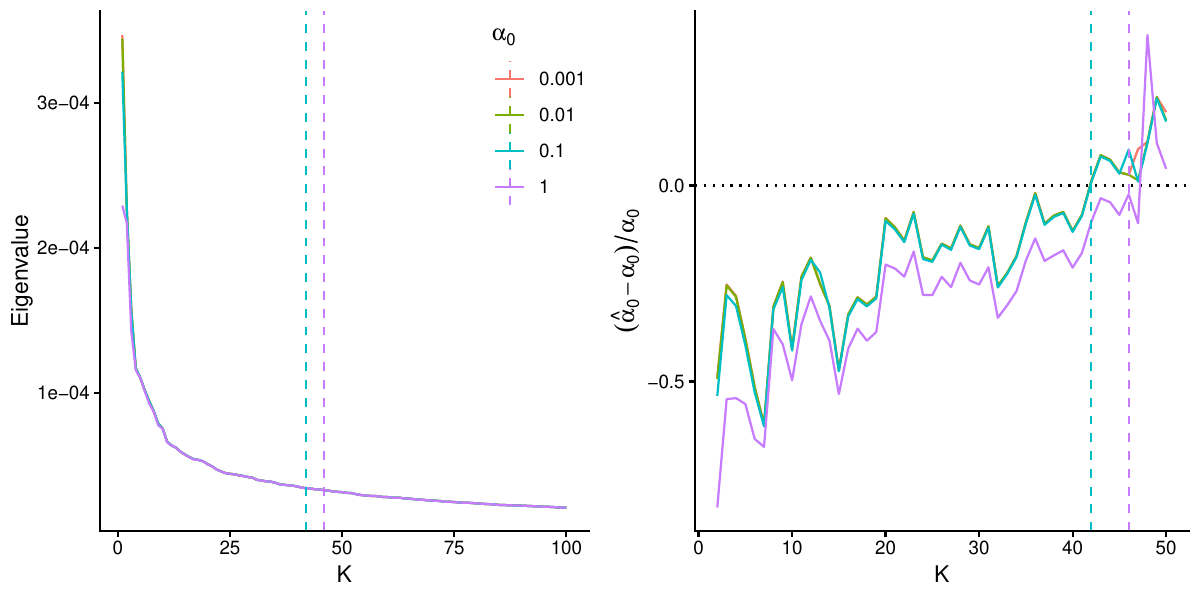}
\caption{Left: The top 100 eigenvalues of $\hat{\Sigma}_D$ across varying $\alpha_0$. Right: Relative error for estimating $\alpha_0$ across varying $K$ and $\alpha_0$. Note that the lines for different $\alpha_0$ overlap, indicating robustness to this choice.}\label{fig:Kalpha0}
\end{figure}

We proceed by applying the extreme-value and full spectral LDA methods to the twenty newsgroups corpus with $(K, \alpha_0) = (42, 0.01)$. The extreme-value procedure took 9 seconds to fit, while the full spectral LDA method took 145 seconds. For both the extreme-value and full spectral LDA methods, we display the 5 topics with the highest probability of occurrence in Tables \ref{tab:evttopic} and \ref{tab:sldatopic}, respectively. In particular, the 10 words with the highest estimated probabilities in each topic are reported. Note that the most frequently occurring topic reported by the extreme-value method contains generic, discussion-level words that appear across newsgroups. The other topics seem to center around sports, science, religion and privacy. The most likely topics reported by the spectral LDA method seem to overlap in computer-related terminology, while the other topics discuss religion and privacy. 

\begin{table}
\centering
\begin{tabular}{ c c c c c}
\hline \hline
Topic 1 & Topic 2 & Topic 3 & Topic 4 & Topic 5 \\
\hline 
  $\hat{\alpha}_k/\hat{\alpha}_0$ = 0.130 & $\hat{\alpha}_k/\hat{\alpha}_0$ = 0.103 & $\hat{\alpha}_k/\hat{\alpha}_0$ = 0.089 & $\hat{\alpha}_k/\hat{\alpha}_0$ = 0.087 & $\hat{\alpha}_k/\hat{\alpha}_0$ = 0.051 \\
  \hline
                                time (0.017)                                  &game (0.010)                              &science (0.009)                               &people (0.035)                           &government (0.028) \\
                              people (0.004)                                  &team (0.010)                                  &god (0.008)                                  &god (0.006)                               &people (0.016) \\
                                data (0.004)                                 &time (0.006)                                 &true (0.007)                           &government (0.006)                           &encryption (0.007) \\
                                real (0.003)                              &players (0.006)                               &people (0.007)                             &religion (0.003)                               &system (0.005) \\
                                life (0.002)                                 &play (0.006)                             &question (0.007)                                 &time (0.003)                               &rights (0.005) \\
                                read (0.002)                               &season (0.005)                                &truth (0.006)                                &jesus (0.003)                                 &time (0.004) \\
                            question (0.002)                                &games (0.005)                                &human (0.004)                           &christians (0.003)                              &privacy (0.004) \\
                                food (0.002)                               &hockey (0.004)                                &world (0.004)                              &militia (0.003)                               &strong (0.004) \\
                               event (0.002)                               &player (0.004)                                &moral (0.004)                               &person (0.003)                                &power (0.003) \\
                               scope (0.002)                                  &win (0.004)                             &morality (0.003)                                 &life (0.002)                                &money (0.003) \\
\hline \hline
\end{tabular}
\caption{Top 10 words from the five most frequently occurring topics as estimated by the extreme-value procedure.}\label{tab:evttopic}
\end{table}

\begin{table}
\centering
\begin{tabular}{ c c c c c}
\hline \hline
Topic 1 & Topic 2 & Topic 3 & Topic 4 & Topic 5 \\
\hline 
$\hat{\alpha}_k/\hat{\alpha}_0$ = 0.098 & $\hat{\alpha}_k/\hat{\alpha}_0$ = 0.087 & $\hat{\alpha}_k/\hat{\alpha}_0$ = 0.063 & $\hat{\alpha}_k/\hat{\alpha}_0$ = 0.061 & $\hat{\alpha}_k/\hat{\alpha}_0$ = 0.061 \\
\hline
                           christians (0.009)                              &windows (0.024)                                &modem (0.014)                               &driver (0.017)                                  &key (0.031) \\
                                 file (0.008)                                  &ftp (0.016)                                &mouse (0.013)                              &windows (0.009)                                 &chip (0.013) \\
                           considered (0.007)                                &files (0.016)                                 &port (0.012)                               &months (0.009)                                 &keys (0.013) \\
                              mormons (0.006)                              &program (0.012)                              &printer (0.011)                                 &card (0.008)                                  &bit (0.012) \\
                                 scsi (0.006)                              &drivers (0.011)                              &windows (0.007)                              &waiting (0.008)                               &system (0.011) \\
                              article (0.004)                               &version (0.010)                               &serial (0.007)                              &program (0.007)                            &encryption (0.010) \\
                                jesus (0.004)                                  &file (0.010)                                  &fax (0.007)                              &drivers (0.006)                              &clipper (0.009) \\
                                bible (0.004)                                 &site (0.009)                                 &data (0.006)                              &printer (0.006)                            &algorithm (0.007) \\
                              archive (0.004)                                  &dos (0.007)                                  &mac (0.006)                                &disks (0.005)                           &government (0.007) \\
                              manager (0.004)                             &graphics (0.007)                             &software (0.006)                                 &post (0.005)                               &secret (0.007) \\
\hline \hline
\end{tabular}
\caption{Top 10 words from the five most frequently occurring topics as estimated by the full spectral method.}\label{tab:sldatopic}
\end{table}

We note that the full spectral LDA method produces a topic that assigns strictly positive probability to only three words: ``faq," ``modified" and ``archive". This topic occurs with estimated probability 0.057\%. In comparison, among the topics estimated by the extreme-value method, the topic with minimal support assigns 95 words a non-zero probability of occurrence. The reason the full spectral LDA method produces such a topic is due to a subtlety in the pre-processing of the twenty newsgroups corpus. 

Observe that a typical document in the corpus has a structure as given in Table \ref{tab:dstructure}. In particular, it has a header (in blue) containing user and subject information, a body (in black) containing the main post text and a signature (in violet) that follows the \texttt{--} string. Section 9.1 of \cite{silge2017text} strips the header information by finding the first empty line in the document and removes the signature by detecting the first line starting with the \texttt{--} string. While this strategy is effective for most documents, the twenty newsgroups corpus also contains some FAQ documents that have the structure presented in Table \ref{tab:dstructure}. Note that the FAQ documents contain a sub-header with the words ``archive" and ``modified," often followed by the word ``faq". This sub-header is not stripped using the pre-processing steps outlined in \cite{silge2017text}. Since document contributions to the spectral LDA-estimated moments are normalized by document length, these three words can contribute substantial signal in short documents. Since the extreme-value based method only uses long documents however, these effects are subdued and such a topic is not recovered. 

Since the words ``faq," ``modified" and ``archive" are artifacts of the FAQ sub-header, they are not of any inferential use beyond uncovering this formatting in the corpus. Since the variability in document formats is quite large in this and many other corpora, the robustness of the extreme-value method to minor and difficult to detect nuances in the data format is attractive. After all, these choices are often under-discussed or omitted when data is provided.

\begin{table}
\centering
\begin{tabular}{ | p{7.5cm} | p{7.5cm} | }
\hline 
\multicolumn{1}{|c|}{Typical document} & \multicolumn{1}{|c|}{FAQ document} \\
\hline 
\textcolor{blue}{\texttt{From:}}  \newline 
\textcolor{blue}{\texttt{Subject:}}  \newline 
\textcolor{blue}{\texttt{Organization:}}  \newline 
\textcolor{blue}{\texttt{Lines:}}  \newline 

\texttt{Document body containing post text.} \newline

\textcolor{violet}{\texttt{--}} \newline
\textcolor{violet}{\texttt{Signature}}
& 
\textcolor{blue}{\texttt{From:}}  \newline 
\textcolor{blue}{\texttt{Subject:}}  \newline 
\textcolor{blue}{\texttt{Organization:}}  \newline 
\textcolor{blue}{\texttt{Lines:}}  \newline 

\texttt{\textcolor{red}{Archive-name:} often contains \textcolor{red}{"faq"}}  \newline 
\textcolor{red}{\texttt{Last-modified:}}  \newline 

\texttt{Document body containing post text.} \newline

\textcolor{violet}{\texttt{--}} \newline
\textcolor{violet}{\texttt{Signature}}
\\
\hline
\end{tabular}
\caption{Document structures in the twenty newsgroups corpus.}\label{tab:dstructure}
\end{table}

\section{Conclusion}\label{sec:conc}

In this paper, we demonstrated that the Latent Dirichlet Allocation model can accommodate power-law word frequencies when the document lengths are power-law distributed. Using this fact, we developed a topic matrix estimation procedure based on normalized extreme word frequencies. Compared to previously developed tensor decomposition methodology, the proposed procedure only leverages the extreme word frequencies to assist computational efficiency and derive simpler moment estimates. The analysis of the twenty newsgroups corpus indicates that the methodology is robust to certain subtleties in the pre-processing of text data.

We also note that our asymptotic analysis treats the vocabulary size $W$ and the number of topics $K$ as fixed while the number of documents grows. In modern corpora, including the twenty newsgroups corpus analyzed in Section \ref{sec:20ng}, the vocabulary size may be large relative to the number of documents, and the constants in our rates of convergence may grow with $W$ and $K$. Extending the present theory to regimes in which $W$ and $K$ grow with $D$, as has been done for other topic model estimators \cite{ke2024using}, is an important direction for future work.

Our results raise many avenues for future inquiry. Traditional topic models, like LDA, are unable to model correlation between topics due to the lack of sufficient dependence between Dirichlet document-topic proportions. Correlated topic models aim to remedy this drawback of the LDA model \cite{blei2006correlated, blei2007correlated}. Since tensor decomposition methods have been developed for correlated topic models, extending our methodology to this setting seems promising \cite{arabshahi2017spectral}. Further, while we prove multivariate regular variation of word/term frequencies, it is of interest to determine the extreme-value behavior of the term frequency-inverse document frequency (TF-IDF)  \cite{sparck1972statistical}. The TF-IDF is a popular metric used in document classification and clustering, and extreme-value procedures may assist such tasks. Lastly, it is valuable to determine how extreme-value procedures may interface with neural models of text such as neural topic models \cite{grootendorst2022bertopic, srivastava2017autoencoding}.

\section*{Acknowledgments}

D. Cirkovic gratefully acknowledges the Northwestern Mutual Data Science Institute Horizon Award. T. Wang gratefully acknowledges the National Natural Science Foundation of China under Grant 12301660 and the Science and Technology Commission of Shanghai Municipality under Grant 23JC1400700.

\newpage
\appendix
\section*{Appendix}

\counterwithin{lemma}{section}
\counterwithin{theorem}{section}

\section{Supporting technical results}\label{sec:tech}

In this section we record the arguments used in the proofs of the main results. Section \ref{sec:gbt} gathers some technical results used in proving the multivariate regular variation of various word frequency statistics (Theorem \ref{thm:freqmrv}). Sections \ref{sec:mrvproof} and \ref{sec:coproof} prove Theorem \ref{thm:freqmrv} and Lemma \ref{lem:cooccur}, respectively. The statistical results are presented in the remaining sections. In particular, Section \ref{sec:mproof} confirms that $\Sigma$ and $\mathcal{T}$ have the desired forms as presented in Lemma \ref{lem:moments}. Section \ref{sec:edmproof} proves the asymptotic normality of the moment estimators presented in \eqref{eq:s1}, \eqref{eq:s2} and \eqref{eq:s3}  (Theorem \ref{thm:EDM}), thereby providing a rate of convergence for $\hat{\tSigma}_D$ and $\hat{\tmathcalT}_D$ as presented in Lemma \ref{lem:Op} and proved in Section \ref{sec:Opproof}. The proof of Theorem \ref{thm:roc}, along with supporting lemmas, is provided in Section \ref{sec:rocproof}.

\subsection{Regular variation of word frequency statistics}\label{sec:gbt}

In order to prove MRV of various word frequency statistics, we rely on a generalized Breiman's theorem \cite{wang2022asymptotic, janssen2025multivariate}. We restate the result here for convenience. 

\begin{theorem}
\label{thm:Breiman}
Suppose $\{Y(t): t \geq 0 \}$ is an at least one-sidedly continuous stochastic process taking on values in $\mathbb{R}_+^p$ for some $p \geq 1$. Let $X$ be a positive random variable with regularly varying distribution satisfying for some scaling function $b(t) \rightarrow \infty$ as $t \rightarrow \infty$
\begin{align*}
 t\mathbb{P}\left(X/b(t) > x \right) \rightarrow \nu_c((x, \infty)),
\end{align*}
where $\nu_c$ is the measure determined by $\nu_c((x, \infty)) = x^{-c}, c > 0, x > 0$. Further suppose that
\begin{enumerate}
\item[(a)] For some finite random vector $Y_\infty \in \mathbb{R}_+^p \setminus \{0\}$, $\lim_{t \rightarrow \infty} Y(t) = Y_\infty$ almost surely.
\item[(b)] The random variable $X$ and the process $Y(\cdot)$ are independent.  
\end{enumerate}
Then
\begin{enumerate}
\item In $\mathbb{M}(\mathbb{R}^p_+ \times (\mathbb{R}_+ \setminus \{0\}))$,
\begin{align*}
t\mathbb{P}\left[\left( Y(X), \frac{X}{b(t)} \right) \in \cdot \right] \rightarrow \mathbb{P}(Y_\infty \in \cdot) \times \nu_c(\cdot).
\end{align*}
\item If for some $c^\star > c$ we have that
\begin{align*}
\sup_{t \geq 0} \mathbb{E}\left[\| Y(t) \|^{c^\star} \right] < \infty,
\end{align*}
then in $\mathbb{M}(\mathbb{R}_+^p \setminus \{0\})$,
\begin{align*}
t\mathbb{P}\left[\frac{XY(X)}{b(t)} \in \cdot \right] \rightarrow \left(\mathbb{P}(Y_\infty \in \cdot) \times \nu_c\right) \circ g^{-1},
\end{align*}
where $g: \mathbb{R}^p_+ \times (\mathbb{R}_+ \setminus \{0\}) \rightarrow \mathbb{R}^p_+$ is defined by $g(x, y) = xy$.
\end{enumerate}
\end{theorem}

Multivariate regular variation of the word frequencies in Theorem \ref{thm:freqmrv} is derived from the following lemma. It is important to remark that the independence imposed in Lemma~\ref{lem:multrv} concerns the heavy-tailed random sample size and the underlying categorical process. The categorical draws remain dependent on the latent probability vector $\rho$. Note that in proofs of (a) and (b) in Theorem \ref{thm:freqmrv}, we employ $\rho = \Phi \theta_d$ and $\rho = \theta_d$, respectively. While the LDA model assumes $\theta_d$ is Dirichlet distributed, more general distributions on the simplex can be considered.

\begin{lemma}
\label{lem:multrv}
Fix $p \geq 1$. Suppose that $\rho$ is a random vector with distribution $\pi$ on $\Delta^{p - 1}$. Given $\rho$, let $C_1, C_2, \dots$ be conditionally independent random variables such that
\begin{align*}
C_i \mid \rho \sim \text{Categorical}(\rho), \qquad \text{for i = 1, 2, \dots.}
\end{align*}
For $n \geq 1$, define the frequencies
\begin{align*}
(N_1(n), \dots, N_p(n)) = \sum_{i = 1}^n \left(1_{\left\lbrace C_i = 1 \right\rbrace}, \dots, 1_{\left\lbrace C_i = p \right\rbrace} \right).
\end{align*}
Let $T \in \mathbb{Z}_+$ be a random variable independent of $\rho, C_1, C_2, \dots$ such that for some scaling function $b(t) \in RV_{1/c}$ with $c > 0$
\begin{align}
\label{eq:Trv}
\lim_{t \rightarrow \infty} t\mathbb{P}\left(T/b(t) > x \right) = x^{-c}, \qquad x > 0.
\end{align}
Then
\begin{align}
\label{eq:multrv2}
(N_1(T), \dots, N_p(T)) \in \text{MRV}(c, b(t), \mu, \mathbb{R}^p_+ \setminus \{0\}),
\end{align}
where $\mu \in\mathbb{M}(\mathbb{R}_+^p \setminus \{0\})$ satisfies for any $f \in \mathcal{C}(\mathbb{R}_+^p \setminus \{0\})$
\begin{align*}
\mu(f) = \int_0^\infty \mathbb{E}\left[f(y\rho)\right] \nu_c(dy).
\end{align*}
Hence in $\mathbb{M}\left(\Delta^{p-1} \times (\mathbb{R}_+ \setminus \{0 \}) \right)$
\begin{align}
\label{eq:multrv}
t\mathbb{P}\left[ \left(\left(\frac{N_1(T)}{T}, \dots, \frac{N_p(T)}{T}\right), \frac{T}{b(t)} \right) \in \cdot \right] \rightarrow \left(\pi \times \nu_c\right)(\cdot), \qquad \text{as } t \rightarrow \infty.
\end{align}
\end{lemma}

\begin{proof}
We aim to apply the generalized Breiman's theorem (see Theorem \ref{thm:Breiman}). Let $\{Y(t) : t \geq 0\}$ be the stochastic process defined by $Y(t) = 0$ for $0 \leq t < 1$ and
\begin{align*}
Y(t) = \left(\frac{N_1(\floor{t})}{t}, \dots, \frac{N_p(\floor{t})}{t} \right)
\end{align*} 
for $t \geq 1$. Note that for $t \geq 1$
\begin{align*}
Y(t) = \frac{\floor{t}}{t} \left(\frac{N_1(\floor{t})}{\floor{t}}, \dots, \frac{N_p(\floor{t})}{\floor{t}} \right),
\end{align*}
Given $\rho$, $(N_1(n), \dots, N_p(n))$ is the sum of $n$ conditionally independent multinomial random vectors. Hence, we may apply the strong law of large numbers to see that $Y(t) \xrightarrow{\text{a.s.}} \rho$ as $t \rightarrow \infty$ \cite[see][for more details]{aldous2006exchangeability, kallenberg2021symmetric}. Hence, condition (a) of Theorem \ref{thm:Breiman} is satisfied. Condition (b) is satisfied by assumption. Lastly, the moment condition is trivially satisfied since $\| Y(t) \|_1 = 0$ for $t < 1$ and $\| Y(t) \|_1 = \floor{t}/t \leq 1$ for $t \geq 1$. Since the distribution of $T$ satisfies \eqref{eq:Trv}, an application of Theorem \ref{thm:Breiman} gives \eqref{eq:multrv2}. The convergence in polar coordinates follows from Corollary 2.1 of \cite{resnick2024art}.
\end{proof}

\subsection{Proof of Theorem \ref{thm:freqmrv}}\label{sec:mrvproof}

\begin{proof}[Proof of Theorem \ref{thm:freqmrv}]
For item (a), note that 
\begin{align*}
\left(N_{d\cdot 1}, \dots, N_{d\cdot W} \right) = \sum_{i = 1}^{N_{d\cdot\cdot}} \left(1_{\left\lbrace X_{di} = 1 \right\rbrace}, \dots, 1_{\left\lbrace  X_{di} = W \right\rbrace} \right),
\end{align*}
where, conditional on $p_d = \Phi \theta_d$, the $X_{di}$ are drawn independently over $i$ from a common categorical distribution with probability vector $p_d$. Since the distribution of $N_{d\cdot\cdot}$ satisfies \eqref{eq:lengthrv}, Lemma \ref{lem:multrv} applies.

For item (b), similarly notice that
\begin{align*}
\left(N_{d1 \cdot}, \dots, N_{dK \cdot} \right) = \sum_{i = 1}^{N_{d\cdot\cdot}} \left(1_{\left\lbrace Z_{di} = 1 \right\rbrace}, \dots, 1_{\left\lbrace  Z_{di} = K \right\rbrace} \right),
\end{align*}
where, conditional on $\theta_d$, the $Z_{di}$ are drawn independently over $i$ from a common categorical distribution with probability vector $\theta_d$. Hence, the conclusion follows similarly.

We now prove item (c). For fixed $k \in \{1, \dots, K\}$, consider the vector
\begin{align*}
\left( N_{\cdot k 1}, \dots, N_{\cdot k W} \right) = \sum_{d = 1}^{D} \left( N_{d k 1}, \dots, N_{d k W} \right).
\end{align*}
Since $\left( N_{\cdot k 1}, \dots, N_{\cdot k W} \right)$ is the sum of $D$ independent vectors, it suffices to prove multivariate regular variation of $\left( N_{d k 1}, \dots, N_{d k W} \right)$ for fixed $d \in \{1, \dots, D\}$ \cite[see Proposition 7.4 of][for example]{resnick2007heavy}. Towards that end, note that  
\begin{align*}
\left( N_{d k 1}, \dots, N_{d k W} \right) = \sum_{i = 1}^{N_{d\cdot\cdot}} \left(1_{\left\lbrace X_{di} = 1, Z_{di} = k \right\rbrace}, \dots, 1_{\left\lbrace  X_{di} = W, Z_{di} = k \right\rbrace} \right).
\end{align*}
For fixed $k \in \{1, \dots, K\}$, the previous display is a sum of $N_{dk\cdot} = \sum_{i = 1}^{N_{d\cdot\cdot}} 1_{\left\lbrace Z_{di} = k\right\rbrace}$ one-hot encoded categorical random vectors with probability vector $\phi_k$. Hence, by Lemma \ref{lem:multrv}, it suffices to show that $N_{dk\cdot}$ is univariate regularly varying. In order to do so, we evaluate the convergence 
\begin{align*}
t\mathbb{P}\left[\left( N_{d1\cdot}, \dots, N_{dK\cdot} \right)/b(t) \in \cdot \right] \rightarrow \eta(\cdot), \qquad \text{in } \mathbb{M}(\mathbb{R}_+^K \setminus \{0\}),
\end{align*}
as $t \rightarrow \infty$ on sets of the form $(0, \infty) \times \cdots \times (z, \infty) \times \cdots  \times (0, \infty)$, $z > 0$. Doing so, we obtain that
\begin{align*}
t\mathbb{P}\left(N_{dk\cdot} > b(t)z \right) \rightarrow \int_0^\infty \mathbb{P}\left(y \theta_{dk} > z\right)\nu_\iota(dy) \equiv \eta^\star((z, \infty)).
\end{align*}
Using the fact that for $X \sim \text{Beta}(a, b)$, $B(a, b)P(X > x) = \int_{x}^1 r^{a - 1}(1 - r)^{b - 1}dr$ for $x \in (0, 1)$ and $P(X > x) = 0$ for $x \geq 1$, we have that
\begin{align*}
\eta^\star((z, \infty)) =& \frac{1}{B(\alpha_k, \alpha_0 - \alpha_k)} \int_z^\infty \left(\int_{z/y}^1 r^{\alpha_k - 1}(1 - r)^{\alpha_0 - \alpha_k - 1}dr\right)\nu_\iota(dy) \\
=& \frac{1}{B(\alpha_k, \alpha_0 - \alpha_k)} \int_0^1 (r/z)^{\iota} r^{\alpha_k - 1}(1 - r)^{\alpha_0 - \alpha_k - 1}dr \\
=& z^{-\iota} \frac{B(\alpha_k + \iota, \alpha_0 - \alpha_k)}{B(\alpha_k, \alpha_0 - \alpha_k)}.
\end{align*}
Hence, defining 
\begin{align*}
b_k(t) = b(t)\left(  \frac{B(\alpha_k + \iota, \alpha_0 - \alpha_k)}{B(\alpha_k, \alpha_0 - \alpha_k)} \right)^{1/\iota},
\end{align*}
we have that $N_{dk\cdot}$ satisfies \eqref{eq:lengthrv} with scaling function $b_k(t)$. Hence
\begin{align*}
\left( N_{d k 1}, \dots, N_{d k W} \right) \in \text{MRV}(\iota, b_k(t), \upsilon, \mathbb{R}^W_+ \setminus \{0\}),
\end{align*}
with $\upsilon$ as defined in the theorem statement. Thus, applying Proposition 7.4 of \cite{resnick2007heavy} gives that
\begin{align*}
\left( N_{\cdot k 1}, \dots, N_{\cdot k W} \right) \in \text{MRV}(\iota, \tilde{b}_k(t), \upsilon, \mathbb{R}^W_+ \setminus \{0\}),
\end{align*}
with $\tilde{b}_k(t) = D^{1/\iota} b_k(t)$. The convergence in terms of polar coordinates follows from Corollary 2.1 of \cite{resnick2024art}.
\end{proof}

\subsection{Proof of Lemma \ref{lem:cooccur}}\label{sec:coproof}

\begin{proof}[Proof of Lemma \ref{lem:cooccur}]
It suffices to prove the statement for $u = 1$. We prove the statement via a variant of Breiman's Theorem found in Proposition 7.6 of \cite{resnick2007heavy} and page 89 of \cite{resnick2024art}. We first verify that
\begin{align*}
\mathbb{P}\left( N_{d\cdot 1} \in \cdot \right) \in \text{MRV}(\iota, cb(t), \nu_\iota, \mathbb{R}_+\setminus \{0\}),
\end{align*}
where $c = \left(\mathbb{E}\left[ p^{\iota}_{d1}\right]\right)^{1/\iota}$. Observe that evaluating the convergence in (a) of Theorem \ref{thm:freqmrv} on sets of the form $(z, \infty) \times (0, \infty) \times \cdots  \times (0, \infty)$, $z > 0$ gives that
\begin{align*}
t\mathbb{P}\left(N_{d\cdot 1} > b(t)z\right) \rightarrow \int_0^\infty \mathbb{P}\left(y p_{d1} > z \right) \nu_\iota(dy) \equiv \mu^\star((z, \infty)).
\end{align*}
Indeed
\begin{align*}
\mu^\star((z, \infty)) = \int_z^\infty \mathbb{P}\left(y p_{d1} > z \right) \nu_\iota(dy) = \mathbb{E}\left[\int_z^\infty 1_{\left\lbrace y p_{d1} > z \right\rbrace} \nu_\iota(dy)\right] = z^{-\iota}\mathbb{E}\left[ p^{\iota}_{d1}\right].
\end{align*}
We now show that in $\mathbb{M}\left(\mathbb{R}_+^{W + 1}  \setminus \{0 \} \right)$, as $t \rightarrow \infty$
\begin{align*}
t\mathbb{P}\left(\left(N_{d\cdot 1}/c, N_{d\cdot 1}, \dots, N_{d\cdot W}\right)/b(t) \in \cdot \right) \rightarrow \psi(\cdot),
\end{align*}
where the limit measure is defined for $f \in \mathcal{C}(\mathbb{R}_+^{W + 1}  \setminus \{0 \})$ by
\begin{align*}
\psi(f) = \int_0^\infty \mathbb{E}\left[f\left(y \left(\frac{p_{d1}}{c},  p_d\right)\right)\right] \nu_\iota(dy).
\end{align*}
Define the map $h: \mathbb{R}^W_+ \rightarrow \mathbb{R}^{W + 1}_+$ by
\begin{align*}
h: (x_1, \dots, x_W) \mapsto (x_1/c, x_1, \dots, x_W). 
\end{align*}
The map is uniformly continuous and maps $0 \in \mathbb{R}^W$ to $0 \in \mathbb{R}^{W+1}$. Thus Theorem 1.3 of \cite{resnick2024art} applies and hence 
\begin{align*}
t\mathbb{P}\left(\left(N_{d\cdot 1}/c, N_{d\cdot 1}, \dots, N_{d\cdot W}\right)/b(t) \in \cdot \right) = t\mathbb{P}\left(\left(N_{d\cdot 1}, \dots, N_{d\cdot W}\right)/b(t) \in \cdot \right) \circ h^{-1} \rightarrow \mu \circ h^{-1} = \psi.
\end{align*}
The limit measure $\psi$ clearly lacks asymptotic independence (i.e. satisfies (7.41) of \cite{resnick2007heavy}) and thus we may apply Proposition 7.6 of \cite{resnick2007heavy} to achieve that
\begin{align*}
\mathbb{P}\left(\left(N_{d\cdot 1}^2, \dots, N_{d\cdot 1}N_{d\cdot W}\right) \in \cdot \right) \in \text{MRV}\left(\iota/2, cb^2(t), \underaccent{\tilde}{\zeta}, \mathbb{R}_+^{W}  \setminus \{0 \}\right),
\end{align*}
with $\underaccent{\tilde}{\zeta}$ satisfying
\begin{align*}
\underaccent{\tilde}{\zeta}(f) = \int_0^\infty \mathbb{E}\left[f\left(y^2 \frac{p_{d1}}{c} p_d\right)\right] \nu_\iota(dy) = \int_0^\infty \mathbb{E}\left[f\left(y \frac{p_{d1}}{c} p_d\right)\right] \nu_{\iota/2}(dy),
\end{align*}
for $f \in \mathcal{C}(\mathbb{R}_+^{W}  \setminus \{0 \})$. The result follows.
\end{proof}

\subsection{Moments of $p_d$}\label{sec:mproof}

In this section we compute the moments of $p_d$ and show that $\Sigma$ and $\mathcal{T}$ have the form presented in Lemma \ref{lem:moments}. The computation of these moments relies on the following lemma from \cite{anandkumar2012spectral}.
\begin{lemma}
\label{lem:dirmoments}
Let $e_1, \dots, e_K$ denote the canonical vectors in $\mathbb{R}^K$. Further let $\alpha \in \mathbb{R}^K$ consist of the elements $\alpha_1, \dots, \alpha_K$ and let $\alpha_0 = \sum_{k = 1}^K \alpha_k$. For $\theta_d \sim \text{Dirichlet}(\alpha_1, \dots, \alpha_K)$, we have that
\begin{align*}
\alpha_0 \mathbb{E}[\theta_d] =& \sum_{k = 1}^K \alpha_ke_k, \\
\alpha_0(\alpha_0 + 1) \mathbb{E}[\theta_d \otimes \theta_d] =& \sum_{k = 1}^K \alpha_k e_k \otimes e_k + \alpha \otimes \alpha,
\end{align*}
and for the third-order tensor, 
\begin{align*}
\alpha_0(\alpha_0 + 1)(\alpha_0 + 2) \mathbb{E}[\theta_d \otimes \theta_d \otimes \theta_d] =&2\sum_{k = 1}^K \alpha_k (e_k \otimes e_k \otimes e_k) +  \sum_{k = 1}^K \alpha_k (e_k \otimes e_k \otimes \alpha) \\
&+ \sum_{k = 1}^K \alpha_k(\alpha \otimes e_k \otimes e_k) + \sum_{k = 1}^K \alpha_k(e_k \otimes \alpha \otimes e_k) + \alpha \otimes \alpha \otimes \alpha .
\end{align*}
\end{lemma}  

\begin{proof}[Proof of Lemma \ref{lem:moments}]
For $\Sigma$, see that
\begin{align*}
\mathbb{E}\left[p_d \otimes p_d \right] =& \mathbb{E}\left[(\Phi \theta_d)(\Phi \theta_d)'\right] \\
=& \Phi \mathbb{E}[\theta_d \otimes \theta_d] \Phi' \\
=& \frac{1}{\alpha_0(\alpha_0 + 1)}\Phi \left( \sum_{k = 1}^K \alpha_k e_k \otimes e_k + \alpha \otimes \alpha \right) \Phi' \\
=& \frac{1}{\alpha_0(\alpha_0 + 1)}\left( \sum_{k = 1}^K \alpha_k (\phi_k \otimes \phi_k) + \alpha_0^2 (M \otimes M) \right).
\end{align*}
Rearranging gives the form of $\Sigma$. For $\mathcal{T}$, see that from the form of $\mathbb{E}[\theta_d \otimes \theta_d \otimes \theta_d]$ in Lemma \ref{lem:dirmoments}
\begin{align*}
\alpha_0&(\alpha_0 + 1)(\alpha_0 + 2)\mathbb{E}\left[p_d \otimes p_d \otimes p_d \right] \\
=& \alpha_0(\alpha_0 + 1)(\alpha_0 + 2) \mathbb{E}\left[(\Phi \theta_d) \otimes  (\Phi \theta_d) \otimes (\Phi \theta_d) \right] \\
=& 2\sum_{k = 1}^K \alpha_k (\phi_k \otimes \phi_k \otimes \phi_k) +  \sum_{k = 1}^K \alpha_k \alpha_0(\phi_k \otimes \phi_k \otimes M) + \sum_{k = 1}^K \alpha_k\alpha_0 (M \otimes \phi_k \otimes \phi_k) \\
&+ \sum_{k = 1}^K \alpha_k \alpha_0 (\phi_k \otimes M \otimes \phi_k) + \alpha_0^3 (M \otimes M \otimes M) \\
=& 2\sum_{k = 1}^K \alpha_k (\phi_k \otimes \phi_k \otimes \phi_k) + \alpha_0^3 (M \otimes M \otimes M) \\
&+ \alpha_0\left(\alpha_0(\alpha_0 + 1)\mathbb{E}\left[p_d \otimes M \otimes p_d \right] - \alpha_0^2 (M \otimes M \otimes M) \right) \\
&+ \alpha_0\left(\alpha_0(\alpha_0 + 1)\mathbb{E}\left[M \otimes p_d  \otimes p_d \right] - \alpha_0^2 (M \otimes M \otimes M) \right) \\
&+ \alpha_0\left(\alpha_0(\alpha_0 + 1)\mathbb{E}\left[p_d  \otimes p_d \otimes M \right] - \alpha_0^2 (M \otimes M \otimes M) \right) \\
=& 2\sum_{k = 1}^K \alpha_k (\phi_k \otimes \phi_k \otimes \phi_k) + \alpha_0^2(\alpha_0 + 1)\mathbb{E}\left[p_d \otimes p_d \otimes M \right] \\
&+ \alpha_0^2(\alpha_0 + 1)\mathbb{E}\left[p_d \otimes M \otimes p_d \right] + \alpha_0^2(\alpha_0 + 1)\mathbb{E}\left[M \otimes p_d  \otimes p_d \right] - 2\alpha_0^3 (M \otimes M \otimes M).
\end{align*}
Rearranging gives the expression for $\mathcal{T}$.
\end{proof}

\subsection{Proof of Theorem \ref{thm:EDM}}\label{sec:edmproof}

\begin{proof}[Proof of Theorem \ref{thm:EDM}]
We only prove \eqref{eq:norm3} since the other two statements are similar. We aim to apply Theorem 1 of \cite{larsson2012extremal}. This theorem is provided for a 2-dimensional angular component. As noted in Theorem 4 of \cite{kluppelberg2021estimating}, a line-by-line reading of the proof gives asymptotic normality of $\hat{\mathbb{E}}_D[p_{du}p_{dv}p_{dw}]$ given that
\begin{align*}
\lim_{D \rightarrow \infty} \sqrt{k_D}\left( \frac{D}{k_D}\mathbb{E}\left[\omega_{du}\omega_{dv}\omega_{dw} 1_{\left\lbrace  R_d \geq b(D/k_D)t^{-1/\iota} \right\rbrace} \right] - \mathbb{E}\left[p_{du}p_{dv}p_{dw} \right] \frac{D}{k_D} \mathbb{P}\left(R_{d} > b(D/k_D)t^{-1/\iota} \right) \right) = 0,
\end{align*}
holds locally uniformly for $t \in [0, \infty)$. Unpacking notation, we have that
\begin{align*}
\mathbb{E}\left[\omega_{du}\omega_{dv}\omega_{dw} 1_{\left\lbrace  R_d \geq b(D/k_D)t^{-1/\iota} \right\rbrace} \right] = \mathbb{E}\left[\frac{N_{d\cdot u}}{N_{d\cdot \cdot}}\frac{N_{d\cdot v}}{N_{d\cdot \cdot}}\frac{N_{d\cdot w}}{N_{d\cdot \cdot}} 1_{\left\lbrace N_{d\cdot \cdot} \geq b(D/k_D)t^{-1/\iota} \right\rbrace} \right].
\end{align*}
We aim to compute the moments $\mathbb{E}[N_{d\cdot u}N_{d\cdot v}N_{d\cdot w} \mid N_{d\cdot\cdot}]$. Using tower property, 
\begin{align*}
\mathbb{E}\left[N_{d\cdot u}N_{d\cdot v}N_{d\cdot w} \mid N_{d\cdot\cdot} \right] = \mathbb{E}\left[ \mathbb{E}\left[N_{d\cdot u}N_{d\cdot v}N_{d\cdot w} \mid N_{d\cdot\cdot}, \theta_d \right] \mid N_{d\cdot\cdot} \right],
\end{align*}
where the inner expectation may be obtained using factorial moment formulas from \cite{mosimann1962compound}. By independence of $\theta_d$ and $N_{d\cdot\cdot}$, we thus obtain the moments in three cases:
\begin{enumerate}
\item when $u = v = w$
\begin{align*}
\mathbb{E}\left[N_{d\cdot u}N_{d\cdot v}N_{d\cdot w} \mid N_{d\cdot\cdot} \right] = N_{d\cdot\cdot}(N_{d\cdot\cdot} - 1)(N_{d\cdot\cdot} - 2)\mathbb{E}[p^3_{du}] + 3 N_{d\cdot\cdot}(N_{d\cdot\cdot} - 1)\mathbb{E}[p^2_{du}] + N_{d\cdot\cdot}\mathbb{E}[p_{du}],
\end{align*}
\item when $u = v \neq w$
\begin{align*}
\mathbb{E}\left[N_{d\cdot u}N_{d\cdot v}N_{d\cdot w} \mid N_{d\cdot\cdot} \right] =&  N_{d\cdot\cdot}(N_{d\cdot\cdot} - 1)(N_{d\cdot\cdot} - 2)\mathbb{E}\left[p^2_{du}p_{dw}\right]+  N_{d\cdot\cdot}(N_{d\cdot\cdot} - 1)\mathbb{E}\left[p_{du}p_{dw}\right],
\end{align*}
\item and when $u,v,w$ distinct
\begin{align*}
\mathbb{E}\left[N_{d\cdot u}N_{d\cdot v}N_{d\cdot w} \mid N_{d\cdot\cdot} \right] =&  N_{d\cdot\cdot}(N_{d\cdot\cdot} - 1)(N_{d\cdot\cdot} - 2)\mathbb{E}\left[p_{du}p_{dv}p_{dw}\right].
\end{align*}
\end{enumerate}
Since $|  N_{d\cdot\cdot}(N_{d\cdot\cdot} - 1)(N_{d\cdot\cdot} - 2) - N^3_{d\cdot\cdot}| = |-3N^2_{d\cdot\cdot} + 2N_{d\cdot\cdot}| \leq 5 N^2_{d\cdot\cdot}$, we have that for all cases
\begin{align*}
\left|N^{-3}_{d\cdot\cdot}\mathbb{E}\left[N_{d\cdot u}N_{d\cdot v}N_{d\cdot w} \mid N_{d\cdot\cdot} \right] - \mathbb{E}\left[p_{du}p_{dv}p_{dw}\right]\right| \leq \frac{9}{N_{d\cdot\cdot}}.
\end{align*}
Hence
\begin{align*}
&\left|\sqrt{k_D}\left( \frac{D}{k_D}\mathbb{E}\left[\omega_{du}\omega_{dv}\omega_{dw} 1_{\left\lbrace  R_d \geq b(D/k_D)t^{-1/\iota} \right\rbrace} \right] - \mathbb{E}\left[p_{du}p_{dv}p_{dw} \right] \frac{D}{k_D} \mathbb{P}\left(R_{d} > b(D/k_D)t^{-1/\iota} \right) \right) \right| \\
&= \left|\sqrt{k_D} \frac{D}{k_D}\mathbb{E}\left[\left(N^{-3}_{d\cdot\cdot} \mathbb{E}\left[N_{d\cdot u}N_{d\cdot v}N_{d\cdot w} \mid N_{d\cdot\cdot} \right]  - \mathbb{E}\left[p_{du}p_{dv}p_{dw} \right]\right)1_{\left\lbrace  N_{d\cdot\cdot} \geq b(D/k_D)t^{-1/\iota} \right\rbrace} \right] \right| \\
&\leq \sqrt{k_D} \frac{D}{k_D}\mathbb{E}\left[\left|N^{-3}_{d\cdot\cdot} \mathbb{E}\left[N_{d\cdot u}N_{d\cdot v}N_{d\cdot w} \mid N_{d\cdot\cdot} \right]  - \mathbb{E}\left[p_{du}p_{dv}p_{dw} \right]\right|1_{\left\lbrace  N_{d\cdot\cdot} \geq b(D/k_D)t^{-1/\iota} \right\rbrace} \right]\\
&\leq 9 \sqrt{k_D} \frac{D}{k_D} \mathbb{E}\left[\frac{1_{\left\lbrace N_{d\cdot \cdot} \geq b(D/k_D)t^{-1/\iota} \right\rbrace}}{N_{d\cdot \cdot}} \right].
\end{align*}
It thus suffices to show that 
\begin{align*}
\lim_{D \rightarrow \infty} \sqrt{k_D} \frac{D}{k_D} \mathbb{E}\left[\frac{1_{\left\lbrace N_{d\cdot \cdot} \geq b(D/k_D)t^{-1/\iota} \right\rbrace}}{N_{d\cdot \cdot}} \right] = 0,
\end{align*}
locally uniformly. Naive bounds give that
\begin{align*}
\sqrt{k_D} \frac{D}{k_D} \mathbb{E}\left[\frac{1_{\left\lbrace N_{d\cdot \cdot} \geq b(D/k_D)t^{-1/\iota} \right\rbrace}}{N_{d\cdot \cdot}} \right] \leq& t^{1/\iota} \frac{\sqrt{k_D}}{b(D/k_D)} \frac{D}{k_D}\mathbb{P}\left( N_{d\cdot \cdot} \geq b(D/k_D)t^{-1/\iota}\right).
\end{align*}
By regular variation, as $D \rightarrow \infty$
\begin{align*}
t^{1/\iota} \frac{D}{k_D}\mathbb{P}\left( N_{d\cdot \cdot} \geq b(D/k_D)t^{-1/\iota}\right) \rightarrow t^{1 + \frac{1}{\iota}}.
\end{align*}
Since the functions are nondecreasing in $t$ and the limit is continuous, the convergence is locally uniform on $[0, \infty)$ \citep[see Proposition 2.1 of][]{resnick2007heavy}. Hence, given that $\sqrt{k_D}/b(D/k_D) \rightarrow 0$, the proof is complete.
\end{proof}

\subsection{Proof of Lemma \ref{lem:Op}}\label{sec:Opproof}

\begin{proof}[Proof of Lemma \ref{lem:Op}]
We only prove the second claim, since the first is similar. We consider the terms in $\hat{\tmathcalT}_D - \tmathcalT$. For the first term, note that
\begin{align*}
\left\Vert \hat{\mathbb{E}}_D\left[  p_d \otimes p_d \otimes p_d \right] - \mathbb{E}\left[ p_d \otimes p_d \otimes p_d \right] \right\Vert \leq \left\Vert \hat{\mathbb{E}}_D\left[  p_d \otimes p_d \otimes p_d \right] - \mathbb{E}\left[ p_d \otimes p_d \otimes p_d \right] \right\Vert_F.
\end{align*}
By Theorem \ref{thm:EDM}, for any fixed $u, v, w \in \{1, \dots, W\}$,
\begin{align*}
(\hat{\mathbb{E}}_D[p_{du}p_{dv}p_{dw}] - \mathbb{E}[p_{du}p_{dv}p_{dw}])^2 = O_p\left(\frac{1}{k_D}\right).
\end{align*}
Since $\| \hat{\mathbb{E}}_D\left[  p_d \otimes p_d \otimes p_d \right] - \mathbb{E}\left[ p_d \otimes p_d \otimes p_d \right] \|^2_F$ is a finite sum over these terms, we have that
\begin{align*}
\left\Vert \hat{\mathbb{E}}_D\left[  p_d \otimes p_d \otimes p_d \right] - \mathbb{E}\left[ p_d \otimes p_d \otimes p_d \right] \right\Vert \leq \left\Vert \hat{\mathbb{E}}_D\left[  p_d \otimes p_d \otimes p_d \right] - \mathbb{E}\left[ p_d \otimes p_d \otimes p_d \right] \right\Vert_F = O_p\left(\frac{1}{\sqrt{k_D}}\right).
\end{align*}
We now focus on the second and first moment correction terms. See that by triangle inequality
\begin{align*}
\left\Vert \hat{\mathbb{E}}_D\left[  p_d \otimes p_d\right] \otimes \hat{\mathbb{E}}_D\left[  p_d \right] - \mathbb{E}\left[  p_d \otimes p_d \right] \otimes M  \right\Vert \leq& \left\Vert \hat{\mathbb{E}}_D\left[  p_d \otimes p_d\right] \otimes \left( \hat{\mathbb{E}}_D\left[  p_d \right] -  M\right)  \right\Vert \\
&+ \left\Vert \left(\hat{\mathbb{E}}_D\left[  p_d \otimes p_d\right] - \mathbb{E}\left[  p_d \otimes p_d \right]\right) \otimes M  \right\Vert \\
\leq& \left\Vert \hat{\mathbb{E}}_D\left[  p_d \otimes p_d\right] \right\Vert \left\Vert  \hat{\mathbb{E}}_D\left[  p_d \right] -  M  \right\Vert_2 \\
&+ \left\Vert \hat{\mathbb{E}}_D\left[  p_d \otimes p_d\right] - \mathbb{E}\left[  p_d \otimes p_d \right] \right\Vert \left\Vert M \right\Vert_2 \\
&= O_p(1)O_p\left( \frac{1}{\sqrt{k_D}}\right) + O_p\left( \frac{1}{\sqrt{k_D}}\right),
\end{align*}
where the stochastic orders again follow from consistency and element-wise asymptotic normality of $\hat{\mathbb{E}}_D\left[  p_d \otimes p_d\right]$ and $\hat{\mathbb{E}}_D\left[ p_d\right]$. The other second and first moment correction terms follow similarly since they are just permutations of the above term. Now consider the last term in $\hat{\tmathcalT}_D - \tmathcalT$. See that
\begin{align*}
\left\Vert \hat{\mathbb{E}}_D\left[  p_d \right] \otimes \hat{\mathbb{E}}_D\left[  p_d \right] \otimes \hat{\mathbb{E}}_D\left[  p_d \right]  - M \otimes M \otimes M \right\Vert \leq& \left\Vert \hat{\mathbb{E}}_D\left[  p_d \right] \otimes \hat{\mathbb{E}}_D\left[  p_d \right] \otimes \left(\hat{\mathbb{E}}_D\left[  p_d \right] - M\right) \right\Vert \\
&+ \left\Vert \hat{\mathbb{E}}_D\left[  p_d \right] \otimes \left( \hat{\mathbb{E}}_D\left[  p_d \right] - M \right) \otimes M   \right\Vert \\
&+ \left\Vert \left( \hat{\mathbb{E}}_D\left[  p_d \right]  - M \right) \otimes M \otimes M  \right\Vert \\
\leq& \left\Vert \hat{\mathbb{E}}_D\left[  p_d \right] \right\Vert_2^2 \left\lVert \hat{\mathbb{E}}_D\left[  p_d \right]  - M \right\Vert_2 \\
&+ \left\Vert \hat{\mathbb{E}}_D\left[  p_d \right] \right\Vert_2 \left\lVert \hat{\mathbb{E}}_D\left[  p_d \right]  - M \right\Vert_2 \left\lVert M \right\Vert_2 \\
&+ \left\Vert \hat{\mathbb{E}}_D\left[  p_d \right]  - M \right\Vert_2 \left\Vert M \right\Vert_2^2 \\
=& O_p(1)O_p\left( \frac{1}{\sqrt{k_D}}\right) + O_p(1)O_p\left( \frac{1}{\sqrt{k_D}}\right) + O_p\left( \frac{1}{\sqrt{k_D}}\right).
\end{align*}
Again, the stochastic orders follow from asymptotic normality of the empirical spectral measure moments in Theorem \ref{thm:EDM}. Hence the proof is complete.
\end{proof}

\subsection{Proof of Theorem \ref{thm:roc}}\label{sec:rocproof}

We now turn to proving the consistency of Algorithm \ref{alg:specLDA} using $\hat{\Sigma}_D$ and $\hat{\mathcal{T}}_D$ as the adjusted moment estimators. Following \cite{anandkumar2015spectral}, we prove that $\hat{\tmathcalT}_D(\hat{\mathcal{W}}_D, \hat{\mathcal{W}}_D, \hat{\mathcal{W}}_D)$ concentrates around an orthogonally decomposable reference tensor $\tmathcalT(\hat{\mathcal{V}}_D, \hat{\mathcal{V}}_D, \hat{\mathcal{V}}_D)$ where 
\begin{align*}
\hat{\mathcal{V}}_D = \hat{\mathcal{W}}_D(\hat{\mathcal{W}}_D'\tSigma \hat{\mathcal{W}}_D)^{\dagger 1/2}.
\end{align*}

Note that on a high probability event, $\hat{\mathcal{V}}_D = \hat{\mathcal{W}}_D(\hat{\mathcal{W}}_D'\tSigma \hat{\mathcal{W}}_D)^{-1/2}$ (see Lemma \ref{lem:whitening}, claim 3). On this high probability event, observe that $\hat{\mathcal{V}}_D'\tSigma \hat{\mathcal{V}}_D = I$ and thus $\tmathcalT(\hat{\mathcal{V}}_D, \hat{\mathcal{V}}_D, \hat{\mathcal{V}}_D)$ is indeed orthogonally decomposable. Given that $\hat{\tmathcalT}_D(\hat{\mathcal{W}}_D, \hat{\mathcal{W}}_D, \hat{\mathcal{W}}_D)$ concentrates around $\tmathcalT(\hat{\mathcal{V}}_D, \hat{\mathcal{V}}_D, \hat{\mathcal{V}}_D)$, Theorem 5.1 of \cite{anandkumar2014tensor} provides that, up to permutation, the estimated eigenvectors and eigenvalues output by Algorithm \ref{alg:tensorpower} will be near the reference eigenvectors $\hat{\mathcal{V}}_D'\tphi_k$ and eigenvalues $\alpha_k^{-1/2}$, $k = 1,\dots, K$. The last step is to prove that transformation in Step 4 in Algorithm \ref{alg:specLDA} retains consistency. We do not analyze Step 5 of Algorithm \ref{alg:specLDA} since it simply projects the reconstructed topic vector to the simplex.

We first prove concentration of $\hat{\tmathcalT}_D(\hat{\mathcal{W}}_D, \hat{\mathcal{W}}_D, \hat{\mathcal{W}}_D)$ around the reference tensor $\tmathcalT(\hat{\mathcal{V}}_D, \hat{\mathcal{V}}_D, \hat{\mathcal{V}}_D)$. In order to do so, we require some supporting lemmas that establish the properties of the matrix $\hat{\mathcal{W}}_D'\tSigma \hat{\mathcal{W}}_D$. These properties largely follow from Lemma 10 of \cite{hsu2013learning}. Here, we let $\ts^{(w)}$ and $\hat{\ts}_D^{(w)}$ be the $w$-th largest eigenvalues of $\tSigma$ and $\hat{\tSigma}_D$, respectively.

\begin{lemma}
\label{lem:whitening}
Fix $\varepsilon_1 \in (0, \ts^{(K)}/3]$ . On the event
\begin{align*}
E_D(\varepsilon_1) = \{ \|\hat{\tSigma}_D - \tSigma \| \leq \varepsilon_1 \},
\end{align*}
we have that
\begin{enumerate}
\item $\hat{\ts}_D^{(K)} \geq 2\ts^{(K)}/3$.
\item $\left\Vert \hat{\mathcal{W}}_D \right\Vert \leq \sqrt{\frac{3}{2\ts^{(K)}}}$.
\item $\hat{\mathcal{W}}_D'\tSigma\hat{\mathcal{W}}_D$ is positive definite. 
\item $\left\Vert (\hat{\mathcal{W}}_D'\tSigma\hat{\mathcal{W}}_D)^{1/2} \right\Vert \leq \left(1 + \frac{3\varepsilon_1}{2\ts^{(K)}}\right)^{1/2}$.
\item $\left\Vert (\hat{\mathcal{W}}_D'\tSigma\hat{\mathcal{W}}_D)^{1/2} - I \right\Vert \leq \frac{3\varepsilon_1}{2\ts^{(K)}}$.
\item $\left\Vert (\hat{\mathcal{W}}_D'\tSigma\hat{\mathcal{W}}_D)^{-1/2} - I \right\Vert \leq \frac{3\varepsilon_1}{2\ts^{(K)}}$.
\item $\|\hat{U}_D \hat{\Lambda}^{1/2}_D \| \leq \left(\left\Vert \tSigma \right\Vert + \frac{\ts^{(K)}}{3}\right)^{1/2}$.
\end{enumerate}
\end{lemma}

\begin{proof}
To prove the first claim, see that by Weyl's inequality \citep[see Lemma 4.1.14 of][]{vershynin2026high}
\begin{align*}
\hat{\ts}_D^{(K)} \geq \ts^{(K)} - \|\hat{\tSigma}_D - \tSigma \| \geq \ts^{(K)} - \varepsilon_1 \geq 2\ts^{(K)}/3.
\end{align*}
Hence $\hat{\mathcal{W}}_D$ can be constructed using $K$ non-zero eigenvalues. To see the second claim, note that 
\begin{align*}
\left\Vert \hat{\mathcal{W}}_D \right\Vert = \left\Vert \hat{U}_D \hat{\Lambda}^{-1/2}_D \right\Vert = \left\Vert \hat{\Lambda}^{-1/2}_D \right\Vert = \sqrt{\frac{1}{\hat{\ts}_D^{(K)}}} \leq \sqrt{\frac{3}{2\ts^{(K)}}},
\end{align*}
by claim one. For the third claim, rewrite
\begin{align*}
\hat{\mathcal{W}}_D'\tSigma\hat{\mathcal{W}}_D =  \hat{\mathcal{W}}_D' \hat{\tSigma}_D \hat{\mathcal{W}}_D + \hat{\mathcal{W}}_D'\left(\tSigma - \hat{\tSigma}_D \right)\hat{\mathcal{W}}_D = I - \hat{\mathcal{W}}_D'\left(\hat{\tSigma}_D - \tSigma \right)\hat{\mathcal{W}}_D.
\end{align*}
Hence 
\begin{align*}
\left\Vert \hat{\mathcal{W}}_D'\tSigma\hat{\mathcal{W}}_D - I \right\Vert = \left\Vert \hat{\mathcal{W}}_D'\left(\tSigma - \hat{\tSigma}_D \right)\hat{\mathcal{W}}_D \right\Vert \leq \left\Vert \hat{\mathcal{W}}_D \right\Vert^2  \left\Vert \tSigma - \hat{\tSigma}_D \right\Vert \leq \frac{3\varepsilon_1}{2\ts^{(K)}} \leq \frac{1}{2}.
\end{align*}
Thus, by Weyl's inequality, all eigenvalues of $\hat{\mathcal{W}}_D'\tSigma\hat{\mathcal{W}}_D $ are within $\frac{3\varepsilon_1}{2\ts^{(K)}} \leq \frac{1}{2}$ of 1. Hence the third claim is proved. We now move to the fourth and fifth claims. Since $\hat{\mathcal{W}}_D'\tSigma\hat{\mathcal{W}}_D$ is positive definite, the spectral theorem gives that
\begin{align*}
\hat{\mathcal{W}}_D'\tSigma\hat{\mathcal{W}}_D = \hat{Q}_D \hat{S}_D \hat{Q}_D',
\end{align*}
for some $\hat{Q}_D \in \mathbb{R}^{K \times K}$ and diagonal $\hat{S}_D \in \mathbb{R}^{K \times K}$ with strictly positive diagonal entries. Hence 
\begin{align*}
(\hat{\mathcal{W}}_D'\tSigma\hat{\mathcal{W}}_D)^{1/2} = \hat{Q}_D \hat{S}^{1/2}_D \hat{Q}_D', \\
(\hat{\mathcal{W}}_D'\tSigma\hat{\mathcal{W}}_D)^{-1/2} = \hat{Q}_D \hat{S}^{-1/2}_D \hat{Q}_D'.
\end{align*}
Thus
\begin{align*}
&\left\Vert (\hat{\mathcal{W}}_D'\tSigma\hat{\mathcal{W}}_D)^{1/2} \right\Vert \leq \left(1 + \frac{3\varepsilon_1}{2\ts^{(K)}}\right)^{1/2}, \\
&\left\Vert (\hat{\mathcal{W}}_D'\tSigma\hat{\mathcal{W}}_D)^{1/2} - I \right\Vert \leq \frac{3\varepsilon_1}{2\ts^{(K)}},
\end{align*}
where the second inequality follows from the fact that $|\sqrt{x} - 1| = |(x - 1)/(\sqrt{x} + 1)| \leq |x - 1|$ for $x > 0$. Hence the fourth and fifth claims are proved. Likewise $|x^{-1/2} - 1| = |(\sqrt{x} - 1)/\sqrt{x}| = \frac{|x - 1|}{\sqrt{x}(\sqrt{x} + 1)} \leq |x - 1|$ for $x > 1/2$ and thus
\begin{align*}
\left\Vert (\hat{\mathcal{W}}_D'\tSigma\hat{\mathcal{W}}_D)^{-1/2} - I \right\Vert \leq \frac{3\varepsilon_1}{2\ts^{(K)}},
\end{align*}
proving the sixth claim. For the seventh claim, note that 
\begin{align*}
\left\Vert\hat{U}_D \hat{\Lambda}^{1/2}_D \right\Vert = \left\Vert \hat{\Lambda}^{1/2}_D \right\Vert = \left\Vert \hat{\Lambda}_D \right\Vert^{1/2} \leq \left\Vert \hat{\tSigma}_D \right\Vert^{1/2} \leq& \left(\left\Vert \tSigma \right\Vert + \left\Vert \hat{\tSigma}_D - \tSigma \right\Vert\right)^{1/2} \\
\leq& \left(\left\Vert \tSigma \right\Vert + \frac{\ts^{(K)}}{3}\right)^{1/2}.
\end{align*}
\end{proof}

With properties of $\hat{\mathcal{W}}_D'\tSigma \hat{\mathcal{W}}_D$ established, we may now prove concentration of $\hat{\tmathcalT}_D(\hat{\mathcal{W}}_D, \hat{\mathcal{W}}_D, \hat{\mathcal{W}}_D)$ around  $\tmathcalT(\hat{\mathcal{V}}_D, \hat{\mathcal{V}}_D, \hat{\mathcal{V}}_D)$. Note that the error between these two quantities can be decomposed into estimation errors emanating from $\hat{\tmathcalT}_D$ and $\hat{\tSigma}_D$. In fact, the error from $\hat{\tSigma}_D$ controls the error between the whiteners $\hat{\mathcal{W}}_D$ and $\hat{\mathcal{V}}_D$.

\begin{lemma}
\label{lem:terror}
Fix $\varepsilon_1 \in (0, \ts^{(K)}/3]$ and $\varepsilon_2 > 0$. On the events
\begin{align*}
E_D(\varepsilon_1) = \{ \|\hat{\tSigma}_D - \tSigma \| \leq \varepsilon_1 \}, \quad B_D(\varepsilon_2) = \{ \|\hat{\tmathcalT}_D - \tmathcalT \| \leq \varepsilon_2 \},
\end{align*}
we have that
\begin{align*}
\left\Vert \hat{\tmathcalT}_D(\hat{\mathcal{W}}_D, \hat{\mathcal{W}}_D, \hat{\mathcal{W}}_D) - \tmathcalT(\hat{\mathcal{V}}_D, \hat{\mathcal{V}}_D,\hat{\mathcal{V}}_D) \right\Vert \leq \frac{2}{\left(\ts^{(K)}\right)^{3/2}} \varepsilon_2 + 6 \frac{\max_k \alpha_k^{-1/2}}{\ts^{(K)}} \varepsilon_1.
\end{align*}
\end{lemma}

\begin{proof}
Note that by claim 2 of Lemma \ref{lem:whitening}, we have that on $E_D(\varepsilon_1)$ and $B_D(\varepsilon_2)$
\begin{align*}
\left\Vert \hat{\tmathcalT}_D(\hat{\mathcal{W}}_D, \hat{\mathcal{W}}_D, \hat{\mathcal{W}}_D) - \tmathcalT(\hat{\mathcal{W}}_D, \hat{\mathcal{W}}_D, \hat{\mathcal{W}}_D) \right\Vert =& \left\Vert (\hat{\tmathcalT}_D - \tmathcalT)(\hat{\mathcal{W}}_D, \hat{\mathcal{W}}_D, \hat{\mathcal{W}}_D) \right\Vert \\
\leq& \left\Vert \hat{\mathcal{W}}_D\right\Vert^3  \left\Vert \hat{\tmathcalT}_D - \tmathcalT\right\Vert \\
\leq& \left(\frac{3}{2\ts^{(K)}}\right)^{3/2} \left\Vert\hat{\tmathcalT}_D - \tmathcalT\right\Vert, \\
\leq& \left(\frac{3}{2\ts^{(K)}}\right)^{3/2} \varepsilon_2 \\
\leq& \frac{2}{\left(\ts^{(K)}\right)^{3/2}} \varepsilon_2.
\end{align*}
The next step is to bound the error between $\tmathcalT(\hat{\mathcal{W}}_D, \hat{\mathcal{W}}_D, \hat{\mathcal{W}}_D)$ and $\tmathcalT(\hat{\mathcal{V}}_D, \hat{\mathcal{V}}_D, \hat{\mathcal{V}}_D)$. Note that $\hat{\mathcal{W}}_D = \hat{\mathcal{V}}_D (\hat{\mathcal{W}}_D'\tSigma \hat{\mathcal{W}}_D)^{1/2}$. In addition, it is easily seen that $\| \tmathcalT(\hat{\mathcal{V}}_D, \hat{\mathcal{V}}_D, \hat{\mathcal{V}}_D)\| \leq \max_k \alpha_k^{-1/2}$. Hence on $E_D(\varepsilon_1)$ and $B_D(\varepsilon_2)$
\begin{align*}
\left\Vert \tmathcalT(\hat{\mathcal{W}}_D, \hat{\mathcal{W}}_D, \hat{\mathcal{W}}_D) - \tmathcalT(\hat{\mathcal{V}}_D, \hat{\mathcal{V}}_D,\hat{\mathcal{V}}_D) \right\Vert \leq& \left\Vert \tmathcalT(\hat{\mathcal{W}}_D - \hat{\mathcal{V}}_D, \hat{\mathcal{W}}_D, \hat{\mathcal{W}}_D) \right\Vert \\
&+ \left\Vert \tmathcalT(\hat{\mathcal{V}}_D, \hat{\mathcal{W}}_D - \hat{\mathcal{V}}_D, \hat{\mathcal{W}}_D) \right\Vert \\
&+ \left\Vert \tmathcalT(\hat{\mathcal{V}}_D, \hat{\mathcal{V}}_D, \hat{\mathcal{W}}_D - \hat{\mathcal{V}}_D) \right\Vert \\
\leq&  \left\Vert (\hat{\mathcal{W}}_D'\tSigma \hat{\mathcal{W}}_D)^{1/2} - I \right\Vert \left\Vert(\hat{\mathcal{W}}_D'\tSigma \hat{\mathcal{W}}_D)^{1/2} \right\Vert^2 \left\Vert \tmathcalT(\hat{\mathcal{V}}_D, \hat{\mathcal{V}}_D, \hat{\mathcal{V}}_D)\right\Vert \\
&+ \left\Vert (\hat{\mathcal{W}}_D'\tSigma \hat{\mathcal{W}}_D)^{1/2} - I \right\Vert \left\Vert(\hat{\mathcal{W}}_D'\tSigma \hat{\mathcal{W}}_D)^{1/2} \right\Vert \left\Vert \tmathcalT(\hat{\mathcal{V}}_D, \hat{\mathcal{V}}_D, \hat{\mathcal{V}}_D)\right\Vert \\
&+ \left\Vert (\hat{\mathcal{W}}_D'\tSigma \hat{\mathcal{W}}_D)^{1/2} - I \right\Vert \left\Vert \tmathcalT(\hat{\mathcal{V}}_D, \hat{\mathcal{V}}_D, \hat{\mathcal{V}}_D)\right\Vert \\
\leq& \max_k \alpha_k^{-1/2} \left( \frac{3\varepsilon_1}{2\ts^{(K)}}\right)\left(\left(1 + \frac{3\varepsilon_1}{2\ts^{(K)}}\right) + \left(1 + \frac{3\varepsilon_1}{2\ts^{(K)}}\right)^{1/2} + 1 \right) \\
\leq& 6 \frac{\max_k \alpha_k^{-1/2}}{\ts^{(K)}} \varepsilon_1,
\end{align*}
where we have applied claims 4 and 5 of Lemma \ref{lem:whitening}. Hence triangle inequality gives the claim.
\end{proof}

Given that $\hat{\tmathcalT}_D(\hat{\mathcal{W}}_D, \hat{\mathcal{W}}_D, \hat{\mathcal{W}}_D)$ concentrates around 
$\tmathcalT(\hat{\mathcal{V}}_D, \hat{\mathcal{V}}_D, \hat{\mathcal{V}}_D)$, Theorem 5.1 of \cite{anandkumar2014tensor} provides that, for sufficiently large $L$ and $N_D$, the robust tensor power method outputs tensor eigenvectors and eigenvalues near the eigenvectors ($\hat{\mathcal{V}}_D'\tphi_k$, $k = 1,\dots, K$) and eigenvalues ($\alpha_k^{-1/2}$, $k = 1,\dots, K$) of the reference tensor $\tmathcalT(\hat{\mathcal{V}}_D, \hat{\mathcal{V}}_D, \hat{\mathcal{V}}_D)$ with high probability. The following lemma guarantees that, given $\hat{\tSigma}_D$ and $\hat{\tmathcalT}_D$ are near their population counterparts (and thus $\hat{\tmathcalT}_D(\hat{\mathcal{W}}_D, \hat{\mathcal{W}}_D, \hat{\mathcal{W}}_D)$ is near $\tmathcalT(\hat{\mathcal{V}}_D, \hat{\mathcal{V}}_D, \hat{\mathcal{V}}_D)$), the robust tensor power method and subsequent dewhitening in Steps 3 and 4 of Algorithm \ref{alg:specLDA} return topic estimators that concentrate around the true topics. 

\begin{lemma}
\label{lem:topicerror}
Fix $\varepsilon_1 \in (0, \ts^{(K)}/3]$ and $\varepsilon_2 > 0$. Define
\begin{align*}
\varepsilon \equiv \frac{2}{\left(\ts^{(K)}\right)^{3/2}} \varepsilon_2 + 6 \frac{\max_k \alpha_k^{-1/2}}{\ts^{(K)}} \varepsilon_1.
\end{align*}
In addition, let
\begin{align*}
E_D(\varepsilon_1) = \{ \|\hat{\tSigma}_D - \tSigma \| \leq \varepsilon_1 \}, \quad B_D(\varepsilon_2) = \{ \|\hat{\tmathcalT}_D - \tmathcalT \| \leq \varepsilon_2 \}.
\end{align*}
Then there exists universal constants $C_1$ and $C_2$ such that for any $\delta > 0$, if
\begin{align*}
\varepsilon \leq& C_1 \frac{\min_k \alpha_k^{-1/2}}{K}, \\
N \geq& C_2 \left(\log(K) + \log \log \left( \frac{\max_k \alpha_k^{-1/2}}{\varepsilon} \right) \right), \\
L \geq& \text{poly}(K)\log(1/\delta),
\end{align*}
for some fixed polynomial defined in Theorem 5.1 of \cite{anandkumar2014tensor}, then for some $M > 0$, after appropriate reordering
\begin{align*}
\mathbb{P}\left( \max_{k = 1, \dots, K} \left\Vert \hat{\beta}_{D, k} - \phi_k \right\Vert_2 \leq M\varepsilon , E_D(\varepsilon_1), B_D(\varepsilon_2) \right) \geq (1 - \delta)\mathbb{P}\left( E_D(\varepsilon_1) \cap B_D(\varepsilon_2) \right).
\end{align*}
\end{lemma}

\begin{proof}
From Lemma \ref{lem:terror}, we have that on $E_D(\varepsilon_1) \cap B_D(\varepsilon_2)$ 
\begin{align*}
\left\Vert \hat{\tmathcalT}_D(\hat{\mathcal{W}}_D, \hat{\mathcal{W}}_D, \hat{\mathcal{W}}_D) - \tmathcalT(\hat{\mathcal{V}}_D, \hat{\mathcal{V}}_D,\hat{\mathcal{V}}_D) \right\Vert \leq \varepsilon.
\end{align*}
By Theorem 5.1 of \cite{anandkumar2014tensor}, there exists universal constants $C_1$ and $C_2$ such that for any $\delta > 0$, if 
\begin{align}
\varepsilon \leq& C_1 \frac{\min_k \alpha_k^{-1/2}}{K}, \label{eq:c1} \\
N \geq& C_2 \left(\log(K) + \log \log \left( \frac{\max_k \alpha_k^{-1/2}}{\varepsilon} \right) \right), \label{eq:c2} \\
L \geq& \text{poly}(K)\log(1/\delta), \label{eq:c3}
\end{align}
then with probability at least $1 - \delta$, we have that after appropriate reordering
\begin{align}
\label{eq:r}
\| \hat{\varphi}_{D,k} - \hat{\mathcal{V}}'_D \tphi_k \|_2 \leq 8\alpha_k^{1/2}\varepsilon, \quad \left|\hat{\lambda}_{D,k} - \alpha_k^{-1/2} \right| \leq 5 \varepsilon.
\end{align}
Work on the event $E_D(\varepsilon_1) \cap B_D(\varepsilon_2)$ and suppose \eqref{eq:c1}, \eqref{eq:c2} and \eqref{eq:c3} are satisfied. Since $\tSigma = \sum_{k = 1}^K \tphi_k \otimes \tphi_k$, each $\tphi_k$ lies in the column space of $\tSigma$, whence $UU'\tphi_k = \tphi_k$ and $\phi_k = \alpha^{-1/2}_k UU'\tphi_k$. We may thus upper bound the error between $\hat{\beta}_{D,k}$ and $\phi_k$ via the triangle inequality
\begin{align*}
\left\Vert \hat{\beta}_{D,k} - \phi_k \right\Vert_2 \leq& \left\Vert (\hat{\lambda}_{D,k} - \alpha^{-1/2}_k)\hat{U}_D\hat{\Lambda}_D^{1/2} \hat{\varphi}_{D,k} \right\Vert_2 \\
&+ \left\Vert \alpha^{-1/2}_k \hat{U}_D\hat{\Lambda}_D^{1/2} (\hat{\varphi}_{D,k} - \hat{\mathcal{V}}'_D\tphi_k) \right\Vert_2 \\
&+ \left\Vert \alpha^{-1/2}_k \hat{U}_D\hat{\Lambda}_D^{1/2} \left( (\hat{\mathcal{W}}_D'\tSigma \hat{\mathcal{W}}_D)^{-1/2} - I \right) \hat{\mathcal{W}}'_D  \tphi_k \right\Vert_2 \\
&+ \left\Vert \alpha^{-1/2}_k \left(\hat{U}_D\hat{\Lambda}_D^{1/2} \hat{\mathcal{W}}'_D  - UU' \right) \tphi_k \right\Vert_2.
\end{align*}
Note that since $\hat{\varphi}_{D,k}$ is a unit vector, we may apply claim 7 of Lemma \ref{lem:whitening} to see that
\begin{align*}
\left\Vert (\hat{\lambda}_{D,k} - \alpha^{-1/2}_k)\hat{U}_D\hat{\Lambda}_D^{1/2} \hat{\varphi}_{D,k} \right\Vert_2 \leq&  \left|\hat{\lambda}_{D,k} - \alpha^{-1/2}_k\right| \left\Vert \hat{U}_D\hat{\Lambda}_D^{1/2} \right\Vert \\
\leq& \left(\left\Vert \tSigma \right\Vert + \frac{\ts^{(K)}}{3}\right)^{1/2} 5 \varepsilon,
\end{align*}
where the last inequality is due to \eqref{eq:r}. Similarly, from \eqref{eq:r}
\begin{align*}
\left\Vert \alpha^{-1/2}_k \hat{U}_D\hat{\Lambda}_D^{1/2} (\hat{\varphi}_{D,k} - \hat{\mathcal{V}}'_D\tphi_k) \right\Vert_2 \leq& \alpha^{-1/2}_k \left\Vert \hat{U}_D\hat{\Lambda}_D^{1/2} \right\Vert \left\Vert \hat{\varphi}_{D,k} - \hat{\mathcal{V}}'_D\tphi_k \right\Vert_2 \\
\leq& \left(\left\Vert \tSigma \right\Vert + \frac{\ts^{(K)}}{3}\right)^{1/2}8 \varepsilon.
\end{align*}
Applying claims 2, 6 and 7 of Lemma \ref{lem:whitening} gives that
\begin{align*}
\left\Vert \alpha^{-1/2}_k \hat{U}_D\hat{\Lambda}_D^{1/2} \left( (\hat{\mathcal{W}}_D'\tSigma \hat{\mathcal{W}}_D)^{-1/2} - I \right) \hat{\mathcal{W}}'_D  \tphi_k \right\Vert_2 \leq& \alpha^{-1/2}_k \left\Vert\hat{U}_D\hat{\Lambda}_D^{1/2} \right\Vert \left\Vert (\hat{\mathcal{W}}_D'\tSigma \hat{\mathcal{W}}_D)^{-1/2} - I\right\Vert  \left\Vert \hat{\mathcal{W}}_D' \right\Vert \|\tphi_k \|_2, \\
\leq& \alpha^{-1/2}_k \|\tphi_k \|_2 \left(\left\Vert \tSigma \right\Vert + \frac{\ts^{(K)}}{3}\right)^{1/2}  \left(\frac{3}{2\ts^{(K)}}\right)^{3/2}\varepsilon_1 \\
\leq& \|\tphi_k \|_2 \left(\left\Vert \tSigma \right\Vert + \frac{\ts^{(K)}}{3}\right)^{1/2}  \left(\frac{3}{2}\right)^{3/2} \frac{1}{6\sqrt{\ts^{(K)}}}\varepsilon.
\end{align*}
For the final term, note that $\hat{U}_D\hat{\Lambda}_D^{1/2} \hat{\mathcal{W}}'_D = \hat{U}_D\hat{U}_D'$. Hence, by the Davis-Kahan theorem for spectral projections (see Corollary 1.2 of \cite{tran2025davis})
\begin{align*}
\left\Vert \alpha^{-1/2}_k \left(\hat{U}_D\hat{\Lambda}_D^{1/2} \hat{\mathcal{W}}'_D  - UU' \right) \tphi_k \right\Vert_2 =& \left\Vert \alpha^{-1/2}_k \left( \hat{U}_D\hat{U}_D' - UU' \right) \tphi_k \right\Vert_2 \\
\leq& \alpha^{-1/2}_k  \frac{\pi}{\ts^{(K)}} \|\tphi_k \|_2 \varepsilon_1 \\
\leq& \|\tphi_k \|_2 \frac{\pi}{6}\varepsilon.
\end{align*}
Hence, there exists a constant $M > 0$ such that on the event $E_D(\varepsilon_1) \cap B_D(\varepsilon_2)$, with probability at least $1 - \delta$,
\begin{align*}
\max_{k = 1, \dots, K} \left\Vert \hat{\beta}_{D,k} - \phi_k \right\Vert_2 \leq& M\varepsilon.
\end{align*}
In other words, 
\begin{align*}
\mathbb{P}&\left( \max_{k = 1, \dots, K} \left\Vert \hat{\beta}_{D,k} - \phi_k \right\Vert_2 \leq M\varepsilon , E_D(\varepsilon_1), B_D(\varepsilon_2) \right) \\
&= \mathbb{E}\left[ \mathbb{P}\left( \max_{k = 1, \dots, K} \left\Vert \hat{\beta}_{D,k} - \phi_k \right\Vert_2 \leq M\varepsilon \mid \hat{\tSigma}_D, \hat{\tmathcalT}_D \right) 1_{\left\lbrace E_D(\varepsilon_1) \cap B_D(\varepsilon_2)  \right\rbrace}\right] \\
&\geq (1 - \delta) \mathbb{P}\left( E_D(\varepsilon_1) \cap B_D(\varepsilon_2) \right).
\end{align*}
\end{proof}

With the supporting lemmas established, we now present the proof of Theorem \ref{thm:roc}.

\begin{proof}[Proof of Theorem \ref{thm:roc}]
From Lemma \ref{lem:Op}, for a fixed $\varsigma > 0$ there exists constants $M_1$ and $M_2$ such that 
\begin{align*}
\sup_D \mathbb{P}\left(E^c_D\left(\frac{M_1}{\sqrt{k_D}} \right) \right) < \varsigma/2, \quad \sup_D \mathbb{P}\left(B^c_D\left(\frac{M_2}{\sqrt{k_D}} \right) \right) < \varsigma/2.
\end{align*}
Choose $C = \frac{2}{\left(\ts^{(K)}\right)^{3/2}} M_2 + 6 \frac{\max_k \alpha_k^{-1/2}}{\ts^{(K)}} M_1$. Then on $E_D\left(\frac{M_1}{\sqrt{k_D}} \right)$ and $B_D\left(\frac{M_2}{\sqrt{k_D}} \right)$ the conditions of Lemma \ref{lem:topicerror} with $\varepsilon = C/\sqrt{k_D}$ are satisfied for $D$ sufficiently large and hence
\begin{align*}
\mathbb{P}&\left( \max_{k = 1, \dots, K} \left\Vert \hat{\beta}_{D,k} - \phi_k \right\Vert_2 \leq M \frac{C}{\sqrt{k_D}}\right) \\
&\geq \mathbb{P}\left( \max_{k = 1, \dots, K} \left\Vert \hat{\beta}_{D,k} - \phi_k \right\Vert_2 \leq M\frac{C}{\sqrt{k_D}}, E_D\left(\frac{M_1}{\sqrt{k_D}} \right), B_D\left(\frac{M_2}{\sqrt{k_D}} \right) \right) \\
&\geq (1 - \delta) \mathbb{P}\left( E_D\left(\frac{M_1}{\sqrt{k_D}} \right) \cap B_D\left(\frac{M_2}{\sqrt{k_D}} \right) \right) \\
&\geq (1 - \delta) \left( 1 - \mathbb{P}\left( E^c_D\left(\frac{M_1}{\sqrt{k_D}} \right)\right) - \mathbb{P}\left( B^c_D\left(\frac{M_2}{\sqrt{k_D}} \right) \right) \right).
\end{align*}
Hence for $D$ sufficiently large
\begin{align*}
\mathbb{P}\left( \max_{k = 1, \dots, K} \left\Vert \hat{\beta}_{D,k} - \phi_k \right\Vert_2 > M\frac{C}{\sqrt{k_D}}\right) \leq& \delta + (1 -\delta)\left( \mathbb{P}\left( E^c_D\left(\frac{M_1}{\sqrt{k_D}} \right)\right) + \mathbb{P}\left( B^c_D\left(\frac{M_2}{\sqrt{k_D}} \right) \right) \right) \\
\leq& \delta + \varsigma.
\end{align*}
\end{proof}

\bibliographystyle{plain}
\bibliography{TopicModelMRV.bib}

@article{blei2003latent,
author = {Blei, David M. and Ng, Andrew Y. and Jordan, Michael I.},
title = {Latent {D}irichlet {A}llocation},
year = {2003},
issue_date = {3/1/2003},
publisher = {JMLR.org},
volume = {3},
number = {Jan},
issn = {1532-4435},
journal = {Journal of Machine Learning Research},
month = mar,
pages = {993–1022},
numpages = {30}
}

@inproceedings{hofmann1999probabilistic,
author = {Hofmann, Thomas},
title = {Probabilistic latent semantic indexing},
year = {1999},
isbn = {1581130961},
publisher = {Association for Computing Machinery},
address = {New York, NY, USA},
url = {https://doi.org/10.1145/312624.312649},
doi = {10.1145/312624.312649},
booktitle = {Proceedings of the 22nd Annual International ACM SIGIR Conference on Research and Development in Information Retrieval},
pages = {50–57},
numpages = {8},
location = {Berkeley, California, USA},
series = {SIGIR '99}
}

@article{wang2022asymptotic,
  title={Asymptotic dependence of in-and out-degrees in a preferential attachment model with reciprocity},
  author={Wang, Tiandong and Resnick, Sidney I},
  journal={Extremes},
  volume={25},
  number={3},
  pages={417--450},
  year={2022},
  publisher={Springer}
}

@article{janssen2025multivariate,
  title={Multivariate regular variation of preferential attachment models}, 
  volume={57}, 
  DOI={10.1017/apr.2024.73}, 
  number={3}, 
  journal={Advances in Applied Probability}, 
  author={Janssen, Anja and Ziegenbalg, Max}, 
  year={2025}, 
  pages={1068–1100}
}

@InProceedings{aldous2006exchangeability,
author="Aldous, David J.",
editor="Hennequin, P. L.",
title="Exchangeability and related topics",
booktitle="{\'E}cole d'{\'E}t{\'e} de Probabilit{\'e}s de Saint-Flour XIII --- 1983",
year="1985",
publisher="Springer Berlin Heidelberg",
address="Berlin, Heidelberg",
pages="1--198",
isbn="978-3-540-39316-0"
}

@article{lee1999learning,
  title={Learning the parts of objects by non-negative matrix factorization},
  author={Lee, Daniel D and Seung, H Sebastian},
  journal={Nature},
  volume={401},
  number={6755},
  pages={788--791},
  year={1999},
  publisher={Nature Publishing Group UK London}
}

@inproceedings{wallach2006topic,
  author = {Wallach, Hanna M.},
  title = {Topic modeling: beyond bag-of-words},
  year = {2006},
  isbn = {1595933832},
  publisher = {Association for Computing Machinery},
  address = {New York, NY, USA},
  url = {https://doi.org/10.1145/1143844.1143967},
  doi = {10.1145/1143844.1143967},
booktitle = {Proceedings of the 23rd International Conference on Machine Learning},
  pages = {977–984},
  numpages = {8},
  location = {Pittsburgh, Pennsylvania, USA},
  series = {ICML '06}
}

@INPROCEEDINGS{arora2012learning,
  author={Arora, Sanjeev and Ge, Rong and Moitra, Ankur},
  booktitle={2012 IEEE 53rd Annual Symposium on Foundations of Computer Science}, 
  title={Learning Topic Models -- Going beyond SVD}, 
  year={2012},
  volume={},
  number={},
  pages={1-10},
  doi={10.1109/FOCS.2012.49}
}

@inproceedings{anandkumar2012spectral,
 author = {Anandkumar, Anima and Foster, Dean P and Hsu, Daniel J and Kakade, Sham M and Liu, Yi-kai},
 booktitle = {Advances in Neural Information Processing Systems},
 editor = {F. Pereira and C.J. Burges and L. Bottou and K.Q. Weinberger},
 pages = {},
 publisher = {Curran Associates, Inc.},
 title = {A Spectral Algorithm for Latent {D}irichlet Allocation},
 url = {https://proceedings.neurips.cc/paper_files/paper/2012/file/15d4e891d784977cacbfcbb00c48f133-Paper.pdf},
 volume = {25},
 year = {2012}
}

@article{bing2020fast,
author = {Xin Bing and Florentina Bunea and Marten Wegkamp},
title = {{A fast algorithm with minimax optimal guarantees for topic models with an unknown number of topics}},
volume = {26},
journal = {Bernoulli},
number = {3},
publisher = {Bernoulli Society for Mathematical Statistics and Probability},
pages = {1765 -- 1796},
year = {2020},
doi = {10.3150/19-BEJ1166},
URL = {https://doi.org/10.3150/19-BEJ1166}
}

@InProceedings{decarolis2020end,
  title = 	 {An end-to-end Differentially Private Latent {D}irichlet Allocation Using a Spectral Algorithm},
  author =       {Decarolis, Chris and Ram, Mukul and Esmaeili, Seyed and Wang, Yu-Xiang and Huang, Furong},
  booktitle = 	 {Proceedings of the 37th International Conference on Machine Learning},
  pages = 	 {2421--2431},
  year = 	 {2020},
  editor = 	 {III, Hal Daumé and Singh, Aarti},
  volume = 	 {119},
  series = 	 {Proceedings of Machine Learning Research},
  month = 	 {13--18 Jul},
  publisher =    {PMLR},
  url = 	 {https://proceedings.mlr.press/v119/decarolis20a.html}
}

@InProceedings{arabshahi2017spectral,
  title = 	 {Spectral Methods for Correlated Topic Models},
  author = 	 {Arabshahi, Forough and Anandkumar, Anima},
  booktitle = 	 {Proceedings of the 20th International Conference on Artificial Intelligence and Statistics},
  pages = 	 {1439--1447},
  year = 	 {2017},
  editor = 	 {Singh, Aarti and Zhu, Jerry},
  volume = 	 {54},
  series = 	 {Proceedings of Machine Learning Research},
  month = 	 {20--22 Apr},
  publisher =    {PMLR},
  url = 	 {https://proceedings.mlr.press/v54/arabshahi17a.html}
}

@article{anandkumar2014tensor,
  author  = {Animashree Anandkumar and Rong Ge and Daniel Hsu and Sham M. Kakade and Matus Telgarsky},
  title   = {Tensor Decompositions for Learning Latent Variable Models},
  journal = {Journal of Machine Learning Research},
  year    = {2014},
  volume  = {15},
  number  = {80},
  pages   = {2773--2832},
  url     = {http://jmlr.org/papers/v15/anandkumar14b.html}
}

@InProceedings{cheng2015model,
  title = 	 {{Model Selection for Topic Models via Spectral Decomposition}},
  author = 	 {Cheng, Dehua and He, Xinran and Liu, Yan},
  booktitle = 	 {Proceedings of the Eighteenth International Conference on Artificial Intelligence and Statistics},
  pages = 	 {183--191},
  year = 	 {2015},
  editor = 	 {Lebanon, Guy and Vishwanathan, S. V. N.},
  volume = 	 {38},
  series = 	 {Proceedings of Machine Learning Research},
  address = 	 {San Diego, California, USA},
  month = 	 {09--12 May},
  publisher =    {PMLR},
  url = 	 {https://proceedings.mlr.press/v38/cheng15.html}
}

@article{das2013living,
  title={Living on the multidimensional edge: seeking hidden risks using regular variation},
  author={Das, Bikramjit and Mitra, Abhimanyu and Resnick, Sidney},
  journal={Advances in Applied Probability},
  volume={45},
  number={1},
  pages={139--163},
  year={2013},
  publisher={Cambridge University Press}
}

@article{lindskog2014regularly,
author = {Filip Lindskog and Sidney I. Resnick and Joyjit Roy},
title = {{Regularly varying measures on metric spaces: Hidden regular variation and hidden jumps}},
volume = {11},
journal = {Probability Surveys},
number = {none},
publisher = {Institute of Mathematical Statistics and Bernoulli Society},
pages = {270 -- 314},
year = {2014},
doi = {10.1214/14-PS231},
URL = {https://doi.org/10.1214/14-PS231}
}

@article{das2017hidden,
  title={Hidden regular variation under full and strong asymptotic dependence},
  author={Das, Bikramjit and Resnick, Sidney I},
  journal={Extremes},
  volume={20},
  number={4},
  pages={873--904},
  year={2017},
  publisher={Springer}
}

@book{resnick2007heavy,
  author    = {Resnick, Sidney I.},
  title     = {Heavy-Tail Phenomena: Probabilistic and Statistical Modeling},
  series    = {Springer Series in Operations Research and Financial Engineering},
  year      = {2007},
  publisher = {Springer},
  address   = {New York, NY},
  isbn      = {978-0-387-24272-9},
  doi       = {10.1007/978-0-387-45024-7},
  url       = {https://doi.org/10.1007/978-0-387-45024-7}
}

@misc{huang2017spectral,
  title={Spectral {L}atent {D}irichlet {A}llocation model on {S}park},
  author={Huang, Furong and Lee, Jencir and Anandkumar, Anima},
  year={2017}
}

@article{kolda2009tensor,
  author = {Kolda, Tamara G. and Bader, Brett W.},
  title = {Tensor Decompositions and Applications},
  journal = {SIAM Review},
  volume = {51},
  number = {3},
  pages = {455-500},
  year = {2009},
  doi = {10.1137/07070111X},
  URL = { https://doi.org/10.1137/07070111X},
  eprint = {https://doi.org/10.1137/07070111X}
}

@article{harshman1970foundations,
  title={Foundations of the {PARAFAC} procedure: Models and conditions for an “explanatory” multi-modal factor analysis},
  author={Harshman, Richard A and others},
  journal={{UCLA} working papers in phonetics},
  volume={16},
  number={1},
  pages={84},
  year={1970},
  publisher={Los Angeles, CA}
}

@incollection{kallenberg2021symmetric,
  title={Symmetric Distributions and Predictable Maps},
  author={Kallenberg, Olav},
  booktitle={Foundations of Modern Probability},
  pages={611--630},
  year={2021},
  publisher={Springer}
}

@book{resnick2024art,
  author    = {Resnick, Sidney},
  title     = {The Art of Finding Hidden Risks: Hidden Regular Variation in the 21st Century},
  series    = {Springer Series in Operations Research and Financial Engineering},
  year      = {2024},
  publisher = {Springer},
  address   = {Cham},
  isbn      = {978-3-031-57598-3},
  doi       = {10.1007/978-3-031-57599-0},
  url       = {https://doi.org/10.1007/978-3-031-57599-0}
}

@article{huang2015online,
  author  = {Furong Huang and U. N. Niranjan and Mohammad Umar Hakeem and Animashree Anandkumar},
  title   = {Online Tensor Methods for Learning Latent Variable Models},
  journal = {Journal of Machine Learning Research},
  year    = {2015},
  volume  = {16},
  number  = {86},
  pages   = {2797--2835},
  url     = {http://jmlr.org/papers/v16/huang15a.html}
}

@book{lehoucq1998arpack,
  title={ARPACK users' guide: solution of large-scale eigenvalue problems with implicitly restarted Arnoldi methods},
  author={Lehoucq, Richard B and Sorensen, Danny C and Yang, Chao},
  year={1998},
  publisher={SIAM}
}

@article{griffiths2004finding,
  title={Finding scientific topics},
  author={Griffiths, Thomas L and Steyvers, Mark},
  journal={Proceedings of the National Academy of Sciences},
  volume={101},
  number={suppl\_1},
  pages={5228--5235},
  year={2004},
  publisher={National Academy of Sciences}
}

@article{bystrov2024choosing,
  title={Choosing the number of topics in {LDA} models--A Monte Carlo comparison of selection criteria},
  author={Bystrov, Victor and Naboka-Krell, Viktoriia and Staszewska-Bystrova, Anna and Winker, Peter},
  journal={Journal of Machine Learning Research},
  volume={25},
  number={79},
  pages={1--30},
  year={2024}
}

@inproceedings{newman2011improving,
  author = {Newman, David and Bonilla, Edwin and Buntine, Wray},
  booktitle = {Advances in Neural Information Processing Systems},
  editor = {J. Shawe-Taylor and R. Zemel and P. Bartlett and F. Pereira and K. Weinberger},
  pages = {},
  publisher = {Curran Associates, Inc.},
  title = {Improving Topic Coherence with Regularized Topic Models},
  url = {https://proceedings.neurips.cc/paper_files/paper/2011/file/5ef698cd9fe650923ea331c15af3b160-Paper.pdf},
  volume = {24},
  year = {2011}
}

@book{vershynin2026high,
  place={Cambridge}, 
  edition={2}, 
  series={Cambridge Series in Statistical and Probabilistic Mathematics}, 
  title={High-Dimensional Probability: An Introduction with Applications in Data Science}, 
  publisher={Cambridge University Press}, author={Vershynin, Roman},
  year={2026}, 
  collection={Cambridge Series in Statistical and Probabilistic Mathematics}
}

@article{deerwester1990indexing,
  author = {Deerwester, Scott and Dumais, Susan T. and Furnas, George W. and Landauer, Thomas K. and Harshman, Richard},
  title = {Indexing by latent semantic analysis},
  journal = {Journal of the American Society for Information Science},
  volume = {41},
  number = {6},
  pages = {391-407},
  doi = {https://doi.org/10.1002/(SICI)1097-4571(199009)41:6<391::AID-ASI1>3.0.CO;2-9},
  url = {https://asistdl.onlinelibrary.wiley.com/doi/abs/10.1002/%28SICI%291097-4571%28199009%2941%3A6%3C391%3A%3AAID-ASI1%3E3.0.CO%3B2-9},
eprint = {https://asistdl.onlinelibrary.wiley.com/doi/pdf/10.1002/%28SICI%291097-4571%28199009%2941%3A6%3C391%3A%3AAID-ASI1%3E3.0.CO%3B2-9},
  year = {1990}
}

@article{landauer1997solution,
  title={A solution to {P}lato's problem: The latent semantic analysis theory of acquisition, induction, and representation of knowledge.},
  author={Landauer, Thomas K and Dumais, Susan T},
  journal={Psychological Review},
  volume={104},
  number={2},
  pages={211},
  year={1997},
  publisher={American Psychological Association}
}

@article{larsson2012extremal,
  title={Extremal dependence measure and extremogram: the regularly varying case},
  author={Larsson, Martin and Resnick, Sidney I},
  journal={Extremes},
  volume={15},
  number={2},
  pages={231--256},
  year={2012},
  publisher={Springer}
}

@article{kluppelberg2021estimating,
  title = {Estimating an extreme {B}ayesian network via scalings},
  journal = {Journal of Multivariate Analysis},
  volume = {181},
  pages = {104672},
  year = {2021},
  issn = {0047-259X},
  doi = {https://doi.org/10.1016/j.jmva.2020.104672},
  url = {https://www.sciencedirect.com/science/article/pii/S0047259X20302530},
  author = {Claudia Klüppelberg and Mario Krali}
}

@inproceedings{duchi2008efficient,
  author = {Duchi, John and Shalev-Shwartz, Shai and Singer, Yoram and Chandra, Tushar},
  title = {Efficient projections onto the l1-ball for learning in high dimensions},
  year = {2008},
  isbn = {9781605582054},
  publisher = {Association for Computing Machinery},
  address = {New York, NY, USA},
  url = {https://doi.org/10.1145/1390156.1390191},
  doi = {10.1145/1390156.1390191},
  booktitle = {Proceedings of the 25th International Conference on Machine Learning},
  pages = {272–279},
  numpages = {8},
  location = {Helsinki, Finland},
  series = {ICML '08}
}

@article{comon2008symmetric,
  title={Symmetric tensors and symmetric tensor rank},
  author={Comon, Pierre and Golub, Gene and Lim, Lek-Heng and Mourrain, Bernard},
  journal={SIAM Journal on Matrix Analysis and Applications},
  volume={30},
  number={3},
  pages={1254--1279},
  year={2008},
  publisher={SIAM}
}

@article{kolda2015symmetric,
  title={Symmetric orthogonal tensor decomposition is trivial},
  author={Kolda, Tamara G},
  journal={arXiv preprint arXiv:1503.01375},
  year={2015}
}

@article{williams2015text,
  title = {Text mixing shapes the anatomy of rank-frequency distributions},
  author = {Williams, Jake Ryland and Bagrow, James P. and Danforth, Christopher M. and Dodds, Peter Sheridan},
  journal = {Phys. Rev. E},
  volume = {91},
  issue = {5},
  pages = {052811},
  numpages = {8},
  year = {2015},
  month = {May},
  publisher = {American Physical Society},
  doi = {10.1103/PhysRevE.91.052811},
  url = {https://link.aps.org/doi/10.1103/PhysRevE.91.052811}
}

@article{montemurro2001beyond,
  title = {Beyond the {Z}ipf–{M}andelbrot law in quantitative linguistics},
journal = {Physica A: Statistical Mechanics and its Applications},
volume = {300},
number = {3},
pages = {567-578},
year = {2001},
issn = {0378-4371},
doi = {https://doi.org/10.1016/S0378-4371(01)00355-7},
url = {https://www.sciencedirect.com/science/article/pii/S0378437101003557},
author = {Marcelo A. Montemurro}
}

@inproceedings{sato2010topic,
  author = {Sato, Issei and Nakagawa, Hiroshi},
  title = {Topic models with power-law using {P}itman-{Y}or process},
  year = {2010},
  isbn = {9781450300551},
  publisher = {Association for Computing Machinery},
  address = {New York, NY, USA},
  url = {https://doi.org/10.1145/1835804.1835890},
  doi = {10.1145/1835804.1835890},
  booktitle = {Proceedings of the 16th ACM SIGKDD International Conference on Knowledge Discovery and Data Mining},
  pages = {673–682},
  numpages = {10},
  location = {Washington, DC, USA},
  series = {KDD '10}
}

@book{zipf1935psycho,
  address = {New York, NY, USA},
  author = {Zipf, George K.},
  publisher = {Houghton-Mifflin},
  title = {The Psychobiology of Language},
  year = {1935}
}

@article{grun2011topicmodels,
  title={topicmodels: {A}n {R} Package for Fitting Topic Models},
  volume={40},
  url={https://www.jstatsoft.org/index.php/jss/article/view/v040i13},
  doi={10.18637/jss.v040.i13},
  number={13},
  journal={Journal of Statistical Software},
  author={Grün, Bettina and Hornik, Kurt},
  year={2011},
  pages={1–30}
}

@Manual{rSpectra,
    title = {RSpectra: Solvers for Large-Scale Eigenvalue and SVD Problems},
    author = {Yixuan Qiu and Jiali Mei},
    year = {2024},
    note = {R package version 0.16-2},
    url = {https://CRAN.R-project.org/package=RSpectra},
    doi = {10.32614/CRAN.package.RSpectra},
}

@incollection{lang1995newsweeder,
  	title = {NewsWeeder: Learning to Filter Netnews},	
	editor = {Armand Prieditis and Stuart Russell},
	booktitle = {Machine Learning Proceedings 1995},
	publisher = {Morgan Kaufmann},
	address = {San Francisco (CA)},
	pages = {331-339},
	year = {1995},
	isbn = {978-1-55860-377-6},
	doi = {https://doi.org/10.1016/B978-1-55860-377-6.50048-7},
	url = {https://www.sciencedirect.com/science/article/pii/B9781558603776500487},
	author = {Ken Lang}
}

@article{pedregosa2011scikit,
  title={Scikit-learn: Machine learning in Python},
  author={Pedregosa, Fabian and Varoquaux, Ga{\"e}l and Gramfort, Alexandre and Michel, Vincent and Thirion, Bertrand and Grisel, Olivier and Blondel, Mathieu and Prettenhofer, Peter and Weiss, Ron and Dubourg, Vincent and others},
  journal={Journal of Machine Learning Research},
  volume={12},
  pages={2825--2830},
  year={2011},
  publisher={JMLR. org}
}

@book{silge2017text,
  title={Text mining with R: A tidy approach},
  author={Silge, Julia and Robinson, David and Robinson, David},
  year={2017},
  publisher={O'Reilly Boston (MA)}
}

@article{soleimani2014parsimonious,
  author={Soleimani, Hossein and Miller, David J.},
  journal={IEEE Transactions on Knowledge and Data Engineering}, 
  title={Parsimonious Topic Models with Salient Word Discovery}, 
  year={2015},
  volume={27},
  number={3},
  pages={824-837},
  doi={10.1109/TKDE.2014.2345378}
}

@article{srivastava2017autoencoding,
  title={Autoencoding variational inference for topic models},
  author={Srivastava, Akash and Sutton, Charles},
  journal={arXiv preprint arXiv:1703.01488},
  year={2017}
}

@book{yule2014statistical,
  title={The Statistical Study of Literary Vocabulary},
  author={Yule, George U.},
  lccn={44029835},
  url={https://books.google.com/books?id=DaCCAAAAIAAJ},
  year={1944},
  publisher={The University Press}
}

@article{simon1955class,
  author = {Simon, Herbert A.},
  title = {ON A CLASS OF SKEW DISTRIBUTION FUNCTIONS},
  journal = {Biometrika},
  volume = {42},
  number = {3-4},
  pages = {425-440},
  year = {1955},
  month = {12},
  issn = {0006-3444},
  doi = {10.1093/biomet/42.3-4.425},
  url = {https://doi.org/10.1093/biomet/42.3-4.425},
  eprint = {https://academic.oup.com/biomet/article-pdf/42/3-4/425/838668/42-3-4-425.pdf},
}

@book{zipf2016human,
  author = {Zipf, George K.},
  publisher = {Addison-Wesley},
  title = {Human Behaviour and the Principle of Least Effort},
  year = {1949}
}

@inproceedings{goldwater2005interpolating,
  author = {Goldwater, Sharon and Johnson, Mark and Griffiths, Thomas},
  booktitle = {Advances in Neural Information Processing Systems},
  editor = {Y. Weiss and B. Sch\"{o}lkopf and J. Platt},
  pages = {},
  publisher = {MIT Press},
  title = {Interpolating between types and tokens by estimating power-law generators},
  url = {https://proceedings.neurips.cc/paper_files/paper/2005/file/4b21cf96d4cf612f239a6c322b10c8fe-Paper.pdf},
  volume = {18},
  year = {2005}
}

@article{goldwater2011producing,
  title={Producing Power-Law Distributions and Damping Word Frequencies with Two-Stage Language Models.},
  author={Goldwater, Sharon and Griffiths, Thomas L and Johnson, Mark},
  journal={Journal of Machine Learning Research},
  volume={12},
  number={7},
  year={2011}
}

@article{wang2025fine,
  author  = {Feifei Wang and Zimeng Zhao and Ruimin Ye and Xiaoge Gu and Xiaoling Lu},
  title   = {Fine-Grained Change Point Detection for Topic Modeling with {P}itman-{Y}or Process},
  journal = {Journal of Machine Learning Research},
  year    = {2025},
  volume  = {26},
  number  = {67},
  pages   = {1--53},
  url     = {http://jmlr.org/papers/v26/23-1576.html}
}

@inproceedings{teh2006hierarchical,
  title={A hierarchical Bayesian language model based on {P}itman-{Y}or processes},
  author={Teh, Yee Whye},
  booktitle={Proceedings of the 21st International Conference on Computational Linguistics and 44th Annual Meeting of the Association for Computational Linguistics},
  pages={985--992},
  year={2006}
}

@article{engelke2021sparse,
  title={Sparse structures for multivariate extremes},
  author={Engelke, Sebastian and Ivanovs, Jevgenijs},
  journal={Annual Review of Statistics and Its Application},
  volume={8},
  number={1},
  pages={241--270},
  year={2021},
  publisher={Annual Reviews}
}

@article{janssen2020k,
 author = {Anja Jan{\ss}en and Phyllis Wan},
 title = {{$k$-means clustering of extremes}},
 volume = {14},
 journal = {Electronic Journal of Statistics},
 number = {1},
 publisher = {Institute of Mathematical Statistics and Bernoulli Society},
 pages = {1211 -- 1233},
 year = {2020},
 doi = {10.1214/20-EJS1689},
 URL = {https://doi.org/10.1214/20-EJS1689}
}

@article{medina2024spectral,
  author  = {Marco Avella Medina and Richard A Davis and Gennady Samorodnitsky},
  title   = {Spectral learning of multivariate extremes},
  journal = {Journal of Machine Learning Research},
  year    = {2024},
  volume  = {25},
  number  = {124},
  pages   = {1--36},
  url     = {http://jmlr.org/papers/v25/21-1367.html}
}

@article{medina2025insights,
  author = {Medina, Marco Avella and Davis, Richard A. and Samorodnitsky, Gennady},
  title = {Insights into Kernel {PCA} with Application to Multivariate Extremes},
  journal = {SIAM Journal on Mathematics of Data Science},
  volume = {7},
  number = {2},
  pages = {777-801},
  year = {2025},
  doi = {10.1137/24M1678635},
  URL = {https://doi.org/10.1137/24M1678635},
  eprint = {https://doi.org/10.1137/24M1678635}
}

@article{reinbott2026principal,
  author = {Felix Reinbott and Anja Janßen},
  title = {Principal Component Analysis for Max-Stable Distributions},
  journal = {Journal of the American Statistical Association},
  volume = {0},
  number = {0},
  pages = {1--12},
  year = {2026},
  publisher = {Taylor \& Francis},
  doi = {10.1080/01621459.2025.2595732},
  URL = {https://doi.org/10.1080/01621459.2025.2595732},
  eprint = {https://doi.org/10.1080/01621459.2025.2595732}
}

@article{cooley2019decompositions,
  title={Decompositions of dependence for high-dimensional extremes},
  author={Cooley, Daniel and Thibaud, Emeric},
  journal={Biometrika},
  volume={106},
  number={3},
  pages={587--604},
  year={2019},
  publisher={Oxford University Press}
}

@article{butsch2025estimation,
  title={Estimation of the number of principal components in high-dimensional multivariate extremes},
  author={Butsch, Lucas and Fasen-Hartmann, Vicky},
  journal={Scandinavian Journal of Statistics},
  volume={52},
  number={4},
  pages={2270--2313},
  year={2025},
  publisher={Wiley Online Library}
}

@article{drees2020minimum,
  title={On a minimum distance procedure for threshold selection in tail analysis},
  author={Drees, Holger and Jan{\ss}en, Anja and Resnick, Sidney I and Wang, Tiandong},
  journal={SIAM Journal on Mathematics of Data Science},
  volume={2},
  number={1},
  pages={75--102},
  year={2020},
  publisher={SIAM}
}

@article{clauset2009power,
  title={Power-law distributions in empirical data},
  author={Clauset, Aaron and Shalizi, Cosma Rohilla and Newman, Mark EJ},
  journal={SIAM Review},
  volume={51},
  number={4},
  pages={661--703},
  year={2009},
  publisher={SIAM}
}

@article{sparck1972statistical,
  title={A statistical interpretation of term specificity and its application in retrieval},
  author={Sparck Jones, Karen},
  journal={Journal of Documentation},
  volume={28},
  number={1},
  pages={11--21},
  year={1972},
  publisher={MCB UP Ltd}
}

@article{grootendorst2022bertopic,
  title={BERTopic: Neural topic modeling with a class-based TF-IDF procedure},
  author={Grootendorst, Maarten},
  journal={arXiv preprint arXiv:2203.05794},
  year={2022}
}

@article{blei2007correlated,
  ISSN = {19326157},
  URL = {http://www.jstor.org/stable/4537420},
  author = {David M. Blei and John D. Lafferty},
  journal = {The Annals of Applied Statistics},
  number = {1},
  pages = {17--35},
  publisher = {Institute of Mathematical Statistics},
  title = {A Correlated Topic Model of Science},
  urldate = {2026-07-17},
  volume = {1},
  year = {2007}
}

@inproceedings{blei2006correlated,
  author = {Lafferty, John and Blei, David},
  booktitle = {Advances in Neural Information Processing Systems},
  editor = {Y. Weiss and B. Sch\"{o}lkopf and J. Platt},
  pages = {},
  publisher = {MIT Press},
  title = {Correlated Topic Models},
  url = {https://proceedings.neurips.cc/paper_files/paper/2005/file/9e82757e9a1c12cb710ad680db11f6f1-Paper.pdf},
  volume = {18},
  year = {2005}
}

@article{mosimann1962compound,
  ISSN = {00063444, 14643510},
  URL = {http://www.jstor.org/stable/2333468},
  author = {James E. Mosimann},
  journal = {Biometrika},
  number = {1/2},
  pages = {65--82},
  publisher = {[Oxford University Press, Biometrika Trust]},
  title = {On the Compound Multinomial Distribution, the Multivariate $\beta$-Distribution, and Correlations Among Proportions},
  urldate = {2026-07-17},
  volume = {49},
  year = {1962}
}

@article{zhang2012best,
  title={The best rank-1 approximation of a symmetric tensor and related spherical optimization problems},
  author={Zhang, Xinzhen and Ling, Chen and Qi, Liqun},
  journal={SIAM Journal on Matrix Analysis and Applications},
  volume={33},
  number={3},
  pages={806--821},
  year={2012},
  publisher={SIAM}
}

@article{anandkumar2015spectral,
  title={A Spectral Algorithm for Latent {D}irichlet Allocation},
  author={Anandkumar, Anima and Foster, Dean P and Hsu, Daniel and Kakade, Sham M and Liu, Yi-Kai},
  journal={Algorithmica},
  volume={72},
  number={1},
  pages={193--214},
  year={2015},
  publisher={Springer-Verlag Berlin, Heidelberg}
}

@inproceedings{hsu2013learning,
  author = {Hsu, Daniel and Kakade, Sham M.},
  title = {Learning mixtures of spherical gaussians: moment methods and spectral decompositions},
  year = {2013},
  isbn = {9781450318594},
  publisher = {Association for Computing Machinery},
  address = {New York, NY, USA},
  url = {https://doi.org/10.1145/2422436.2422439},
  doi = {10.1145/2422436.2422439},
booktitle = {Proceedings of the 4th Conference on Innovations in Theoretical Computer Science},
  pages = {11–20},
  numpages = {10},
  location = {Berkeley, California, USA},
  series = {ITCS '13}
}

@article{tran2025davis,
  title={Davis-Kahan Theorem under a moderate gap condition},
  author={Tran, Phuc and Vu, Van},
  journal={arXiv preprint arXiv:2510.22393},
  year={2025}
}

@article{dieng2020topic,
  author = {Dieng, Adji B. and Ruiz, Francisco J. R. and Blei, David M.},
  title = {Topic Modeling in Embedding Spaces},
  journal = {Transactions of the Association for Computational Linguistics},
  volume = {8},
  pages = {439-453},
  year = {2020},
  month = {07},
  issn = {2307-387X},
  doi = {10.1162/tacl_a_00325},
  url = {https://doi.org/10.1162/tacl_a_00325},
  eprint = {https://direct.mit.edu/tacl/article-pdf/doi/10.1162/tacl_a_00325/1923074/tacl_a_00325.pdf}
}

@article{ke2024using,
  author = {Ke, Zheng Tracy and Wang, Minzhe},
  title = {Using {SVD} for Topic Modeling},
  journal = {Journal of the American Statistical Association},
  volume = {119},
  number = {545},
  pages = {434--449},
  year = {2024},
  doi = {10.1080/01621459.2022.2123813}
}

\end{document}